\PassOptionsToPackage{unicode}{hyperref}
\PassOptionsToPackage{hyphens}{url}
\PassOptionsToPackage{dvipsnames,svgnames,x11names}{xcolor}
\documentclass[
12pt]{article}
\usepackage{lineno}
\usepackage{amsmath,amssymb}
\usepackage{mathtools}
\usepackage{mathrsfs}
\DeclareMathOperator*{\argmin}{arg\,min}

\usepackage{amsthm}
\usepackage{algorithm}
\usepackage[noend]{algpseudocode}
\theoremstyle{plain}
\newtheorem{theorem}{Theorem}      
\newtheorem{proposition}[theorem]{Proposition}

\theoremstyle{definition}

\theoremstyle{plain}
\newtheorem{remark}{Remark}

\usepackage{pdflscape}
\usepackage{placeins}
\usepackage{fancyvrb}
\usepackage{subcaption} 
\usepackage{multirow}
\usepackage{dsfont}
\usepackage{bm}
\usepackage{xcolor}

\usepackage{algcompatible}
\usepackage{algorithm}
\usepackage{rotating}
\usepackage{pdflscape}

\usepackage{adjustbox}
\usepackage{graphicx}
\usepackage{amsmath}
\usepackage{amssymb}
\usepackage{multirow}
\usepackage{setspace}

\usepackage{iftex}
\ifPDFTeX
\usepackage[T1]{fontenc}
\usepackage[utf8]{inputenc}
\usepackage{textcomp} 
\else 
\usepackage{unicode-math}
\defaultfontfeatures{Scale=MatchLowercase}
\defaultfontfeatures[\rmfamily]{Ligatures=TeX,Scale=1}
\fi
\usepackage{lmodern}
\ifPDFTeX\else  
\fi
\IfFileExists{upquote.sty}{\usepackage{upquote}}{}
\IfFileExists{microtype.sty}{
	\usepackage[]{microtype}
	\UseMicrotypeSet[protrusion]{basicmath} 
}{}
\makeatletter
\@ifundefined{KOMAClassName}{
	\IfFileExists{parskip.sty}{%
		\usepackage{parskip}
	}{
		\setlength{\parindent}{0pt}
		\setlength{\parskip}{6pt plus 2pt minus 1pt}}
}{
	\KOMAoptions{parskip=half}}
\makeatother
\usepackage{xcolor}
\makeatletter
\ifx\paragraph\undefined\else
\let\oldparagraph\paragraph
\renewcommand{\paragraph}{
	\@ifstar
	\xxxParagraphStar
	\xxxParagraphNoStar
}
\newcommand{\xxxParagraphStar}[1]{\oldparagraph*{#1}\mbox{}}
\newcommand{\xxxParagraphNoStar}[1]{\oldparagraph{#1}\mbox{}}
\fi
\ifx\subparagraph\undefined\else
\let\oldsubparagraph\subparagraph
\renewcommand{\subparagraph}{
	\@ifstar
	\xxxSubParagraphStar
	\xxxSubParagraphNoStar
}
\newcommand{\xxxSubParagraphStar}[1]{\oldsubparagraph*{#1}\mbox{}}
\newcommand{\xxxSubParagraphNoStar}[1]{\oldsubparagraph{#1}\mbox{}}
\fi
\makeatother

\usepackage{longtable,booktabs,array}
\usepackage{calc} 
\usepackage{etoolbox}
\makeatletter
\patchcmd\longtable{\par}{\if@noskipsec\mbox{}\fi\par}{}{}
\makeatother
\IfFileExists{footnotehyper.sty}{\usepackage{footnotehyper}}{\usepackage{footnote}}
\makesavenoteenv{longtable}
\makeatletter
\def\maxwidth{\ifdim\Gin@nat@width>\linewidth\linewidth\else\Gin@nat@width\fi}
\def\maxheight{\ifdim\Gin@nat@height>\textheight\textheight\else\Gin@nat@height\fi}
\makeatother
\setkeys{Gin}{width=\maxwidth,height=\maxheight,keepaspectratio}
\makeatletter
\def\fps@figure{htbp}
\makeatother

\makeatletter
\@ifpackageloaded{caption}{}{\usepackage{caption}}
\AtBeginDocument{%
	\ifdefined\contentsname
	\renewcommand*\contentsname{Table of contents}
	\else
	\newcommand\contentsname{Table of contents}
	\fi
	\ifdefined\listfigurename
	\renewcommand*\listfigurename{List of Figures}
	\else
	\newcommand\listfigurename{List of Figures}
	\fi
	\ifdefined\listtablename
	\renewcommand*\listtablename{List of Tables}
	\else
	\newcommand\listtablename{List of Tables}
	\fi
	\ifdefined\figurename
	\renewcommand*\figurename{Figure}
	\else
	\newcommand\figurename{Figure}
	\fi
	\ifdefined\tablename
	\renewcommand*\tablename{Table}
	\else
	\newcommand\tablename{Table}
	\fi
}
\@ifpackageloaded{float}{}{\usepackage{float}}
\floatstyle{ruled}
\@ifundefined{c@chapter}{\newfloat{codelisting}{h}{lop}}{\newfloat{codelisting}{h}{lop}[chapter]}
\floatname{codelisting}{Listing}

\makeatother
\makeatletter
\@ifpackageloaded{caption}{}{\usepackage{caption}}
\@ifpackageloaded{subcaption}{}{\usepackage{subcaption}}
\makeatother

\ifLuaTeX
\usepackage{selnolig}  
\fi
\usepackage[]{natbib}
\usepackage{bookmark}

\IfFileExists{xurl.sty}{\usepackage{xurl}}{} 
\hypersetup{
	pdftitle={Title},
	pdfauthor={Author 1; Author 2},
	pdfkeywords={3 to 6 keywords, that do not appear in the title},
	colorlinks=true,
	linkcolor={blue},
	filecolor={Maroon},
	citecolor={Blue},
	urlcolor={Blue},
	pdfcreator={LaTeX via pandoc}}

\newcommand{\anon}{1}

\usepackage{rotating}
\begin{document}

	\def\spacingset#1{\renewcommand{\baselinestretch}%
		{#1}\small\normalsize} \spacingset{1}


	\if1\anon
	{
		\title{\bf BLOC: A Global Optimization Framework for Sparse Covariance Estimation with Non-Convex Penalties}
		\author{Priyam Das
			\hspace{.2cm}\\
			{\normalsize \medskip Department of Biostatistics, Virginia Commonwealth University}\\
			Trambak Banerjee
			\hspace{.2cm}\\
			{\normalsize \medskip Department of Analytics, Information and Operations, University of Kansas}\\
			and\\
			Prajamitra Bhuyan
			\hspace{.2cm}\\
			{\normalsize \medskip Department of Operations Management, Indian Institute of Management, India}
		}
		\maketitle
	} \fi
	
	\if0\anon
	{
		\bigskip
		\bigskip
		\bigskip
		\begin{center}
			{\LARGE\bf BLOC: A Global Optimization Framework for Sparse Covariance Estimation with Non-Convex Penalties}
		\end{center}
		\medskip
	} \fi
	
	\begin{abstract}
		We introduce \textsc{BLOC} (Black-box Optimization over Correlation matrices), a general framework for sparse covariance estimation with non-convex penalties. BLOC operates on the manifold of correlation matrices and reparameterizes it via an angular Cholesky mapping, transforming the positive-definite, unit-diagonal constraint into an unconstrained search over a Euclidean parameter space. This enables gradient-free optimization of diverse objectives, including non-differentiable or black-box losses, using a pattern search routine with adaptive coordinate polling, run-wise restarts for broader exploration, and leveraging up to $d(d-1)$ parallel threads when optimizing a $d$-dimensional correlation matrix. The method is penalty-agnostic and ensures that every iterate is a valid correlation matrix, from which covariance estimates are obtained. We establish convergence guarantees, including stationarity, global reachability under randomized restarts, and convergence-rate characterization under smooth convex losses. From a statistical perspective, we prove consistency, convergence rates, and sparsistency for penalized correlation estimators under general conditions, extending sparse covariance theory beyond the Gaussian setting. Empirically, BLOC is competitive with specialized estimators under smooth losses and performs strongly for a discontinuous matrix-depth criterion, while preserving positive definiteness and favorable support recovery. A parallel implementation enhances scalability, and a pan-gynecologic cancer proteomic network application demonstrates flexible, positive-definite sparse covariance estimation.
	\end{abstract}
	
	\noindent%
	\vfill
	
	\newpage
	\spacingset{1.69} 

	\section{Introduction}\vspace{-0.5cm}
	The population covariance matrix $\bm \Sigma_0$ of a random vector is a fundamental object in multivariate statistics, and its estimation plays a central role in a wide range of scientific applications. Reliable estimation of $\bm \Sigma_0$ underlies problems in financial factor analysis \citep{fan2011high}, high-dimensional compositional data analysis \citep{li2023robust}, 
	and functional genomics and transcriptome analysis \citep{schafer2005shrinkage}, among many others. As the data dimension grows, however, covariance estimation becomes increasingly challenging: the number of free parameters in $\bm \Sigma_0$ grows quadratically with dimension, rendering classical estimators unstable or ill-conditioned in modern high-dimensional regimes. Consequently, a large body of literature has focused on regularized covariance estimation strategies that reduce effective model complexity, including eigenvalue regularization, shrinkage methods, and sparsity-inducing approaches; see \citet{bien2019graph} for a comprehensive review.
	
	Among these strategies, estimating $\bm \Sigma_0$ under sparsity assumptions, where only a small fraction of off-diagonal entries are nonzero, has emerged as the dominant paradigm in high-dimensional covariance estimation. Early work includes banding and tapering estimators \citep{bickel2008regularized}, thresholding-based procedures \citep{cai2011adaptive}, and positive-definite modifications thereof \citep{cai2012minimax}. Penalized likelihood formulations further unified these ideas by framing covariance estimation as an optimization problem with explicit sparsity-inducing penalties on off-diagonal elements.
	
	To enforce sparsity, the $\ell_1$ penalty is by far the most widely used regularizer, owing to its simplicity, interpretability, and the availability of scalable optimization routines in related settings. However, it is well known that $\ell_1$ penalization induces systematic shrinkage bias in nonzero estimates, particularly when true correlations are moderate or large \citep{fan2001variable}. To mitigate this bias, non-convex penalties such as the smoothly clipped absolute deviation (SCAD) penalty \citep{fan2001variable} and the minimax concave penalty (MCP) \citep{zhang2010nearly} have been proposed and studied extensively. In the context of covariance and correlation estimation, \citet{lam2009sparsistency} established that non-convex penalization can achieve both consistency and sparsistency under suitable conditions, while reducing estimation bias relative to $\ell_1$-based methods. More recent work has further extended these theoretical guarantees to broader classes of losses and penalties \citep{wang2024estimation}.

	Despite these strong statistical guarantees, existing computational approaches for covariance estimation with non-convex penalties suffer from several important limitations. First, many algorithms are tightly coupled to specific penalty functions \citep{wang2024estimation} or likelihood-based objectives \citep{Bien2011}, making them difficult to adapt when applications require alternative penalties or user-defined loss functions. Second, methods based on convex relaxations or majorization--minimization schemes typically rely on objective-specific constructions, such as Gaussian likelihoods \citep{Bien2011} or Frobenius-norm \citep{Wen2021} fitting criteria, and it is unclear how these techniques extend to more general, possibly non-differentiable or black-box objectives. Finally, most existing algorithms for covariance estimation under non-convex penalties or parameterizations are inherently local optimization procedures \citep{Rothman2009, Liu2014, Wen2016}. Once trapped in poor local minima, they lack principled mechanisms for global exploration, a limitation that becomes increasingly severe as dimensionality increases.

	In this paper, we propose a general framework for sparse covariance estimation with non-convex penalties, termed BLOC (\emph{Black-box Optimization over Correlation Matrices}). Rather than developing algorithms tailored to specific penalties or objective functions, BLOC formulates covariance estimation as a black-box optimization problem over the space of correlation matrices, requiring only that the objective function be evaluable.
	Our approach exploits two key ideas. First, instead of optimizing directly over covariance matrices, BLOC operates on the corresponding correlation matrix. This strategy is advantageous because the correlation matrix shares the sparsity structure of $\bm \Sigma_0$ while typically being easier to estimate and numerically more stable \citep{lam2009sparsistency}. We leverage the geometric structure of the correlation matrix space: full-rank correlation matrices form a smooth nonlinear manifold characterized by positive definiteness and unit diagonal constraints. By employing an angular Cholesky reparameterization, BLOC provides a
	bijective canonical representation of the correlation-matrix space and enables unconstrained optimization in a Euclidean parameter space, ensuring that every iterate corresponds to a valid correlation matrix. Second, BLOC couples this reparameterization with a derivative-free global optimization strategy based on a recursive modified pattern search \citep{Torczon1997, Das2023MsiCOR, Das2023RMPSH, kim2025smart}. In contrast to gradient-based or convex-relaxation methods, this approach explicitly balances local refinement with global exploration through coordinate-wise polling, adaptive step-size control, and systematic restart mechanisms. These features make BLOC particularly useful when the loss is nonsmooth or black-box, or when a problem-specific gradient/MM construction is unavailable; when highly specialized solvers exist, BLOC may trade computational speed for generality. As a result, BLOC enjoys several distinct advantages over existing methods. It is objective-agnostic and penalty-agnostic, providing a unified framework for covariance estimation under a wide class of non-convex penalties without requiring penalty-specific algorithmic derivations. 
	Our theoretical analyses provide guarantees on stationarity, global reachability, convergence in probability under mild assumptions, and sublinear convergence rates for smooth convex objectives. 

	The remainder of the article is organized as follows. In Section~\ref{sec:prelims}, we describe the penalized correlation-estimation setup. Section~\ref{sec:bloc} presents the BLOC framework in detail, while Section~\ref{sec:theory} develops its theoretical properties. Benchmark comparisons and simulation studies are reported in Section~\ref{sec:sims}, respectively. In Section~\ref{sec:realdata}, we apply BLOC to pathway-informed correlation estimation in a pan-gynecologic proteomics study. The article concludes with a discussion in Section~\ref{sec:discuss}, with proofs and technical details deferred to the Supplementary Material.
	\vspace{-0.99cm}
	\section{Penalized Correlation Estimation}
	\label{sec:prelims}\vspace{-0.5cm}
	\subsection{Statistical model and correlation-based formulation}\vspace{-0.2cm}
	Let $\bm X \in \mathbb R^d$ be a random vector with positive definite covariance matrix $\bm\Sigma_0$. Given $n$ i.i.d. observations $\bm X_1,\ldots,\bm X_n$, our objective is to estimate $\bm\Sigma_0$ under the assumption that it is sparse. Following a common strategy in high-dimensional covariance estimation \citep{lam2009sparsistency}, we focus on estimation of the corresponding correlation matrix $\bm\Gamma_0$, which preserves the sparsity pattern of $\bm\Sigma_0$ while having a fixed and known diagonal structure. An estimator of $\bm\Sigma_0$ can then be recovered as $\hat{\bm\Sigma}_n = \hat{\bm W}\hat{\bm\Gamma}_n\hat{\bm W}$, where $\hat{\bm W}^2=\text{diag}(\bm S)$ and $\bm S$ is the sample covariance matrix. Let $\mathcal C_d = \{\bm C \in \mathbb R^{d\times d} : \bm C \succ 0, \bm C = \bm C^\top, \mathrm{diag}(\bm C)=1\}$ denote the space of valid correlation matrices. Our goal is to construct an estimator $\hat{\bm\Gamma}_n \in \mathcal C_d$ that is both statistically accurate and sparse.
	\vspace{-0.8cm}
	\subsection{General penalized estimation framework}\vspace{-0.45cm}
	We consider a general class of penalized estimators defined as solutions to
	\vspace{-0.5cm}
	\begin{align}
		\label{eq:problem 1}
		\min_{\bm\Gamma \in \mathcal C_d}
		\; h_n(\bm\Gamma) + \sum_{i\neq j} p_{\lambda_n}(|\gamma_{ij}|), 
	\end{align}\vspace{-1.4cm}
	
	where $h_n:\mathcal C_d\to\mathbb R$ may represent a likelihood, a non-convex loss such as the capped-$\ell_1$ loss \citep{zhang2010analysis}, or any differentiable, non-differentiable, or even potentially discontinuous criterion and $p_{\lambda_n}$ is a general non-convex coordinate-wise penalty, such as SCAD, MCP, applied only to the off-diagonal entries of $\bm \Gamma$. Typical examples of $h_n$ include the multivariate Gaussian likelihood \citep{lam2009sparsistency}, matrix depth \citep{chen2018robust}, the minimum regularized covariance determinant \citep{boudt2020minimum} and various matrix norms (see for instance \cite{bien2019graph,CaiZhou2012} and the references therein) among others. Since Problem \eqref{eq:problem 1} is non-convex, existing solution strategies generally rely on problem-specific convex relaxations \citep{wang2024estimation,Wen2021} or majorization schemes \citep{Bien2011}. In contrast, we adopt a different perspective and develop BLOC, a global optimization framework that is fully agnostic to the choice of objective or penalty, requiring only that they can be evaluated. Directly optimizing over the correlation-matrix space $\mathcal C_d$ is, nevertheless, challenging due to the inherent positive definiteness and unit diagonal constraints. BLOC overcomes these difficulties by re-parametrizing \eqref{eq:problem 1} as an unconstrained problem in a set of angular coordinates and performing a global search via a novel global optimization routine. The theoretical properties of the resulting estimators are developed in Section~\ref{sec:theory}.
	\vspace{-0.6cm}
	\subsection{Notation and asymptotic conventions}\vspace{-0.4cm}
	Throughout, $\gamma_{ij}$ denotes the $(i,j)$ entry of $\bm\Gamma$. We write $\|\cdot\|_F$ and $\|\cdot\|$ for the Frobenius and operator norms, respectively, and use standard asymptotic notation. Additional technical definitions and regularity conditions are deferred to the Supplementary Section E.
	\vspace{-0.9cm}
	\section{BLOC}\label{sec:bloc}
	\vspace{-0.5cm}
	We develop BLOC as a computational framework for evaluating penalized correlation matrix estimators of the form introduced in Section~\ref{sec:prelims}. The central challenge arises from the constrained nature of the correlation matrix space, which requires symmetry, positive definiteness, and unit diagonal elements to be satisfied simultaneously. Direct optimization over this space is nontrivial, particularly when the objective function is nonconvex, nonsmooth, or available only through black-box evaluation.
	
	BLOC addresses these challenges through a reparameterization that enforces all structural constraints by construction, coupled with a derivative-free optimization strategy operating in an unconstrained Euclidean space. This design ensures that every iterate corresponds to a valid correlation matrix while remaining flexible enough to accommodate a broad class of objective functions, including nonconvex penalties and robust loss functions. The proposed framework consists of three main components:  
	(i) a bijective reparameterization of the correlation matrix space via angular coordinates;  
	(ii) an unconstrained reformulation of the penalized estimation problem; and  
	(iii) a robust coordinate-wise optimization strategy for computing the resulting estimator.  
	We describe each component in turn.
	\vspace{-0.7cm}
	\subsection{Correlation-matrix reparameterization}
	\label{subsec:reparam}\vspace{-0.4cm}
	
	Let \( \bm{C} \in \mathcal{C}_d \) denote a \( d \times d \) correlation matrix, i.e., a symmetric positive definite matrix with unit diagonal. A convenient and widely used representation of \( \bm{C} \) is obtained through its Cholesky decomposition,
	$\bm{C} = \bm{L}\bm{L}^\top$, where \( \bm{L} \) is a lower triangular matrix with strictly positive diagonal entries. When \( \bm{C} \in \mathcal{C}_d \), the additional unit-diagonal constraint implies that each row of \( \bm{L} \) has unit Euclidean norm.
	\vspace{-0.2cm}
	\begin{proposition}[Correlation matrix bijection]
		\label{thm:corr}
		A matrix \( \bm{C} \in \mathcal{C}_d \) if and only if there exists a unique lower triangular matrix \( \bm{L} \) such that
		\(
		\bm{C} = \bm{L}\bm{L}^\top,
		\)
		where \( \bm{L} \) is full rank, has strictly positive diagonal entries, and satisfies
		\(
		\sum_{j=1}^m l_{mj}^2 = 1
		\)
		for each row \( m = 1,\dots,d \).
	\end{proposition}
	\vspace{-0.1cm}
	This representation transforms the positive definiteness and unit-diagonal constraints on \( \bm{C} \) into geometric constraints on the rows of \( \bm{L} \). Specifically, the \( m \)-th row of \( \bm{L} \) lies on the positive hemisphere of the unit sphere in \( \mathbb R^m \). This observation motivates an angular reparameterization of \( \bm{L} \), which converts the constrained matrix space \( \mathcal C_d \) into a product of canonical angular domains. Let \( \bm{L} = (l_{mg}) \) denote the Cholesky factor. The angular construction proceeds recursively as follows.\vspace{-0.25cm}
	\begin{itemize}
		\item Set \( l_{11} = 1 \).
		\vspace{-0.25cm}
		\item For \( m = 2 \), represent \( (l_{21}, l_{22}) \) as a point on the upper half of the unit circle:
		\vspace{-0.8cm}
		\[
		l_{21} = \sin \omega_{21}, \qquad
		l_{22} = \cos \omega_{21}, \qquad
		\omega_{21} \in (-\pi/2, \pi/2).
		\]
		
		\vspace{-0.8cm}
		\item For \( m \ge 3 \), represent the row
		\( \bm l_m = (l_{m1},\ldots,l_{mm}) \)
		as a point on the positive hemisphere of the unit \( (m-1) \)-sphere using hyperspherical coordinates
		\vspace{-0.4cm}
		\[
		\Omega_m=(\omega_{m1},\ldots,\omega_{m,m-1})\in\Theta_m,
		\qquad
		\Theta_m\subset[0,\pi/2)\times[0,\pi]^{m-3}\times[0,2\pi),
		\]
		
		\vspace{-0.8cm}
		where $\Theta_m$ denotes the subset of the displayed product satisfying the canonical spherical-coordinate convention that, if $\prod_{j=1}^{k}\sin(\omega_{mj})=0$ for some $k\le m-2$, then $\omega_{m,k+1}=\cdots=\omega_{m,m-1}=0$. Here,
		\vspace{-0.5cm}
		\begin{align*}
			&l_{m1} = \Big( \prod_{i=1}^{m-2} \sin \omega_{mi} \Big)\sin \omega_{m,m-1},\;l_{m2} = \Big( \prod_{i=1}^{m-2} \sin \omega_{mi} \Big)\cos \omega_{m,m-1},\\
			&l_{m3} = \Big( \prod_{i=1}^{m-3} \sin \omega_{mi} \Big)\cos \omega_{m,m-2},\; l_{m4} = \Big( \prod_{i=1}^{m-4} \sin \omega_{mi} \Big)\cos \omega_{m,m-3},\\
			&\vdots\\
			&l_{m,m-1} = \sin \omega_{m1}\cos \omega_{m2},\; l_{mm} = \cos \omega_{m1}.
		\end{align*}
		\vspace{-1.4cm}
	\end{itemize}
	\vspace{-0.4cm}
	
	This construction yields a one-to-one correspondence between correlation matrices and collections of angular parameters.
	\vspace{-0.2cm}
	\begin{proposition}[Angular bijection]
		\label{thm:bijection}
		Define the angular domain $\mathcal A_d=(-\pi/2,\pi/2)\times\prod_{m=3}^d\Theta_m$. The mapping $\Phi:\mathcal A_d\to\mathcal C_d$, given by $\bm\omega\mapsto\bm L(\bm\omega)\bm L(\bm\omega)^\top$, is bijective.
	\end{proposition}
	While related arguments underlying Propositions \ref{thm:corr} and \ref{thm:bijection} may appear implicitly in the literature \citep{Lewandowski2009}, we have not found these results stated in the present form; for completeness, detailed proofs are provided in the Supplementary Section E.
	\vspace{-0.6cm}
	\subsection{Transformation to an unconstrained parameter space}\vspace{-0.3cm}
	The angular representation in Section~\ref{subsec:reparam} expresses correlation matrices in terms of parameters constrained to the structured domain \( \mathcal A_d \). For numerical optimization, we use the corresponding non-canonical angular cover
	\vspace{-0.5cm}
	\[
	\widetilde{\mathcal A}_d
	=
	(-\pi/2,\pi/2)
	\times
	\prod_{m=3}^d
	\left([0,\pi/2)\times[0,\pi]^{m-3}\times[0,2\pi)\right),
	\]
	
	\vspace{-0.8cm}
	which contains $\mathcal A_d$. The angular-Cholesky map extends naturally to
	$\widetilde{\mathcal A}_d$ (and is still denoted by $\Phi$), generating the same
	collection of correlation matrices, although boundary representations need not be unique. To enable unconstrained numerical optimization, BLOC maps Euclidean space to this angular cover. Let $N = d(d-1)/2$ denote the total number of angular parameters. Let
	\(
	\boldsymbol\theta = \{\theta_{m,k}\}
	\)
	denote an unconstrained collection of real-valued parameters indexed according to the angular coordinates. We define a wrapping map
	$
	\mathscr M : \mathbb R^N \to \widetilde{\mathcal A}_d
	$
	that enforces the angular constraints through periodic folding and reflection. Component-wise, the map is defined as
	$
	\theta^{\mathrm{mapped}}_{m,k}
	=
	\mathscr M^{m,k}(\theta_{m,k})$, with
	\vspace{-0.5cm}
	\begin{align*}
		\theta_{m,k}^{\mathrm{mapped}}
		=
		\begin{cases}
			(\theta_{m,k} + \pi/2)\bmod \pi - \pi/2,
			& m=2,\;k=1,\\[4pt]
			\pi/2 - |(\theta_{m,k}\bmod \pi) - \pi/2|,
			& m\ge3,\;k=1,\\[4pt]
			\pi - |(\theta_{m,k}\bmod 2\pi) - \pi|,
			& m\ge3,\;2\le k\le m-2,\\[4pt]
			\theta_{m,k}\bmod 2\pi,
			& m\ge3,\;k=m-1.
		\end{cases}
	\end{align*}
	
	\vspace{-0.7cm}
	At folding values for which the first mapped angular coordinate would equal
	an excluded endpoint ($-\pi/2$ for $m=2$ or $\pi/2$ for $m\ge3$), we set
	that coordinate to $0$. This convention ensures that all diagonal entries
	of $\bm L$ remain strictly positive.
	Thus, arbitrary perturbations in \( \mathbb R^N \) are mapped to admissible angles in
	\( \widetilde{\mathcal A}_d \), and hence, through $\Phi$, to valid correlation matrices.
	
	For notational convenience, we vectorize the parameters into
	$\boldsymbol\varphi = (\varphi_1,\ldots,\varphi_N)\in\mathbb R^N$, using a fixed one-to-one correspondence between indices \( n \) and pairs \( (m(n),k(n)) \). Under this identification,
	\vspace{-0.3cm}
	\[
	\mathscr M(\boldsymbol\varphi)
	=
	\big(\mathscr M^{m(n),k(n)}(\varphi_n)\big)_{n=1}^N.
	\]
	
	\vspace{-0.6cm}
	Combining the wrapping map with the angular-Cholesky mapping yields the full transformation
	$\boldsymbol\varphi
	\;\xrightarrow{\;\mathscr M\;}\;
	\bm\omega \in \widetilde{\mathcal A}_d
	\;\xrightarrow{\;\Phi\;}\;
	\bm C \in \mathcal C_d$.
	The middle and bottom panels of Figure~\ref{fermi} illustrate this transformation and its role within BLOC.
	\begin{figure}[!h]
		\centering
		\includegraphics[width=0.99\textwidth]{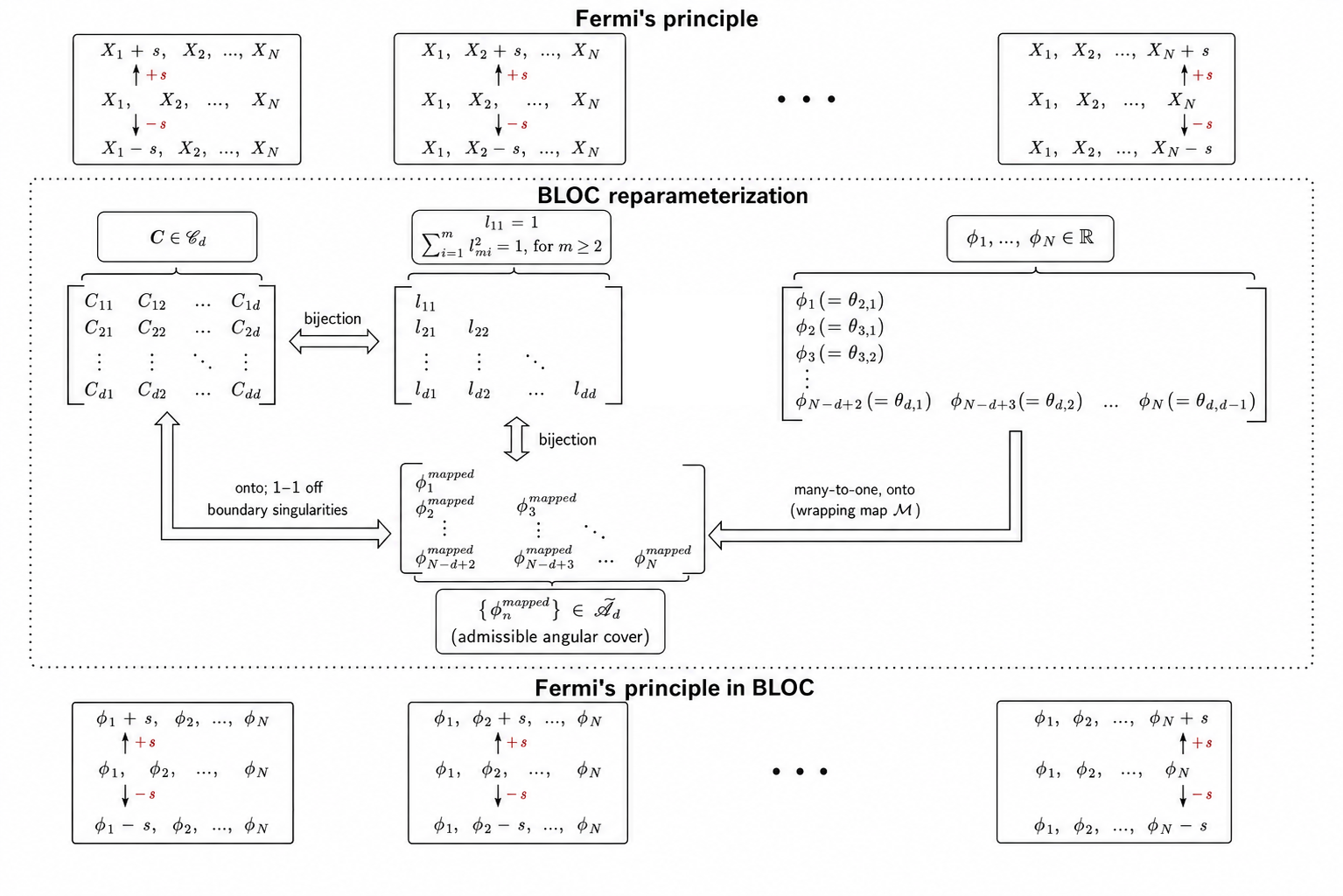}
		\caption{Illustration of Fermi’s principle and its implementation within BLOC.
			\textbf{Top:} Classical Fermi’s principle in $\mathbb{R}^N$, where $2N$ coordinate-wise perturbations are generated using a fixed step-size $s$. 
			\textbf{Middle:} BLOC reparameterization, mapping a correlation matrix $\bm{C} \in \mathcal{C}_d$ to an unconstrained vector $\boldsymbol\varphi \in \mathbb{R}^N$ via Cholesky and angular representations, with a wrapping map $\mathscr{M}$ mapping into the admissible angular cover $\widetilde{\mathcal{A}}_d$. \textbf{Bottom:} Fermi’s principle applied in the transformed space, where coordinate-wise perturbations in $\mathbb{R}^N$ are wrapped into $\widetilde{\mathcal{A}}_d$ and mapped to valid correlation matrices.}
		\label{fermi}\vspace{-0.4cm}
	\end{figure}
	\vspace{-0.8cm}
	\subsection{Optimization strategy}
	\label{subsec:opt_strategy}\vspace{-0.3cm}
	After reparameterization, estimation of the correlation matrix reduces to optimizing a potentially nonconvex and nonsmooth objective over an unconstrained Euclidean parameter space. In this setting, gradient-based methods are often unreliable or infeasible, either because derivatives are unavailable or because the objective surface contains multiple local minima induced by nonconvex penalties or robust loss functions. 
	To address these challenges, BLOC adopts a derivative-free, coordinate-wise search strategy inspired by Fermi’s principle \citep{Fermi1952}. At each iteration, the objective is evaluated at a finite set of candidate points obtained by perturbing individual coordinates in positive and negative directions using a common step size. The best-performing candidate is selected, and the step size is adaptively reduced when no improvement is observed, allowing the search to transition naturally from global exploration to local refinement. This strategy is closely related to classical pattern search and direct search methods \citep{Torczon1997}, which are well suited for derivative-free optimization and admit convergence guarantees under mild conditions. To improve robustness in nonconvex landscapes, BLOC builds on the Recursive Modified Pattern Search (RMPS) framework \citep{Das2023RMPSH, Das2023MsiCOR, kim2025smart, Das2022}, which incorporates adaptive step-size control and a restart mechanism. These features enable escape from unfavorable local stationary points while retaining computational scalability. Further discussion of related global optimization literature is provided in Supplementary Section~A.
	In the present setting, RMPS is applied to the unconstrained angular parameterization of the correlation matrix space, yielding an optimization strategy that respects the induced geometry and accommodates general nonconvex and nonsmooth objectives. This procedure forms the computational backbone of BLOC.
	\vspace{-0.5cm}
	\subsection{BLOC algorithm}
	\label{subsec:bloc_algorithm}\vspace{-0.3cm}
	We now describe the algorithmic procedure underlying BLOC for computing penalized correlation matrix estimators. For convenience, we rewrite the penalized estimation problem in \eqref{eq:problem 1} as
	\vspace{-0.3cm}
	\begin{equation}
		\label{eq:problem 2}
		\min_{\bm \Gamma \in \mathcal C_d} g_{\lambda_n}(\bm \Gamma),
	\end{equation}
	where \( g_{\lambda_n}(\bm \Gamma) = h_n(\bm \Gamma) + \sum_{i \neq j} p_{\lambda_n}(|\gamma_{ij}|) \).
	As discussed in the previous subsections, BLOC transforms the constrained optimization problem \eqref{eq:problem 2}, defined over the manifold of full-rank correlation matrices \( \mathcal C_d \), into the unconstrained problem
	\vspace{-0.3cm}
	\begin{equation}
		\label{eq:problem 3}
		\min_{\boldsymbol{\varphi} \in \mathbb{R}^N}
		f_{\lambda_n}(\boldsymbol{\varphi})
		:=
		g_{\lambda_n}\!\big( \Phi(\mathscr{M}(\boldsymbol{\varphi})) \big),
	\end{equation}
	
	\vspace{-0.5cm}
	where \( \boldsymbol{\varphi} \in \mathbb{R}^N \) denotes the vector of unconstrained angular parameters, \( \mathscr{M} : \mathbb{R}^N \to \widetilde{\mathcal A}_d \) enforces the admissible angular cover, and \( \Phi : \widetilde{\mathcal A}_d \to \mathcal C_d \) maps angles to correlation matrices via the angular-Cholesky representation. To solve \eqref{eq:problem 3}, BLOC adapts a recursive coordinate-wise pattern search procedure, namely Recursive Modified Pattern Search (RMPS), to the angular parameterization of the correlation matrix space. The resulting algorithm is fully derivative-free and requires only pointwise evaluation of the objective function, making it applicable even when \( g_{\lambda_n} \) is non-differentiable or available through a black-box routine.
	
	The BLOC algorithm proceeds over multiple \emph{runs}. Within each \emph{run}, the method performs a sequence of coordinate-wise searches along the directions \( \pm \bm e_i \) in the angular space, using a global step size that is adaptively reduced over iterations. At each iteration, \( 2N \) candidate perturbations are evaluated, the best improving direction (if any) is selected, and the current iterate is updated accordingly. If no improvement is achieved, the step size is reduced. Between \emph{runs}, the algorithm may restart from the best solution obtained so far, which allows the search to escape unfavorable local stationary points. We now summarize the main components of the algorithm.
	
	\noindent \textit{Parallelization.}
	BLOC is naturally parallelizable. At each iteration within a \emph{run}, the objective is evaluated at \(2N\) coordinate-wise perturbations along \( \pm \bm e_i \), where \(N=d(d-1)/2\) is the number of angular parameters. These evaluations are independent and can be executed in parallel, allowing BLOC to exploit up to \(d(d-1)\) threads per iteration. Parallelization affects only computational efficiency and does not alter the algorithmic logic or theoretical guarantees.
	
	\noindent \textit{Initialization and transformation.}
	Each \emph{run} is initialized either from a user-specified correlation matrix or directly in the angular space. When a correlation matrix is supplied, it is mapped to its angular representation via the bijective transformation in Section~\ref{subsec:reparam}, ensuring validity of the initial iterate.
	
	\noindent \textit{Coordinate-wise search and step-size control.}
	Within a \emph{run}, BLOC performs coordinate-wise pattern search by evaluating \(2N\) candidate updates at each iteration and accepting the candidate that yields the largest decrease in the objective. If no improvement is found, the iterate remains unchanged and the global step size is geometrically reduced, enabling a transition from coarse exploration to local refinement.
	
	\noindent \textit{Termination and restarts.}
	A \emph{run} terminates when a maximum iteration budget is reached or the step size falls below a prescribed threshold. To mitigate entrapment in local minima, BLOC is executed over multiple \emph{runs}. In practice, each new \emph{run} is initialized from the best solution of the previous one while resetting the step size, a strategy that empirically improves exploration; formal reachability is studied for a randomized restart variant in Section~\ref{sec:bloc_theory}.
	
	\noindent \textit{Final solution.}
	After completion of all \emph{runs}, the angular vector attaining the lowest objective value is mapped back to the correlation matrix space using the transformation \( \Phi \circ \mathscr{M} \), yielding the final estimator.
	
	In summary, BLOC converts the constrained penalized correlation estimation problem into an unconstrained formulation and computes the resulting estimator using a derivative-free coordinate-wise search strategy with adaptive step-size control and restarts. This design ensures that all iterates correspond to valid correlation matrices while accommodating a wide class of objective functions. The overall computational flow is illustrated in Figure~\ref{BLOC_concept}, and the complete procedure is summarized in Algorithm~S1 of the Supplementary Material.
	\begin{figure}[]
		\centering
		\includegraphics[width=0.99\textwidth]{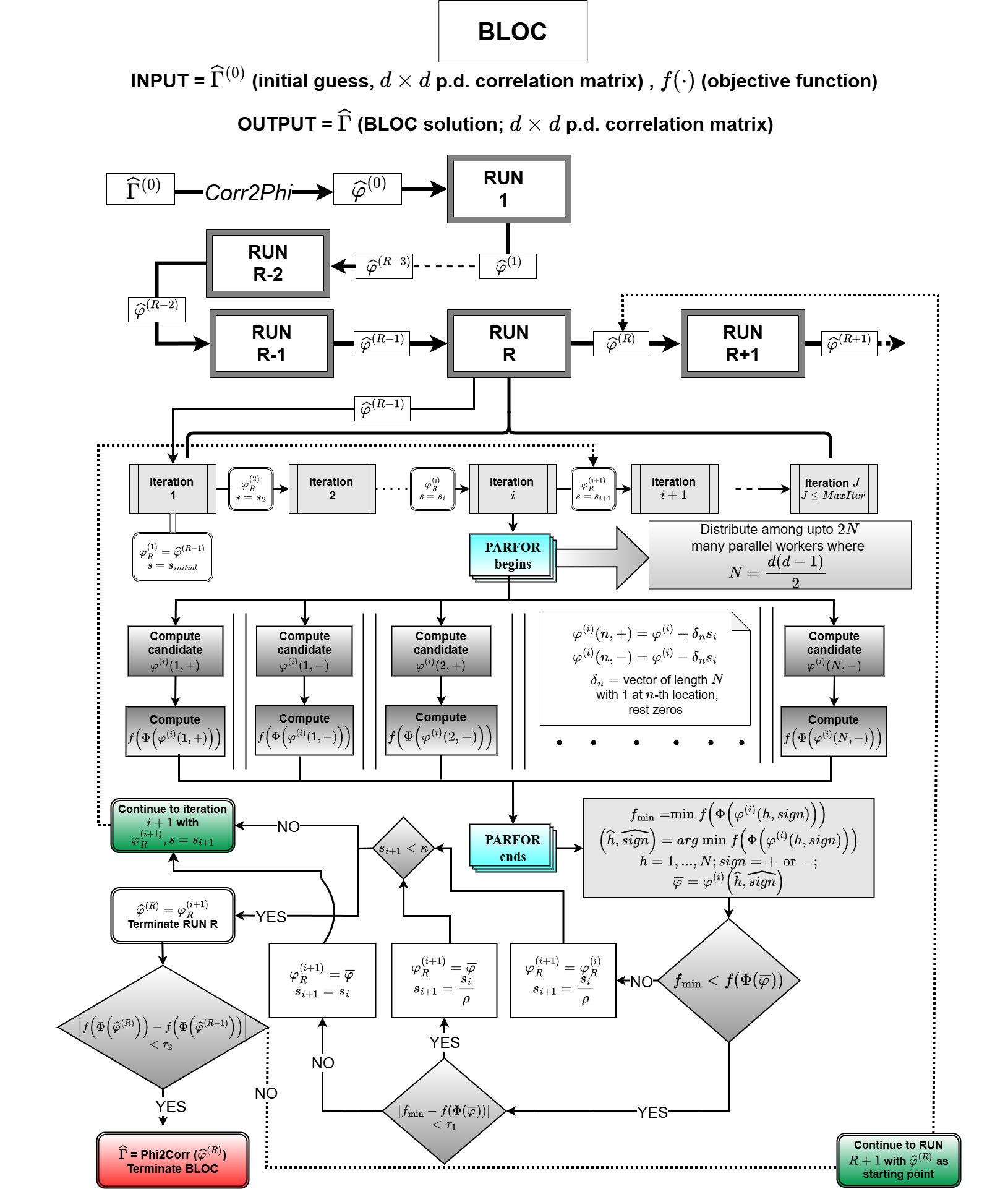}
		\caption{Flowchart of the BLOC algorithm. The method alternates between coordinate-wise pattern searches and adaptive step-size updates across multiple \emph{runs}, evaluating up to \(2N\) candidate perturbations per iteration (with \(N=d(d-1)/2\)) in parallel, and mapping the final angular solution back to a valid correlation matrix.}
		\label{BLOC_concept}
	\end{figure}
	\vspace{-1.4cm}
	\subsection{Practical considerations}\vspace{-0.4cm}
	Global reachability and convergence-in-probability guarantees are established for a randomized multi-start variant; see Supplementary Lemma S1 and Theorem S1. In practice, however, only a small number of \emph{runs} (typically 5–20) is sufficient to obtain high-quality solutions. Within each \emph{run}, coordinate-wise polling combined with progressive step-size reduction enables exploration at multiple resolutions, transitioning naturally from global search to local refinement. Rather than fully randomized restarts, BLOC initializes successive \emph{runs} from the best previous solution while resetting the step size, a strategy that empirically facilitates escape from local minima. This warm-restart design has proven effective in earlier RMPS-based methods \citep{Das2023RMPSH, Das2022, Das2021, Das2023MsiCOR, kim2025smart}, yielding a computationally efficient and numerically stable procedure. Its practical benefits are demonstrated in the benchmark studies of Section~\ref{sec:sims} and are further detailed in Supplementary Section~B.
	\vspace{-0.8cm}
	\section{Theory}
	\label{sec:theory}\vspace{-0.5cm}
	We first establish the existence of a local minimizer of Problem \eqref{eq:problem 1} and derive its convergence rate under Frobenius norm (Section \ref{sec:stat_theory}). Thereafter, we present the theoretical guarantees for BLOC (Section \ref{sec:bloc_theory}).
	\vspace{-0.6cm}
	\subsection{Statistical guarantees: rates and sparsistency}
	\label{sec:stat_theory} \vspace{-0.2cm}
	We begin by establishing the existence of a local minimizer of Problem \eqref{eq:problem 1} and deriving its Frobenius-norm convergence rate (Theorem \ref{thm:G7}). We subsequently show that the local minimizer constructed in Theorem \ref{thm:G7} recovers the correct off-diagonal support of the true correlation matrix under additional regularity conditions (Theorem \ref{thm:G8}). Let $\bm\Gamma_0=(\gamma_{0,ij}:1\le i,j\le d)$ be the true $d\times d$ correlation matrix and
	define $\mathcal I_d
	=
	\{(i,j):1\le i,j\le d,\ i\neq j\}$, $S
	=
	\{(i,j)\in\mathcal I_d:\gamma_{0,ij}\neq0\},~
	S^c
	=
	\mathcal I_d\setminus S,~
	s=|S|\ge 1$. A local minimizer of Problem \eqref{eq:problem 1} will be denoted $\hat{\bm \Gamma}_n=(\hat{\gamma}_{n,ij}:1\le i,j\le d)$. Let $\mathcal D_d
	=
	\left\{
	\bm\Gamma\in\mathbb R^{d\times d}:
	\bm\Gamma=\bm\Gamma^\top,\ 
	\operatorname{diag}(\bm\Gamma)=\mathbf 1
	\right\}$, and let $\mathcal T_d
	=
	\left\{
	\Delta\in\mathbb R^{d\times d}:
	\Delta=\Delta^\top,\ 
	\operatorname{diag}(\Delta)=0
	\right\}$ denote the associated linear perturbation space. We equip \(\mathcal T_d\) with the Frobenius inner product $\langle A,B\rangle = \operatorname{tr}(A^\top B)$. 
	Throughout, \(\partial h_n(\bm\Gamma)\) denotes the Clarke
	subdifferential \citep{clarke1990optimization} of \(h_n\) relative to the affine space
	\(\mathcal D_d\), with elements of the dual space identified with
	their unique zero-diagonal symmetric representatives in
	\(\mathcal T_d\).
	Define $q_n
	=
	\sqrt{\log d/n}$ and $r_n =\sqrt{s}\,q_n$. In our discussion, we will assume that for every fixed \(C>0\) and for all sufficiently large \(n\), \(p_{\lambda_n}\) is twice continuously differentiable on
	$\mathcal J_n(C)
	=
	\bigcup_{(i,j)\in S}
	\left[
	|\gamma_{0,ij}|-Cr_n,\,
	|\gamma_{0,ij}|+Cr_n
	\right]$. 
	
	We need the following regularity conditions.
	\begin{enumerate}
		\item[(A1)]There exist constants $0<c_\Gamma<C_\Gamma<\infty$ such that, for all sufficiently large \(n\), $
		c_\Gamma
		\le
		\lambda_{\min}(\bm\Gamma_0)
		\le
		\lambda_{\max}(\bm\Gamma_0)
		\le
		C_\Gamma$, where $\lambda_{\min}(A)$ and $\lambda_{\max}(A)$ denote the smallest and largest eigenvalues, respectively, of a symmetric matrix $A$.
		\item[(A2)]There exists \(\delta_h>0\) such that, with probability tending to one, \(h_n\) is locally Lipschitz on
		$
		\left\{
		\bm\Gamma\in\mathcal D_d:
		\|\bm\Gamma-\bm\Gamma_0\|_F<\delta_h
		\right\}$. Moreover, there exist \(\kappa>0\) and \(g_{0,n}\in\partial h_n(\bm\Gamma_0)\) such that, for every fixed \(C>0\),
		$$
		h_n(\bm\Gamma_0+\Delta)-h_n(\bm\Gamma_0)-\langle g_{0,n},\Delta\rangle
		\ge
		\left\{
		\frac{\kappa}{2}-\eta_n(C)
		\right\}\|\Delta\|_F^2
		$$ uniformly over all \(\Delta\in\mathcal T_d\) satisfying
		$
		\|\Delta\|_F\le Cr_n,~
		\bm\Gamma_0+\Delta\in\mathcal C_d,
		$ where \(\eta_n(C)=o_P(1)\).
		\item[(A3)]For the same \(g_{0,n}\) appearing in \textnormal{(A2)}, there exists
		a constant \(C_g<\infty\), independent of $n$, such that $\|g_{0,n}\|_\infty
		\le
		C_gq_n$ with probability tending to one, where $\|A\|_\infty
		=
		\max_{1\le i,j\le d}|a_{ij}|$.
		\item[(A4)]For every \(n\), \(p_{\lambda_n}:[0,\infty)\to[0,\infty)\) is continuous with \(p_{\lambda_n}(0)=0\). Moreover, there exists \(k>0\) such that, for every \(k_1\in(0,k)\), there exists \(\varepsilon_0>0\), independent of \(n\), for which $
		p_{\lambda_n}(t)\ge k_1\lambda_n t
		$ for all sufficiently large \(n\) and all \(0\le t\le \varepsilon_0\lambda_n\).
		
		\item[(A5)]Define $a_n
		=
		\max_{(i,j)\in S}
		\left|
		p_{\lambda_n}'(|\gamma_{0,ij}|)
		\right|$ and denote $
		b_n(C)
		=
		\max_{(i,j)\in S}
		\sup_{
			u>0:
			|u-|\gamma_{0,ij}||\le Cr_n
		}
		\left|
		p_{\lambda_n}''(u)
		\right|$ for every fixed \(C>0\). Then, there exists a constant \(C_a<\infty\), independent of \(n\), such
		that
		$a_n\le C_aq_n$ for all sufficiently large \(n\). Moreover,
		$
		\min_{(i,j)\in S}
		{|\gamma_{0,ij}|}/{\lambda_n}
		\longrightarrow\infty$ and $b_n(C)\to0$ for every fixed $C>0$.
	\end{enumerate}
	Our analysis generalizes the framework of \citet{lam2009sparsistency} by replacing the Gaussian negative log-likelihood with a general empirical loss \(h_n(\bm\Gamma)\). Since \(h_n\) is not required to be likelihood-based or differentiable, the local curvature and subgradient concentration properties that follow from explicit likelihood calculations in the Gaussian setting must instead be imposed directly. 
	In particular, Assumption (A2) requires local quadratic growth of the loss around \(\bm\Gamma_0\), while Assumption (A3) controls an associated Clarke subgradient at the usual high-dimensional rate. 
	Assumptions (A4) and (A5) impose, respectively, the near-origin and active-set regularity conditions on the penalty \(p_{\lambda_n}(\cdot)\), paralleling the conditions used by \citet{lam2009sparsistency} and satisfied by commonly used nonconvex penalties such as SCAD. 
	Further discussion and interpretation of these assumptions, along with proofs of Theorems~\ref{thm:G7} and~\ref{thm:G8}, are provided in the Supplementary sections E.4--E.6.
	\vspace{-0.2cm}
	\setcounter{theorem}{0}
	\begin{theorem}[Convergence rates]
		\label{thm:G7}
		Under assumptions (A1)--(A5), if $r_n\to0$ and 
		$r_n/\lambda_n\to 0$ as $n\to\infty$, then there exists a sequence of local minimizers $\hat{\bm\Gamma}_n$ satisfying
		\[
		\|\hat{\bm\Gamma}_n-\bm\Gamma_0\|^2_F\ =\ O_P\!\Big({s\,\dfrac{\log d}{n}}\Big),
		\qquad
		\|\hat{\bm\Sigma}_n-\bm\Sigma_0\|^2 =\ O_P\!\Big\{(s+1)\dfrac{\log d}{n}\Big\},
		\]
		where $\hat{\bm\Sigma}_n=\hat{\bm W}\,\hat{\bm\Gamma}_n\,\hat{\bm W}$ and $\hat{\bm  W}^2=\mathrm{diag}(\bm S)$.
	\end{theorem}
	Theorem~\ref{thm:G7} shows that the local minimizer \(\hat{\bm\Gamma}_n\) constructed above attains the same Frobenius-norm rate as in the Gaussian setting of Theorem 7 of \citet{lam2009sparsistency}. We next turn to support recovery and establish an analogue of their Theorem~8 for this same sequence of local minimizers. Let \(C_0<\infty\) denote the fixed localization constant associated with the local minimizer constructed in the proof of Theorem~\ref{thm:G7}, so that $
	\|\hat{\bm\Gamma}_n-\bm\Gamma_0\|_F\le C_0r_n
	$ with probability tending to one.
	\begin{theorem}[Sparsistency]
		\label{thm:G8}
		Suppose Assumptions \textnormal{(A1)--(A5)} of
		Theorem~\ref{thm:G7} hold and $r_n\to0,~{r_n}/{\lambda_n}\to0
		$ as $n\to\infty$. In addition, suppose that there exists \(C_h<\infty\) such that, with probability tending to one,
		$$
		\sup_{\substack{
				\bm\Gamma\in\mathcal C_d\\
				\|\bm\Gamma-\bm\Gamma_0\|_F\le C_0r_n
		}}
		\sup_{g\in\partial h_n(\bm\Gamma)}
		\|g\|_\infty
		\le
		C_gq_n+C_hr_n,
		$$
		where \(C_g\) is the constant in Assumption \textnormal{(A3)}.
		Then, with probability tending to one, the local minimizer
		\(\hat{\bm\Gamma}_n\) constructed in Theorem~\ref{thm:G7} satisfies $
		\hat{\gamma}_{n,ij}=0$ for all $(i,j)\in S^c$.
	\end{theorem}
	The additional condition in Theorem~\ref{thm:G8} extends the subgradient control in Assumption~(A3) at \(\bm\Gamma_0\) to the shrinking neighborhood containing $\hat{\bm\Gamma}_n$. This is needed for support recovery because the proof compares $\hat{\bm\Gamma}_n$ with nearby matrices obtained by setting candidate false-positive entries to zero. Further technical details are provided in Supplementary Section~E.5.
	\vspace{-0.5cm}
	\subsection{Theoretical properties of BLOC}
	\label{sec:bloc_theory}
	We establish theoretical guarantees for BLOC by analyzing the behavior of its underlying recursive modified pattern search (RMPS) routine in the transformed
	parameter space. The analysis is conducted on a compact search region $\mathcal C\subset\mathbb R^N$ for the transformed objective in~\eqref{eq:problem 3}.
	Through the mappings $\mathscr M$ and $\Phi$, every evaluated point corresponds to a valid correlation matrix, so the resulting properties characterize the
	corresponding BLOC iterates.
	
	Theorem~\ref{thm:stationarity} shows that failure of coordinate-wise descent at arbitrarily small step sizes implies first-order stationarity. Theorem~\ref{thm:conv_rate} establishes an \(O(1/r)\) convergence-rate bound for smooth convex objectives within a single \emph{run}, where \(r\) indexes unsuccessful polls, equivalently step-size reductions. Supplementary Lemma~S1 and Theorem~S1 provide complementary global-search results under a randomized multi-start restart scheme. Specifically, Lemma~S1 establishes that repeated restarts on increasingly refined grids reach any prescribed neighborhood of a global minimizer with probability one, while Theorem~S1 gives convergence in probability to a unique strict global minimizer under additional smoothness and local convexity conditions. These latter results concern an exploratory restart schedule in which successive \emph{runs} are initialized independently, rather than the warm-restart strategy used in the default BLOC implementation. Complete proofs, along with important remarks regarding Supplementary Theorem~S1 and additional comments on Theorem~\ref{thm:conv_rate}, are provided in Supplementary Section~E.
	\vspace{-0.2cm}
	\begin{theorem}[Stationarity under Coordinatewise Descent Failure]
		\label{thm:stationarity}
		Let \( f : \mathcal{C} \mapsto \mathbb{R} \) be differentiable on a compact hyperrectangle \( \mathcal{C} \subset \mathbb{R}^N \). Suppose \( \bm{\nu} = (\nu_1,\ldots, \nu_N) \in \mathcal{C} \), and there exists an open neighborhood \( \bm{U} \subset \mathcal{C} \) of \( \bm{\nu} \) such that \( f \) is convex on \( \bm{U} \). Let \( \delta_k = \frac{s}{\rho^k} \) for \( k \in \mathbb{N} \), with \( s > 0 \), \( \rho > 1 \). Define:
		$
		\bm{\nu}_k^{(i+)} = (\nu_1,\ldots, \nu_{i-1}, \nu_i+\delta_k, \nu_{i+1}, \ldots, \nu_N), ~
		\bm{\nu}_k^{(i-)} = (\nu_1,\ldots, \nu_{i-1}, \nu_i-\delta_k, \nu_{i+1}, \ldots, \nu_N)
		$ for \( i = 1,\ldots,N \). If for all \( k \in \mathbb{N} \) and all \( i = 1,\ldots,N \),
		\vspace{-0.3cm}
		\[
		f(\bm{\nu}) \leq f(\bm{\nu}_k^{(i+)}), \quad \text{and} \quad f(\bm{\nu}) \leq f(\bm{\nu}_k^{(i-)}),
		\]
		
		\vspace{-0.8cm}
		then the gradient of \( f \) at \( \bm{\nu} \) vanishes: \( \nabla f(\bm{\nu}) = \bm{0} \).
	\end{theorem}
	\vspace{-0.2cm}
	
	\begin{theorem}[Convergence Rate of BLOC]\label{thm:conv_rate}
		Let \( f : \mathcal{C} \rightarrow \mathbb{R} \) be a convex, continuously
		differentiable function defined on a compact convex set
		\( \mathcal{C} \subset \mathbb{R}^N \), and suppose its gradient
		\( \nabla f \) is Lipschitz continuous with Lipschitz constant \( L>0 \).
		Consider BLOC initialized at \( \bm{\nu}^{(0)}\in\mathcal C \), with
		initial step size \(s_0>0\), reduction factor \(\rho>1\), and coordinate
		directions \( \{\pm\bm e_1,\ldots,\pm\bm e_N\} \). Let \(\bm{\nu}^{(r)}\) denote the iterate at the unsuccessful poll
		corresponding to step size \(s_r=s_0/\rho^r\). Assume that
		\(\bm{\nu}^{(r)}\pm s_r\bm e_j\in\mathcal C\) for all
		\(j=1,\ldots,N\), and that the step size is reduced only when no polling
		direction yields a decrease in the objective. Then
		\vspace{-0.3cm}
		\[
		f(\bm{\nu}^{(r)})-f(\bm{\nu}^*)
		\leq \frac{C}{r+1},
		\]
		
		\vspace{-0.8cm}
		for some constant \(C>0\), where
		\(\bm{\nu}^*\in\arg\min_{\bm{\nu}\in\mathcal C}f(\bm{\nu})\).
	\end{theorem}
	Note that here $r$ denotes the number of unsuccessful polls, equivalently the number of step-size reductions, rather than the total number of BLOC iterations inside a \emph{run}. The number of intermediate polling iterations before each reduction depends on the objective landscape and should therefore be taken into account when interpreting this rate.
	\vspace{-0.6cm}
	\section{Numerical results}\label{sec:sims}\vspace{-0.2cm}
	To evaluate the finite-sample performance of BLOC, we consider a sequence of numerical experiments encompassing both synthetic benchmark optimization problems and sparse covariance estimation under standard statistical models. We first summarize results from a benchmark study on structured nonconvex objectives over correlation matrix spaces, before turning to simulation experiments assessing estimation accuracy, support recovery, and positive definiteness in low- and high-dimensional regimes ($n>d$ and $d\ge n$). Performance is evaluated relative to representative existing methods, providing a comparative assessment of the gains achieved by the proposed framework.
	\vspace{-0.5cm}
	\subsection{Benchmark study}\vspace{-0.2cm}
	We first assess the optimization performance of BLOC using four standard nonconvex benchmark functions Ackley, Griewank, Rosenbrock, and Rastrigin, adapted to the space of correlation matrices \citep{surjanovic2013virtual}. These benchmarks span oscillatory, multimodal, and narrow-valley landscapes and are evaluated over increasing dimensions ($d = 5, 10, 20, 50$), with additional experiments for BLOC conducted up to $d = 100$. We compare BLOC and its parallel implementation against constrained solvers in MATLAB (\texttt{fmincon}: interior-point, SQP, active-set) and optimization routines from the \texttt{Manopt} toolbox \citep{boumal2014manopt}. Across the benchmarks, BLOC shows stable performance as dimensionality increases, attaining near-optimal values on Ackley and Griewank and generally stronger solutions than the competing solvers on Rosenbrock and Rastrigin at larger dimensions. Parallelization yields pronounced gains in moderate to large dimensions. Full numerical results and runtime comparisons are reported in Tables~S1–S2 of the Supplementary Section C.
	\vspace{-0.6cm}
	\subsection{Sparse covariance estimation with Gaussian likelihood ($n>d$)}
	\label{sec:sims_small_p}\vspace{-0.2cm}
	We first examine BLOC in the classical regime where the number of observations exceeds the number of variables ($n>d$), so that the Gaussian log-likelihood provides a natural loss. In this setting, the optimization problem in \eqref{eq:problem 1} is instantiated with
	$h_n(\bm\Gamma)
	=
	\operatorname{tr}(\bm\Gamma^{-1}\hat{\bm\Gamma}_S)
	+\log|\bm\Gamma|$, where $\hat{\bm\Gamma}_S$ denotes the sample correlation matrix. We compare BLOC with SCAD and MCP penalties against two classes of competitors. The first is the $\ell_1$-penalized covariance estimator implemented in the \texttt{Spcov} package \citep{Bien2011}. The second is an LLA-MM approach for SCAD- and MCP-penalized covariance estimation, motivated by the local-linear-approximation strategy underlying I-LAMM \citep{Fan2018ILAMM} and adapted to the present covariance-estimation problem as described in Supplementary Section~D. This comparison is particularly useful because LLA-MM targets the same nonconvex penalties as BLOC, while relying on a penalty-specific iterative majorization strategy rather than black-box global search.
	
	We consider block-diagonal and uniform-sparse correlation structures over $(d,n)\in\{(20,50),(50,100),(100,150)\}$, with 10 independent replications per setting. The block-diagonal design consists of disjoint $5\times5$ blocks, whereas the uniform-sparse design contains approximately $95\%$, $98\%$, and $99\%$ off-diagonal zeros for the three dimension settings, respectively, with nonzero entries drawn from $\mathrm{Uniform}[0.3,0.6]$ subject to positive definiteness. Data are generated from a mean-zero multivariate normal distribution. Performance is assessed using RMSE, MAD, TPR, FPR, and MCC. All methods are evaluated over the same grid of 20 logarithmically spaced $\lambda$ values in $[10^{-2},10^{1}]$. For Spcov, $\lambda$ is selected using the package-recommended five-fold cross-validation procedure, whereas BLOC and LLA-MM select $\lambda$ using EBIC with $\gamma_{\mathrm{EBIC}}=0.5$ \citep{chen2008extended}. For both BLOC and LLA-MM, we use $a=3.7$ for SCAD and $\gamma=3$ for MCP. Full design construction and data-generation procedures are provided in Supplementary Section~D.
	
	Table~\ref{tab:small_p} shows that no single method uniformly dominates across all settings. For $(d,n)=(20,50)$ under the block-diagonal design, BLOC attains the highest MCC ($0.47$ for MCP), reflecting a favorable balance between sensitivity and false-positive control, although LLA-MM yields somewhat smaller RMSE and MAD. Under the corresponding uniform-sparse design, LLA-MM-MCP performs best overall, with MCC $0.65$ and the smallest estimation errors, while BLOC remains competitive and maintains low false-positive rates. At $(d,n)=(50,100)$, LLA-MM performs particularly well for the block-diagonal structure, improving both estimation error and MCC relative to BLOC, whereas BLOC attains somewhat higher TPR. Under the uniform-sparse structure, \texttt{Spcov} gives the highest MCC, while LLA-MM again provides stronger recovery than BLOC.
	
	The larger $(d,n)=(100,150)$ setting illustrates a different computational regime. BLOC continues to produce reasonable estimates under both covariance structures. In contrast, LLA-MM did not complete within 72 hours in either structure and is therefore not reported. \texttt{Spcov} returned an error for the block-diagonal setting and produced extremely poor estimates under the uniform-sparse setting, so its $d=100$ results are also omitted. Thus, while the specialized LLA-MM procedure can outperform BLOC in several low- and moderate-dimensional settings, BLOC remains applicable in the larger setting considered here and, unlike LLA-MM, does not require a penalty-specific optimization construction.
	
	Computation times are reported in Supplementary Table~S3. The specialized procedures are substantially faster in the smaller settings, whereas at $d=100$ LLA-MM exceeds the 72-hour limit while BLOC completes all configurations, illustrating the computational trade-off associated with BLOC's more general search strategy.

	\vspace{-0.5cm}
	\begin{table}[!h]
		\centering
		\resizebox{0.99\columnwidth}{!}{%
			\begin{tabular}{|c|c|c|lllll|}
				\hline
				Scenarios & \begin{tabular}[c]{@{}c@{}}Matrix\\ Type\end{tabular} & Methods & \multicolumn{1}{c}{TPR} & \multicolumn{1}{c}{FPR} & \multicolumn{1}{c}{MCC} & \multicolumn{1}{c}{RMSE} & \multicolumn{1}{c|}{MAD} \\ \hline
				\multirow{10}{*}{\begin{tabular}[c]{@{}c@{}}$d = 20$\\  $n = 50$\end{tabular}} & \multirow{5}{*}{Block-diagonal} & BLOC-MCP & 0.77 (0.014) & 0.22 (0.016) & 0.47 (0.015) & 0.086 (0.004) & 0.036 (0.001) \\
				&  & BLOC-SCAD & 0.78 (0.014) & 0.25 (0.018) & 0.46 (0.020) & 0.086 (0.003) & 0.038 (0.001) \\
				&  & LLA-MM-MCP & 0.77 (0.016) & 0.24 (0.020) & 0.45 (0.014) & 0.074 (0.003) & 0.031 (0.001) \\
				&  & LLA-MM-SCAD & 0.80 (0.009) & 0.27 (0.015) & 0.44 (0.014) & 0.077 (0.003) & 0.032 (0.001) \\
				&  & SPCOV & 0.93 (0.012) & 0.62 (0.037) & 0.28 (0.019) & 0.099 (0.005) & 0.067 (0.005) \\ \cline{2-8} 
				& \multirow{5}{*}{Uniform-sparse} & BLOC-MCP & 0.73 (0.045) & 0.06 (0.016) & 0.56 (0.024) & 0.061 (0.002) & 0.015 (0.001) \\
				&  & BLOC-SCAD & 0.79 (0.046) & 0.07 (0.015) & 0.54 (0.031) & 0.057 (0.004) & 0.015 (0.001) \\
				&  & LLA-MM-MCP & 0.82 (0.033) & 0.05 (0.017) & 0.65 (0.045) & 0.052 (0.004) & 0.012 (0.002) \\
				&  & LLA-MM-SCAD & 0.78 (0.047) & 0.06 (0.019) & 0.60 (0.049) & 0.057 (0.005) & 0.014 (0.001) \\
				&  & SPCOV & 0.57 (0.050) & 0.02 (0.004) & 0.59 (0.039) & 0.069 (0.003) & 0.015 (0.001) \\ \hline
				\multirow{10}{*}{\begin{tabular}[c]{@{}c@{}}$d = 50$\\  $n = 100$\end{tabular}} & \multirow{5}{*}{Block-diagonal} & BLOC-MCP & 0.83 (0.022) & 0.26 (0.029) & 0.35 (0.017) & 0.066 (0.003) & 0.026 (0.001) \\
				&  & BLOC-SCAD & 0.87 (0.017) & 0.32 (0.015) & 0.32 (0.008) & 0.065 (0.003) & 0.028 (0.001) \\
				&  & LLA-MM-MCP & 0.79 (0.007) & 0.16 (0.003) & 0.42 (0.008) & 0.040 (0.001) & 0.012 (0.000) \\
				&  & LLA-MM-SCAD & 0.80 (0.007) & 0.16 (0.004) & 0.43 (0.006) & 0.047 (0.002) & 0.013 (0.000) \\
				&  & SPCOV & 0.25 (0.129) & 0.13 (0.086) & 0.12 (0.069) & 0.117 (0.011) & 0.034 (0.003) \\ \cline{2-8} 
				& \multirow{5}{*}{Uniform-sparse} & BLOC-MCP & 0.74 (0.032) & 0.05 (0.005) & 0.41 (0.016) & 0.033 (0.002) & 0.005 (0.000) \\
				&  & BLOC-SCAD & 0.73 (0.033) & 0.04 (0.007) & 0.45 (0.015) & 0.035 (0.002) & 0.005 (0.000) \\
				&  & LLA-MM-MCP & 0.94 (0.019) & 0.04 (0.006) & 0.58 (0.032) & 0.023 (0.001) & 0.004 (0.000) \\
				&  & LLA-MM-SCAD & 0.96 (0.016) & 0.05 (0.004) & 0.52 (0.016) & 0.024 (0.001) & 0.004 (0.000) \\
				&  & SPCOV & 0.85 (0.022) & 0.01 (0.001) & 0.73 (0.023) & 0.034 (0.001) & 0.004 (0.000) \\ \hline
				\multirow{4}{*}{\begin{tabular}[c]{@{}c@{}}$d = 100$\\  $n = 150$\end{tabular}} & \multirow{2}{*}{Block-diagonal} & BLOC-MCP & 0.89 (0.018) & 0.46 (0.068) & 0.20 (0.033) & 0.049 (0.004) & 0.028 (0.004) \\
				&  & BLOC-SCAD & 0.81 (0.062) & 0.37 (0.086) & 0.27 (0.054) & 0.053 (0.004) & 0.024 (0.004) \\ \cline{2-8} 
				& \multirow{2}{*}{Uniform-sparse} & BLOC-MCP & 0.66 (0.009) & 0.03 (0.002) & 0.34 (0.010) & 0.027 (0.000) & 0.003 (0.000) \\
				&  & BLOC-SCAD & 0.65 (0.011) & 0.03 (0.002) & 0.35 (0.009) & 0.027 (0.000) & 0.003 (0.000) \\ \hline
		\end{tabular}}
		\caption{Simulation results for sparse covariance estimation with Gaussian likelihood under block-diagonal and uniform-sparse correlation structures. Reported metrics include true positive rate (TPR), false positive rate (FPR), Matthews correlation coefficient (MCC), root mean squared error (RMSE), and mean absolute deviation (MAD). Each entry is the average across 10 replications, with standard errors in parentheses. Results are shown for BLOC and LLA-MM under SCAD and MCP penalties, together with the LASSO-based Spcov method.}
		\label{tab:small_p} \vspace{-0.5cm}
	\end{table}
	\vspace{-0.2cm}
	
	\subsection{Sparse covariance estimation with Frobenius norm ({$d\ge n$})}
	\label{sec:sims_large_p} \vspace{-0.2cm}
	We next examine the high-dimensional regime in which the number of variables is comparable to or exceeds the sample size ($d\ge n$). Estimation accuracy is evaluated using the Frobenius norm loss,
	$
	h_n(\bm\Gamma)=\|\bm\Gamma-\hat{\bm\Gamma}_S\|_F^2
	$, where $\hat{\bm\Gamma}_S$ denotes the sample correlation matrix. Simulations are conducted under three sparse covariance structures: block-diagonal with blocks of size $10$, truncated-Toeplitz with geometrically decaying correlations truncated below $0.001$, and banded with
	$\Gamma_{ij}=(1-|i-j|/10)I(|i-j|\leq10)$. We consider
	$(d,n)\in\{(100,50),(100,100),(200,100)\}$; full construction and tuning details are provided in Supplementary Section~D.
	
	BLOC with nonconvex penalties (SCAD and MCP) is compared against a broad set of existing sparse covariance estimators, including the $\ell_1$-penalized ADMM estimator under eigenvalue constraints \citep{Liu2014}; iterative reweighting approaches such as IRW-ADMM \citep{Wen2021} and the classical double-loop reweighted ADMM \citep{Rothman2009}; block coordinate descent methods with nonconvex penalties, including $\ell_q$-BCD ($q=0.5$ and $q=0$) and SCAD-BCD \citep{Wen2016}; and elementwise thresholding rules (soft, hard, and SCAD). We additionally include the MCP-based majorization--minimization estimator of \cite{WeiZhao2023} in its correlation- and covariance-matrix formulations, denoted WeiZhao-MCP-corr and WeiZhao-MCP-cov, respectively. Implementations of the remaining competing methods are taken from the public repository accompanying \cite{Wen2021}. For BLOC and both WeiZhao formulations, $\lambda$ is selected from 15 logarithmically spaced values over $[10^{-3},1]$ using five cross-validation splits, minimizing the mean Frobenius validation error; BLOC uses $a=3.7$ for SCAD and $\gamma=3$ for MCP, while WeiZhao-MCP uses its default MCP parameterization.
	
	Performance is assessed using Frobenius and spectral norm errors, MAD, TPR, FPR, MCC, and positive-definiteness (PD) rate. Table~\ref{tab:d200_p100} reports results for $(d,n)=(200,100)$, while Supplementary Tables~S4 and~S5 report $(100,50)$ and $(100,100)$, respectively. For $(d,n)=(200,100)$, BLOC yields the smallest Frobenius and spectral norm errors under both the truncated-Toeplitz and banded structures. Under truncated-Toeplitz, BLOC also achieves the highest TPR with essentially zero FPR, resulting in MCC values around $0.60$, compared with at most $0.37$ among the competing methods. Under the banded structure, BLOC combines high TPR ($0.93$--$0.98$) with low FPR ($0.03$--$0.04$), producing MCC values among the strongest competitors while substantially reducing estimation error. The block-diagonal setting is less differentiated, with WeiZhao-MCP-corr and several other methods showing particularly strong support recovery.
	
	The supplementary experiments show a similar pattern. For $(d,n)=(100,50)$ and $(100,100)$, BLOC again provides the strongest estimation and support-recovery performance under truncated-Toeplitz, and under the banded structure BLOC-SCAD attains the highest or tied-highest MCC ($0.83$ and $0.87$, respectively) while maintaining near-zero FPR. Supplementary Table~S6 shows the corresponding computational trade-off: BLOC is substantially more computationally intensive than the specialized competing procedures, reflecting the additional cost of its global search.
	\begin{table}[!t]
		\centering
		\resizebox{0.99\columnwidth}{!}{%
			\begin{tabular}{clccccccc}
				\hline
				Scenario & Method & Frob. norm & Spec. norm & MAD & TPR & FPR & MCC & PD rate \\ \hline
				\multicolumn{1}{c|}{\multirow{13}{*}{\begin{tabular}[c]{@{}c@{}}Block\\ diagonal\end{tabular}}} & BLOC-MCP & 6.12 (0.883) & 2.88 (0.365) & 0.008 (0.001) & 1.00 (0.000) & 0.03 (0.007) & 0.89 (0.027) & 1.00 (0.000) \\
				\multicolumn{1}{c|}{} & BLOC-SCAD & 4.75 (0.409) & 2.44 (0.256) & 0.006 (0.001) & 1.00 (0.000) & 0.03 (0.006) & 0.87 (0.024) & 1.00 (0.000) \\
				\multicolumn{1}{c|}{} & WeiZhao-MCP-corr & 3.35 (0.366) & 1.56 (0.241) & 0.004 (0.000) & 1.00 (0.000) & 0.00 (0.000) & 1.00 (0.000) & 1.00 (0.000) \\
				\multicolumn{1}{c|}{} & WeiZhao-MCP-cov & 11.87 (0.498) & 6.24 (0.351) & 0.017 (0.001) & 1.00 (0.000) & 0.04 (0.019) & 0.87 (0.040) & 1.00 (0.000) \\
				\multicolumn{1}{c|}{} & L1ADMM & 12.49 (0.352) & 6.37 (0.272) & 0.023 (0.002) & 1.00 (0.000) & 0.19 (0.032) & 0.56 (0.033) & 1.00 (0.000) \\
				\multicolumn{1}{c|}{} & LqBCD$(q=0.5)$ & 9.22 (0.427) & 4.93 (0.387) & 0.012 (0.001) & 1.00 (0.000) & 0.00 (0.000) & 1.00 (0.000) & 1.00 (0.000) \\
				\multicolumn{1}{c|}{} & LqBCD$(q=0)$ & 8.57 (0.480) & 4.87 (0.421) & 0.011 (0.001) & 1.00 (0.000) & 0.00 (0.000) & 1.00 (0.000) & 1.00 (0.000) \\
				\multicolumn{1}{c|}{} & SCAD-BCD & 8.88 (0.449) & 4.82 (0.406) & 0.012 (0.001) & 1.00 (0.000) & 0.00 (0.000) & 1.00 (0.001) & 1.00 (0.000) \\
				\multicolumn{1}{c|}{} & IRW-ADMM & 8.77 (0.441) & 4.76 (0.402) & 0.012 (0.001) & 1.00 (0.000) & 0.00 (0.000) & 1.00 (0.000) & 1.00 (0.000) \\
				\multicolumn{1}{c|}{} & TradIRW & 8.95 (0.464) & 4.74 (0.353) & 0.012 (0.001) & 1.00 (0.000) & 0.00 (0.000) & 1.00 (0.000) & 1.00 (0.000) \\
				\multicolumn{1}{c|}{} & Softthres & 13.35 (0.410) & 6.84 (0.291) & 0.025 (0.002) & 1.00 (0.000) & 0.20 (0.039) & 0.55 (0.043) & 0.30 (0.153) \\
				\multicolumn{1}{c|}{} & Hardthres & 8.61 (0.476) & 4.88 (0.423) & 0.011 (0.001) & 1.00 (0.000) & 0.00 (0.000) & 1.00 (0.000) & 0.60 (0.163) \\
				\multicolumn{1}{c|}{} & SCAD-thres & 9.91 (0.490) & 5.57 (0.453) & 0.016 (0.001) & 1.00 (0.000) & 0.10 (0.018) & 0.69 (0.042) & 0.00 (0.000) \\ \hline
				\multicolumn{1}{c|}{\multirow{13}{*}{\begin{tabular}[c]{@{}c@{}}Truncated\\ Toeplitz\end{tabular}}} & BLOC-MCP & 6.20 (0.077) & 1.82 (0.046) & 0.010 (0.000) & 0.45 (0.026) & 0.00 (0.001) & 0.61 (0.020) & 1.00 (0.000) \\
				\multicolumn{1}{c|}{} & BLOC-SCAD & 6.14 (0.161) & 1.80 (0.083) & 0.009 (0.000) & 0.42 (0.023) & 0.00 (0.000) & 0.60 (0.018) & 1.00 (0.000) \\
				\multicolumn{1}{c|}{} & WeiZhao-MCP-corr & 9.20 (0.080) & 2.95 (0.047) & 0.014 (0.000) & 0.23 (0.003) & 0.02 (0.001) & 0.35 (0.006) & 1.00 (0.000) \\
				\multicolumn{1}{c|}{} & WeiZhao-MCP-cov & 10.22 (0.130) & 3.41 (0.073) & 0.016 (0.000) & 0.23 (0.004) & 0.02 (0.001) & 0.34 (0.005) & 1.00 (0.000) \\
				\multicolumn{1}{c|}{} & L1ADMM & 9.95 (0.112) & 3.25 (0.059) & 0.020 (0.000) & 0.38 (0.004) & 0.16 (0.003) & 0.23 (0.006) & 1.00 (0.000) \\
				\multicolumn{1}{c|}{} & LqBCD$(q=0.5)$ & 9.27 (0.091) & 2.81 (0.061) & 0.015 (0.000) & 0.22 (0.002) & 0.02 (0.001) & 0.33 (0.005) & 1.00 (0.000) \\
				\multicolumn{1}{c|}{} & LqBCD$(q=0)$ & 9.47 (0.079) & 2.76 (0.061) & 0.014 (0.000) & 0.19 (0.002) & 0.01 (0.001) & 0.36 (0.003) & 1.00 (0.000) \\
				\multicolumn{1}{c|}{} & SCAD-BCD & 9.26 (0.321) & 2.75 (0.138) & 0.016 (0.000) & 0.29 (0.016) & 0.07 (0.008) & 0.30 (0.008) & 1.00 (0.000) \\
				\multicolumn{1}{c|}{} & IRW-ADMM & 9.17 (0.172) & 2.77 (0.092) & 0.014 (0.000) & 0.20 (0.006) & 0.01 (0.001) & 0.36 (0.005) & 1.00 (0.000) \\
				\multicolumn{1}{c|}{} & TradIRW & 9.18 (0.085) & 2.81 (0.061) & 0.014 (0.000) & 0.19 (0.002) & 0.00 (0.000) & 0.37 (0.003) & 1.00 (0.000) \\
				\multicolumn{1}{c|}{} & Softthres & 10.39 (0.157) & 3.45 (0.074) & 0.019 (0.000) & 0.35 (0.011) & 0.12 (0.010) & 0.25 (0.010) & 1.00 (0.000) \\
				\multicolumn{1}{c|}{} & Hardthres & 11.33 (0.166) & 3.23 (0.090) & 0.016 (0.000) & 0.13 (0.006) & 0.00 (0.001) & 0.32 (0.005) & 0.00 (0.000) \\
				\multicolumn{1}{c|}{} & SCAD-thres & 9.81 (0.159) & 3.00 (0.102) & 0.018 (0.000) & 0.34 (0.014) & 0.12 (0.013) & 0.25 (0.013) & 0.00 (0.000) \\ \hline
				\multicolumn{1}{c|}{\multirow{13}{*}{Banded}} & BLOC-MCP & 6.79 (0.889) & 2.26 (0.226) & 0.009 (0.001) & 0.93 (0.035) & 0.03 (0.006) & 0.85 (0.018) & 1.00 (0.000) \\
				\multicolumn{1}{c|}{} & BLOC-SCAD & 5.83 (0.221) & 2.22 (0.181) & 0.008 (0.000) & 0.98 (0.003) & 0.04 (0.005) & 0.82 (0.016) & 1.00 (0.000) \\
				\multicolumn{1}{c|}{} & WeiZhao-MCP-corr & 11.27 (0.132) & 3.53 (0.042) & 0.016 (0.000) & 0.80 (0.005) & 0.02 (0.001) & 0.78 (0.006) & 1.00 (0.000) \\
				\multicolumn{1}{c|}{} & WeiZhao-MCP-cov & 12.46 (0.235) & 4.53 (0.095) & 0.017 (0.000) & 0.79 (0.006) & 0.02 (0.001) & 0.77 (0.005) & 1.00 (0.000) \\
				\multicolumn{1}{c|}{} & L1ADMM & 10.80 (0.204) & 3.93 (0.093) & 0.019 (0.000) & 0.89 (0.005) & 0.16 (0.005) & 0.50 (0.009) & 1.00 (0.000) \\
				\multicolumn{1}{c|}{} & LqBCD$(q=0.5)$ & 9.09 (0.177) & 3.11 (0.096) & 0.012 (0.000) & 0.88 (0.004) & 0.02 (0.001) & 0.83 (0.006) & 1.00 (0.000) \\
				\multicolumn{1}{c|}{} & LqBCD$(q=0)$ & 9.79 (0.165) & 3.24 (0.095) & 0.013 (0.000) & 0.80 (0.004) & 0.01 (0.000) & 0.85 (0.004) & 1.00 (0.000) \\
				\multicolumn{1}{c|}{} & SCAD-BCD & 9.18 (0.473) & 3.17 (0.132) & 0.013 (0.000) & 0.90 (0.023) & 0.08 (0.017) & 0.69 (0.040) & 1.00 (0.000) \\
				\multicolumn{1}{c|}{} & IRW-ADMM & 8.96 (0.161) & 3.01 (0.086) & 0.012 (0.000) & 0.81 (0.005) & 0.02 (0.001) & 0.77 (0.006) & 1.00 (0.000) \\
				\multicolumn{1}{c|}{} & TradIRW & 9.13 (0.162) & 3.07 (0.086) & 0.012 (0.000) & 0.79 (0.005) & 0.00 (0.000) & 0.87 (0.003) & 1.00 (0.000) \\
				\multicolumn{1}{c|}{} & Softthres & 11.37 (0.225) & 4.14 (0.094) & 0.019 (0.000) & 0.87 (0.006) & 0.13 (0.003) & 0.53 (0.007) & 0.00 (0.000) \\
				\multicolumn{1}{c|}{} & Hardthres & 10.71 (0.254) & 3.25 (0.102) & 0.013 (0.000) & 0.67 (0.016) & 0.00 (0.001) & 0.79 (0.007) & 0.00 (0.000) \\
				\multicolumn{1}{c|}{} & SCAD-thres & 9.77 (0.155) & 3.39 (0.126) & 0.016 (0.000) & 0.87 (0.006) & 0.13 (0.003) & 0.53 (0.007) & 0.00 (0.000) \\ \hline
		\end{tabular}}
		\caption{Simulation results for sparse covariance estimation under Frobenius loss with $(d,n)=(200,100)$. Performance is evaluated across block-diagonal, truncated-Toeplitz, and banded covariance structures. Reported metrics are Frobenius norm error (Frob. norm), spectral norm error (Spec. norm), mean absolute deviation (MAD), true positive rate (TPR), false positive rate (FPR), Matthews correlation coefficient (MCC), and positive-definiteness (PD) rate. Entries are means across 10 independent replications, with standard errors in parentheses.}
		\label{tab:d200_p100} \vspace{-0.5cm}
	\end{table}
	\vspace{-0.7cm}
	\subsection{Sparse covariance estimation with matrix-depth loss}
	\label{sec:sims_depth}\vspace{-0.2cm}
	To examine BLOC in a genuinely discontinuous setting, we consider robust sparse correlation estimation using a finite-direction matrix-depth criterion motivated by \cite{chen2018robust}. For $K$ fixed unit directions $\{\bm u_k\}_{k=1}^K$, we use the average empirical directional depth as the loss component and combine it with SCAD or MCP penalization of the off-diagonal entries. Since the empirical depth is constructed from indicator counts, it is piecewise constant with jumps as a function of the candidate correlation matrix, yielding a discontinuous objective for which standard local smooth-optimization routines can be problematic.
	
	We consider $(d,n)=(20,200)$ under block-diagonal and banded correlation structures, with $10\%$ of observations replaced by draws from $N_d(\bm 0,9\bm I_d)$ to induce gross Gaussian contamination. The block-diagonal design follows the corresponding Frobenius-loss setting but uses blocks of size $5$, while the banded structure is unchanged. BLOC-SCAD and BLOC-MCP are compared with active-set (AS), interior-point (IP), and SQP implementations of \texttt{fmincon}, and with Barzilai--Borwein (BB), conjugate-gradient (CG), and steepest-descent (SD) methods from \texttt{Manopt}. For fairness, all methods are initialized at $\bm I_d$ and use the same penalized matrix-depth objective, data realizations, fixed directions, and tuning protocol. Specifically, $\lambda$ is selected from 15 logarithmically spaced values over $[10^{-3},1]$ using five repeated train--test splits, minimizing the mean Frobenius distance to a Kendall-$\tau$-based robust correlation estimate computed from the held-out observations. Full construction and implementation details are given in Supplementary Section~D. 
	
	\begin{table}[t]
		\centering
		\resizebox{0.9\columnwidth}{!}{%
			\begin{tabular}{clccccccc}
				\hline
				Scenario & Method & \begin{tabular}[c]{@{}c@{}}Frob. \\ norm\end{tabular} & \begin{tabular}[c]{@{}c@{}}Spec.\\ norm\end{tabular} & MAD & TPR & FPR & MCC & \begin{tabular}[c]{@{}c@{}}Depth\\ (theoretical\\ max = 0.5)\end{tabular} \\ \hline
				\multicolumn{1}{c|}{\multirow{14}{*}{\begin{tabular}[c]{@{}c@{}}Block\\ diagonal\end{tabular}}} & BLOC-SCAD & 1.685 (0.1734) & 1.027 (0.1487) & 0.031 (0.0029) & 0.990 (0.010) & 0.051 (0.006) & 0.882 (0.013) & 0.461 (0.0027) \\
				\multicolumn{1}{c|}{} & BLOC-MCP & 1.648 (0.1776) & 1.021 (0.1468) & 0.030 (0.0030) & 0.990 (0.010) & 0.044 (0.006) & 0.896 (0.014) & 0.461 (0.0027) \\
				\multicolumn{1}{c|}{} & fmincon:AS-SCAD & 7.156 (0.0001) & 3.200 (0.0001) & 0.160 (0.0000) & 0.200 (0.133) & 0.188 (0.125) & --- & 0.417 (0.0010) \\
				\multicolumn{1}{c|}{} & fmincon:AS-MCP & 7.156 (0.0001) & 3.200 (0.0001) & 0.160 (0.0000) & 0.200 (0.133) & 0.188 (0.125) & --- & 0.417 (0.0010) \\
				\multicolumn{1}{c|}{} & fmincon:IP-SCAD & 7.155 (0.0000) & 3.200 (0.0000) & 0.160 (0.0000) & 0.000 (0.000) & 0.000 (0.000) & --- & 0.417 (0.0010) \\
				\multicolumn{1}{c|}{} & fmincon:IP-MCP & 7.155 (0.0000) & 3.200 (0.0000) & 0.160 (0.0000) & 0.000 (0.000) & 0.000 (0.000) & --- & 0.417 (0.0010) \\
				\multicolumn{1}{c|}{} & fmincon:SQP-SCAD & 7.155 (0.0004) & 3.200 (0.0001) & 0.160 (0.0000) & 0.200 (0.133) & 0.188 (0.125) & --- & 0.417 (0.0010) \\
				\multicolumn{1}{c|}{} & fmincon:SQP-MCP & 7.155 (0.0004) & 3.200 (0.0001) & 0.160 (0.0000) & 0.200 (0.133) & 0.188 (0.125) & --- & 0.417 (0.0010) \\
				\multicolumn{1}{c|}{} & Manopt:BB-SCAD & 7.155 (0.0000) & 3.200 (0.0004) & 0.160 (0.0001) & 0.110 (0.099) & 0.104 (0.093) & --- & 0.417 (0.0010) \\
				\multicolumn{1}{c|}{} & Manopt:BB-MCP & 7.156 (0.0001) & 3.201 (0.0006) & 0.160 (0.0001) & 0.138 (0.103) & 0.131 (0.097) & --- & 0.417 (0.0010) \\
				\multicolumn{1}{c|}{} & Manopt:CG-SCAD & 7.146 (0.0052) & 3.217 (0.0087) & 0.167 (0.0036) & 0.865 (0.098) & 0.816 (0.092) & --- & 0.417 (0.0010) \\
				\multicolumn{1}{c|}{} & Manopt:CG-MCP & 7.157 (0.0034) & 3.240 (0.0171) & 0.173 (0.0058) & 0.730 (0.120) & 0.679 (0.119) & --- & 0.417 (0.0010) \\
				\multicolumn{1}{c|}{} & Manopt:SD-SCAD & 7.174 (0.0179) & 3.259 (0.0500) & 0.170 (0.0083) & 0.497 (0.166) & 0.466 (0.155) & --- & 0.417 (0.0010) \\
				\multicolumn{1}{c|}{} & Manopt:SD-MCP & 7.157 (0.0030) & 3.217 (0.0093) & 0.164 (0.0028) & 0.500 (0.167) & 0.461 (0.154) & --- & 0.417 (0.0010) \\ \hline
				\multicolumn{1}{c|}{\multirow{14}{*}{Banded}} & BLOC-SCAD & 4.942 (0.4091) & 3.946 (0.3168) & 0.168 (0.0143) & 0.662 (0.048) & 0.002 (0.002) & 0.631 (0.040) & 0.448 (0.0040) \\
				\multicolumn{1}{c|}{} & BLOC-MCP & 4.910 (0.4035) & 3.925 (0.3121) & 0.167 (0.0143) & 0.659 (0.046) & 0.002 (0.002) & 0.628 (0.039) & 0.449 (0.0040) \\
				\multicolumn{1}{c|}{} & fmincon:AS-SCAD & 9.848 (0.0109) & 7.697 (0.0130) & 0.367 (0.0003) & 0.400 (0.163) & 0.338 (0.138) & --- & 0.390 (0.0009) \\
				\multicolumn{1}{c|}{} & fmincon:AS-MCP & 9.849 (0.0107) & 7.698 (0.0128) & 0.367 (0.0003) & 0.400 (0.163) & 0.338 (0.138) & --- & 0.390 (0.0009) \\
				\multicolumn{1}{c|}{} & fmincon:IP-SCAD & 9.874 (0.0000) & 7.728 (0.0000) & 0.368 (0.0000) & 0.000 (0.000) & 0.000 (0.000) & --- & 0.389 (0.0009) \\
				\multicolumn{1}{c|}{} & fmincon:IP-MCP & 9.874 (0.0000) & 7.728 (0.0000) & 0.368 (0.0000) & 0.000 (0.000) & 0.000 (0.000) & --- & 0.389 (0.0009) \\
				\multicolumn{1}{c|}{} & fmincon:SQP-SCAD & 9.874 (0.0004) & 7.728 (0.0004) & 0.367 (0.0000) & 0.400 (0.163) & 0.338 (0.138) & --- & 0.389 (0.0009) \\
				\multicolumn{1}{c|}{} & fmincon:SQP-MCP & 9.874 (0.0004) & 7.728 (0.0004) & 0.367 (0.0000) & 0.400 (0.163) & 0.338 (0.138) & --- & 0.389 (0.0009) \\
				\multicolumn{1}{c|}{} & Manopt:BB-SCAD & 9.874 (0.0000) & 7.728 (0.0000) & 0.368 (0.0000) & 0.049 (0.049) & 0.054 (0.052) & --- & 0.389 (0.0009) \\
				\multicolumn{1}{c|}{} & Manopt:BB-MCP & 9.874 (0.0000) & 7.728 (0.0000) & 0.368 (0.0000) & 0.109 (0.099) & 0.089 (0.084) & --- & 0.389 (0.0009) \\
				\multicolumn{1}{c|}{} & Manopt:CG-SCAD & 9.882 (0.0158) & 7.720 (0.0064) & 0.370 (0.0015) & 0.744 (0.127) & 0.611 (0.108) & --- & 0.390 (0.0009) \\
				\multicolumn{1}{c|}{} & Manopt:CG-MCP & 9.874 (0.0092) & 7.728 (0.0067) & 0.370 (0.0018) & 0.635 (0.142) & 0.551 (0.122) & --- & 0.390 (0.0009) \\
				\multicolumn{1}{c|}{} & Manopt:SD-SCAD & 9.867 (0.0037) & 7.689 (0.0206) & 0.370 (0.0010) & 0.786 (0.131) & 0.669 (0.112) & --- & 0.390 (0.0008) \\
				\multicolumn{1}{c|}{} & Manopt:SD-MCP & 9.873 (0.0153) & 7.544 (0.1500) & 0.374 (0.0034) & 0.885 (0.099) & 0.752 (0.084) & --- & 0.391 (0.0012) \\ \hline
		\end{tabular}}
		\caption{Simulation results for sparse correlation estimation under the discontinuous matrix-depth loss with $(d,n)=(20,200)$ and $10\%$ gross contamination. BLOC with SCAD and MCP penalties is compared with local optimization methods from \texttt{fmincon} and \texttt{Manopt} applied to the same penalized objective. Reported metrics are Frobenius and spectral norm errors, mean absolute deviation (MAD), true positive rate (TPR), false positive rate (FPR), Matthews correlation coefficient (MCC), and the achieved unpenalized empirical depth, whose theoretical maximum is $0.5$. Entries are means across 10 replications, with standard errors in parentheses. ``---'' denotes an undefined mean MCC due to at least one undefined replicate.}
		\label{tab:mat_depth} \vspace{-0.5cm}
	\end{table}
	
	Across both structures, BLOC attains substantially greater empirical depth and smaller estimation errors than the local optimization methods (Table~\ref{tab:mat_depth}). Most competing methods return solutions close to the identity initialization. This behavior is consistent with the discontinuous, locally flat objective: small local moves may provide no improvement, while a better region may require evaluating a candidate at a substantially larger distance. BLOC can directly examine such finite-distance moves through its coarse coordinate polling, without requiring an improving path through the intervening region. Under the block-diagonal design, for example, the Frobenius error is approximately $1.65$ for BLOC compared with about $7.15$ for the competing optimizers, while the achieved depth increases from approximately $0.42$ to $0.46$. Similar gains are observed under the banded design. Although some local methods attain relatively high TPR in the banded setting, this occurs at the expense of markedly inflated FPR, whereas BLOC maintains low false-positive rates while providing substantially better matrix estimation. These results illustrate a setting in which BLOC is particularly useful: constrained covariance estimation with a discontinuous or black-box criterion for which local smooth-optimization methods can have difficulty traversing flat regions of the objective. Final-fit computation times are reported in Supplementary Table~S7.
	
	\vspace{-0.7cm}
	\section{Pathway-informed correlation estimation for pan-gynecologic proteomics data}
	\label{sec:realdata}\vspace{-0.2cm}
	We illustrate the utility of the proposed BLOC framework through an application to reverse-phase protein array (RPPA) data from The Cancer Genome Atlas (TCGA), focusing on five pan-gynecologic cancers: breast carcinoma (BRCA), cervical squamous cell carcinoma (CESC), ovarian serous cystadenocarcinoma (OV), uterine corpus endometrial carcinoma (UCEC), and uterine carcinosarcoma (UCS). These cancers exhibit shared hormonal drivers and pathway dysregulation, yet they differ markedly in clinical presentation and progression. Understanding the protein--protein correlation structure across key signaling modules is therefore of biological and translational interest. Our proteomics data arises from The Cancer Proteome Atlas (TCPA, \citealp{Li2013}), which provides protein abundance data for TCGA samples, and consists of RPPA-based quantifications using antibodies that cover functions including proliferation, DNA damage, polarity, vesicle function, epithelial–mesenchymal transition (EMT), invasiveness, hormone signaling, apoptosis, metabolism, immunological and stromal function, as well as other critical cellular signaling pathways \citep{Akbani2014}.

	Following prior literature highlighting the central role of specific signaling pathways in pan-gynecologic cancers \citep{Akbani2014, Das2020}, we focused on a curated panel of $d=27$ proteins that are routinely studied in this context. These proteins are commonly organized into well-characterized functional modules, or pathways, including Cell Cycle (CDK1, CYCLINB1, CYCLINE2, P27\_pT157, P27\_pT198, PCNA, FOXM1), Hormone Receptor (ERALPHA, ERALPHA\_pS118, PR, AR), Hormone Signaling (BCL2, INPP4B, GATA3), PI3K/AKT (AKT\_ps473, AKT\_ps308, GSK3ALPHABETA\_pS21S9, GSK3\_pS9, PRAS40\_pT246, TUBERIN\_pT1462, PTEN), and Breast Reactive (BETACATENIN, CAVEOLIN1, MYH11, RAB11, GAPDH, RBM15). 
	
	To encode biological pathway structure, we exploit BLOC’s flexibility to optimize arbitrary objectives over the correlation matrix space. We incorporate structured penalization through a penalty cover that exempts within-pathway pairs from shrinkage while penalizing cross-pathway associations. Since proteins within the same pathway are biologically expected to interact, their correlations are left unpenalized, whereas cross-pathway correlations are penalized unless strongly supported by the data. Let the $d$ proteins be partitioned into $G=5$ non-overlapping pathways, and let $g(i)\in\{1,\dots,5\}$ denote the pathway membership of protein $i$. Define the symmetric penalty-cover matrix $\bm P=(P_{ij})$ by
	\vspace{-0.3cm}
	\[
	P_{ij} \;=\; 
	\begin{cases}
		0, & \text{if } g(i)=g(j),\\[2pt]
		1, & \text{if } g(i)\neq g(j),
	\end{cases}
	\qquad P_{ij}=P_{ji}.
	\]
	
	\vspace{-0.6cm}
	Thus, only cross-pathway pairs ($P_{ij}=1$) are penalized.
	
	For each cancer type, we estimate the correlation matrix by solving
	$
	\min_{\bm\Gamma\in\mathcal C_d}
	h_n(\bm\Gamma)
	+
	\lambda\sum_{1\le i<j\le d}
	P_{ij}\rho_{\mathrm{SCAD}}(|\gamma_{ij}|;a)$, where $h_n(\bm\Gamma)$ is the Frobenius loss, $\lambda>0$ is a tuning parameter, and $\rho_{\mathrm{SCAD}}(\cdot;a)$ denotes the SCAD penalty. For each cancer
	type, $\lambda$ was selected by minimizing EBIC \citep{chen2008extended} over 10 logarithmically spaced values in $[10^{-2},10^{1}]$, with
	$\gamma_{\mathrm{EBIC}}=0.5$; the SCAD shape parameter was fixed at $a=3.7$. Equivalently, penalization applies only to the Hadamard-masked entries
	$P\circ\Gamma$, yielding a block-structured penalty that preserves within-pathway coherence while encouraging sparsity across pathways.

	\begin{figure}[htbp]
		\centering
		\begin{subfigure}{0.46\textwidth}
			\centering
			\includegraphics[width=\linewidth]{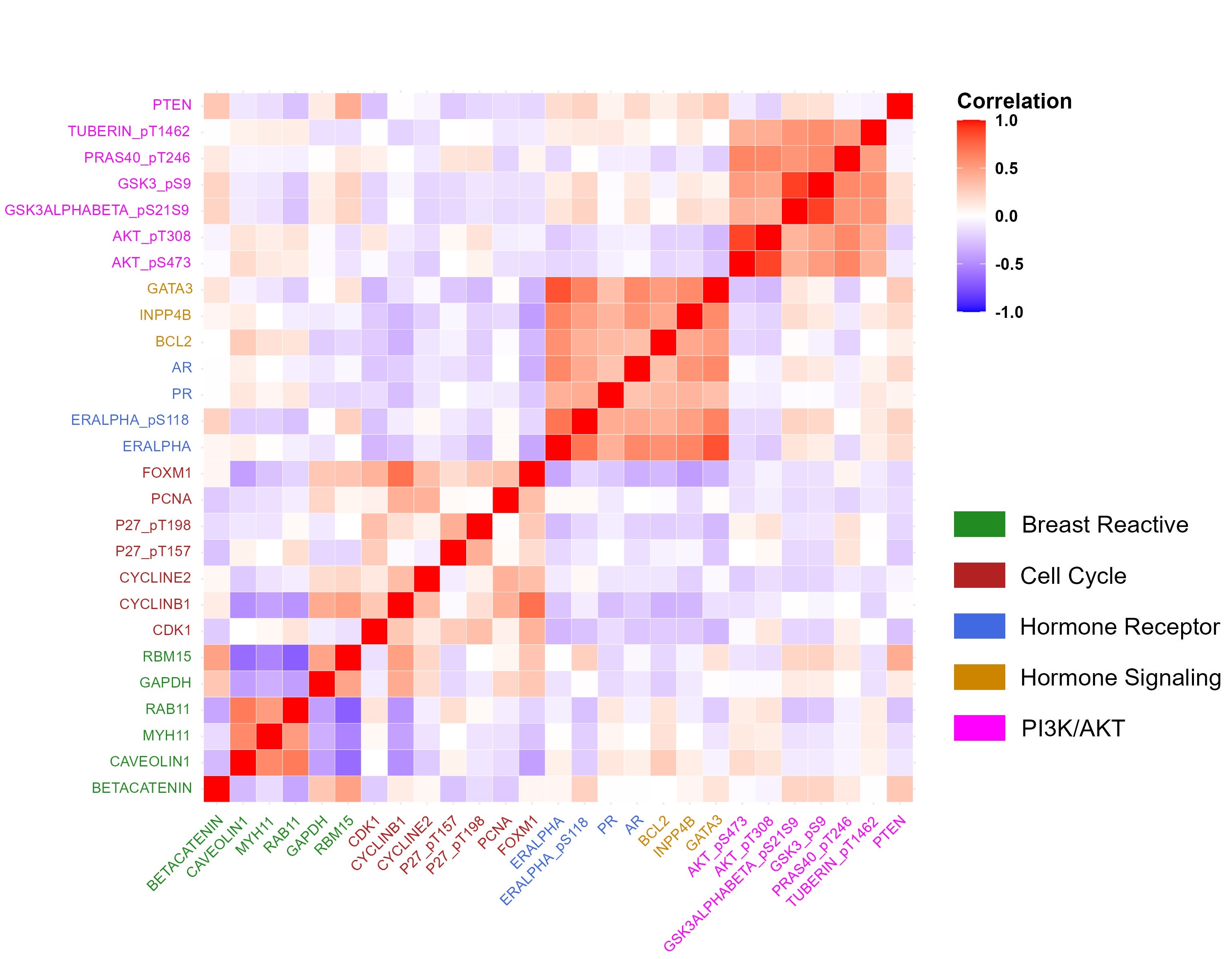}
			\caption{BRCA}
			\label{fig:BRCA}
		\end{subfigure}
		\hfill
		\begin{subfigure}{0.46\textwidth}
			\centering
			\includegraphics[width=\linewidth]{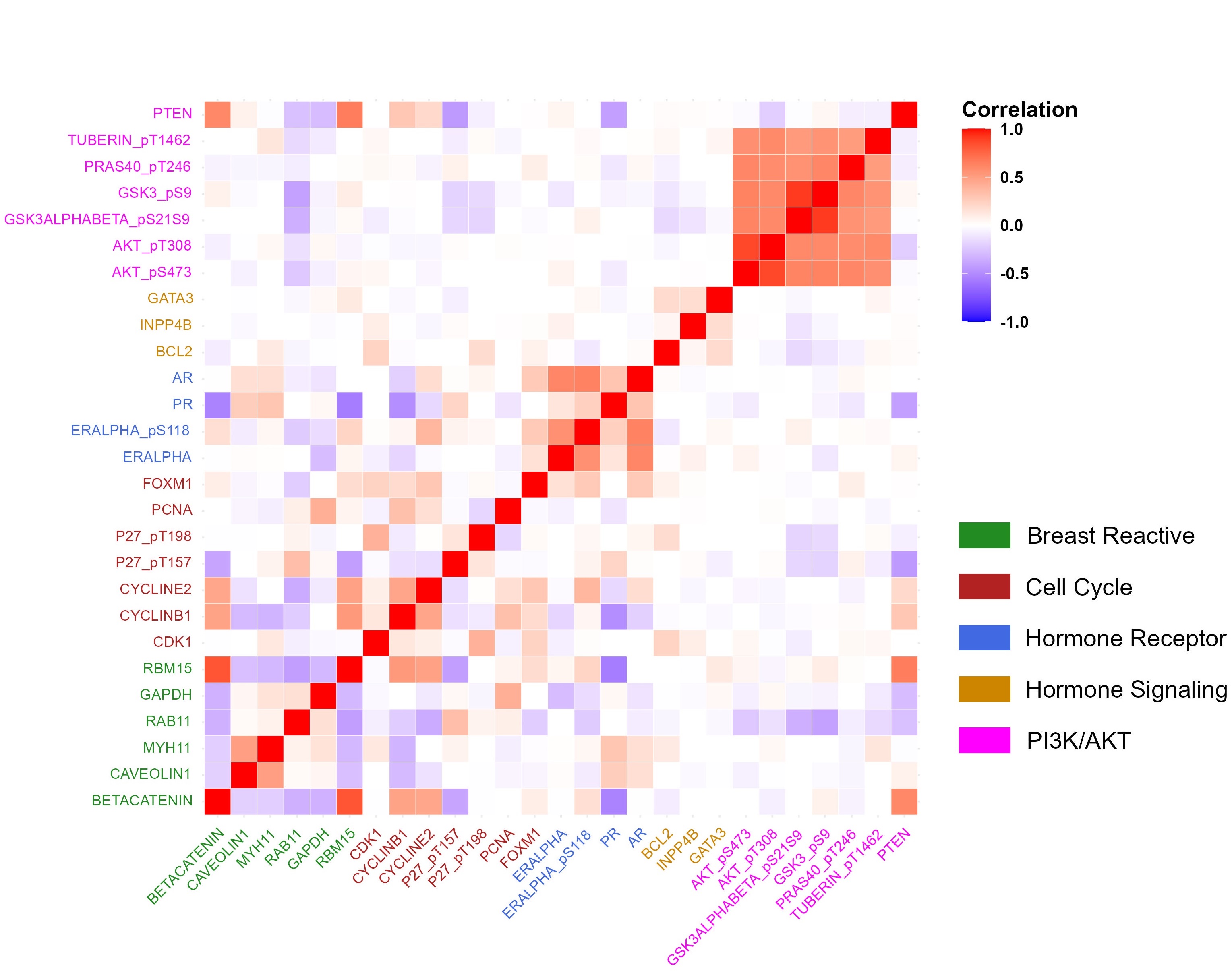}
			\caption{CESC}
			\label{fig:CESC}
		\end{subfigure}
		
		\begin{subfigure}{0.46\textwidth}
			\centering
			\includegraphics[width=\linewidth]{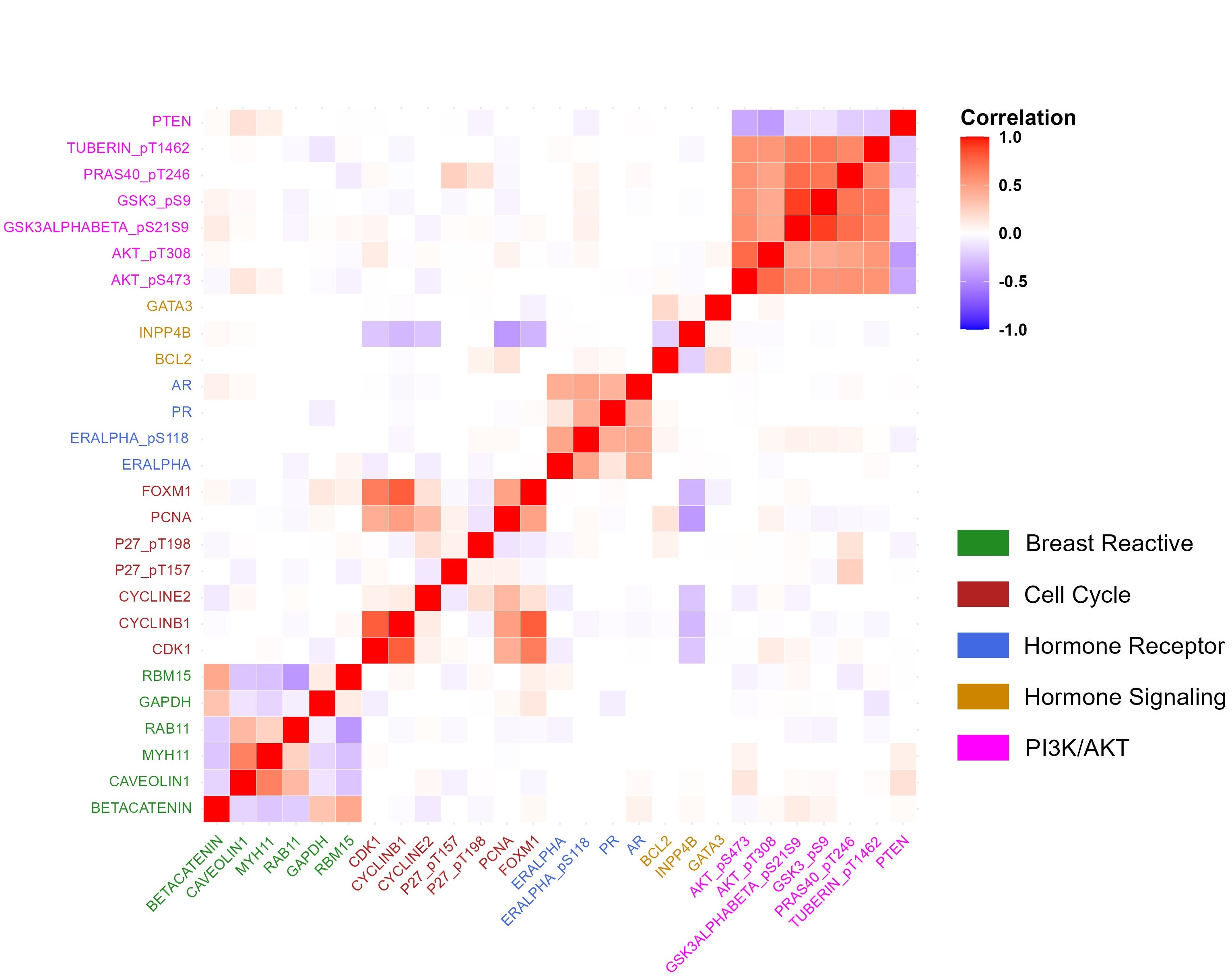}
			\caption{OV}
			\label{fig:OV}
		\end{subfigure}
		\hfill
		\begin{subfigure}{0.46\textwidth}
			\centering
			\includegraphics[width=\linewidth]{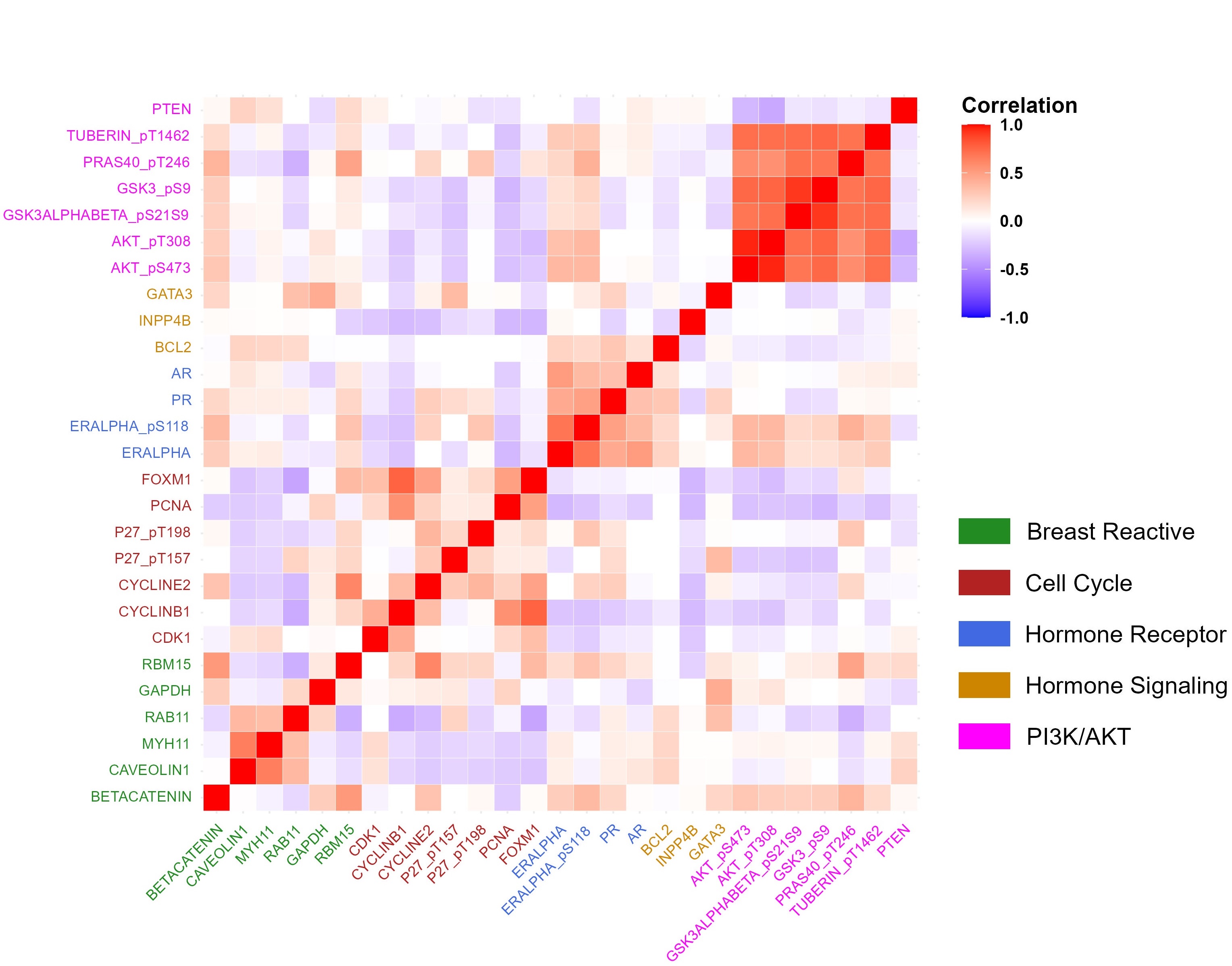}
			\caption{UCEC}
			\label{fig:UCEC}
		\end{subfigure}
		
		\begin{subfigure}{0.46\textwidth}
			\centering
			\includegraphics[width=\linewidth]{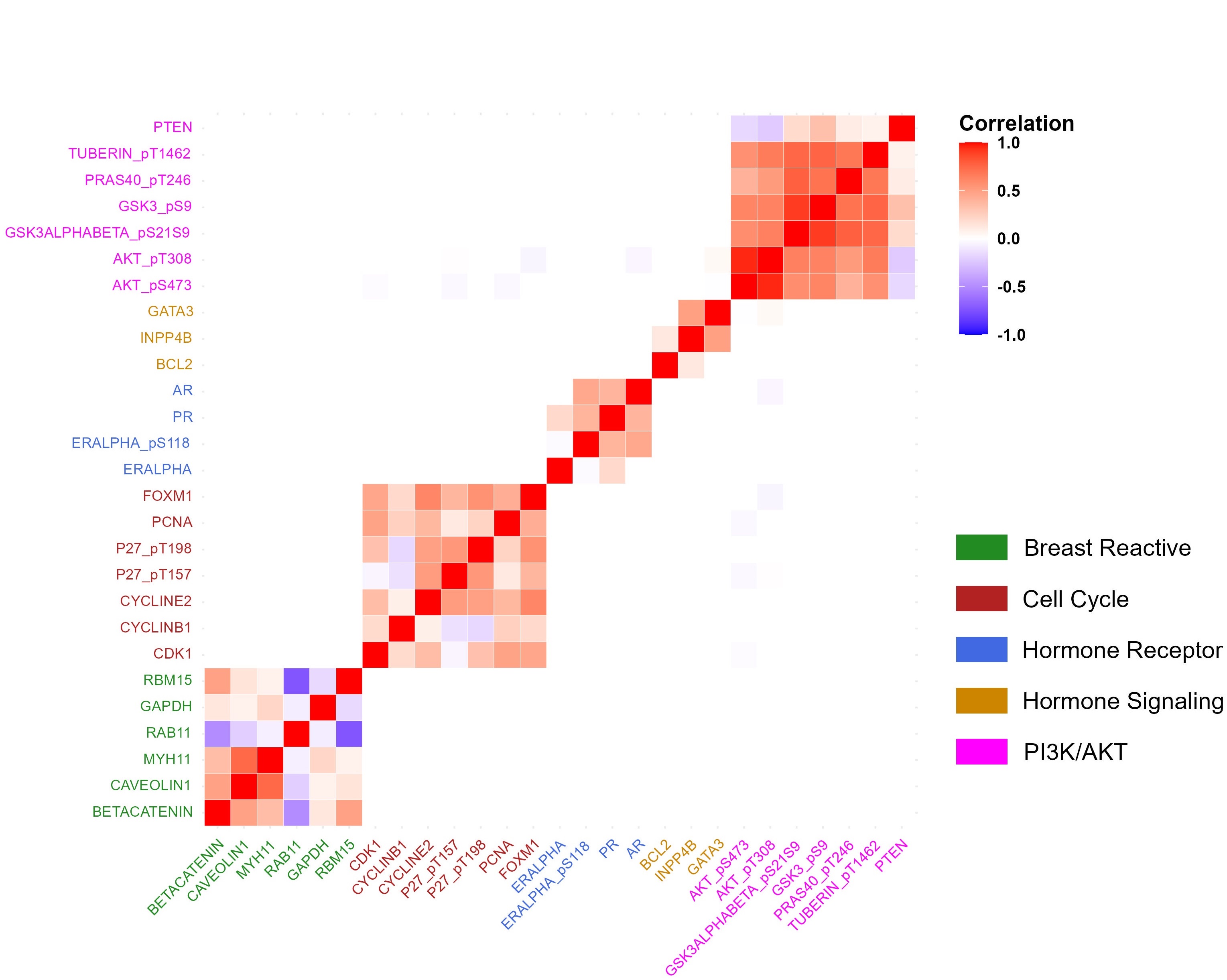}
			\caption{UCS}
			\label{fig:UCS}
		\end{subfigure}
		
		\caption{Estimated sparse correlation heatmaps for five pan-gynecologic cancers using BLOC with SCAD and a pathway-based penalty cover. Within-pathway coherence is preserved, while across-pathway edges highlight tumor-specific differences in integration.}
		\label{fig:all_corr}
	\end{figure}
	
	The estimated sparse correlation heatmaps for each cancer type are displayed in Figures~\ref{fig:BRCA}--\ref{fig:UCS}. Several clear and biologically interpretable features emerge from these results. First, by construction, within-pathway associations are preserved, and indeed strong positive correlations are consistently observed among proteins within the Cell Cycle module (FOXM1, PCNA, CYCLINB1, P27 phospho-sites) as well as within the PI3K/AKT module (AKT\_pS473, AKT\_pT308, PRAS40, GSK3 isoforms, PTEN). This is reassuring, as it confirms that the penalty cover preserves expected pathway coherence. 
	
	Beyond these within-pathway signals, the penalization highlights selective cross-pathway connections. In BRCA and UCEC, Hormone Signaling proteins such as INPP4B and GATA3 show positive correlations with ERALPHA and PR from the Hormone Receptor pathway, consistent with known crosstalk between estrogen receptor activity and PI3K signaling. In OV, a subset of PI3K/AKT proteins retain residual correlation with Cell Cycle regulators, suggesting potential pathway integration that may drive tumor proliferation. CESC, in contrast, exhibits weaker and more diffuse cross-talk, reflecting a less hormone-driven biology. UCS presents a strikingly different pattern, with two nearly independent correlation blocks--one corresponding to the Cell Cycle cluster and the other to PI3K/AKT--indicating a high degree of separation between modules, consistent with its aggressive and heterogeneous profile. It is worth noting that the estimated sparsity patterns are influenced by sample sizes across cancer types (BRCA: $n=879$, CESC: $n=171$, OV: $n=428$, UCEC: $n=404$, UCS: $n=48$). In particular, the extreme sparsity observed for UCS likely reflects, in part, the limited sample size, which reduces power to detect moderate correlations. Thus, while the block structure in UCS may indicate biological heterogeneity, it should be interpreted cautiously given the sample-size constraints.
	\vspace{-0.1cm}
	
	Taken together, these results underscore the value of embedding biological prior knowledge through a structured penalty cover. The approach preserves expected within-pathway correlation while adaptively shrinking across-pathway edges, thereby uncovering tumor-specific differences in pathway integration. BRCA and UCEC demonstrate strong Hormone Receptor-signaling interplay, OV suggests partial integration between PI3K/AKT and the Cell Cycle, CESC reveals comparatively sparse cross-talk, and UCS highlights sharp modular separation, though in part amplified by limited sample size. These differences may have implications for therapeutic targeting, as they suggest varying degrees of pathway dependency across tumor types.

	\vspace{-0.8cm}
	\section{Discussion}
	\label{sec:discuss}\vspace{-0.4cm}
	We have developed \textsc{BLOC}, a general-purpose algorithm for sparse covariance estimation that unifies reparameterization of correlation matrices with a recursive global search strategy. The central contribution is methodological generality: BLOC is designed to optimize arbitrary objective functions over $\mathcal{C}_d$, the $d \times d$ correlation matrix space, with no restrictions on the choice of penalty or loss. This makes the framework broadly applicable to a wide spectrum of statistical problems where correlation structure plays a central role. Unlike existing methods that are tailored to specific penalties or likelihoods, BLOC can seamlessly accommodate nonconvex penalties such as SCAD and MCP, as well as customized masking structures that encode prior knowledge. The algorithm therefore bridges a gap between theoretical flexibility and practical feasibility, offering researchers a tool to directly impose domain-informed structures while preserving validity of the correlation matrix.
	
	Our empirical investigations also clarify this trade-off. Under the smooth Gaussian-likelihood setting, specialized LLA-MM can outperform BLOC in some low- and moderate-dimensional settings, while BLOC remains competitive and applicable in the larger setting considered. Under Frobenius loss, BLOC performs particularly strongly for the truncated-Toeplitz and banded structures, while remaining competitive with specialized estimators across the other settings. In contrast, under the discontinuous matrix-depth criterion, BLOC attains markedly better depth, estimation accuracy, and support-recovery balance than the local optimization competitors. Together with the benchmark study, these results identify BLOC's main strength: constrained correlation estimation when the objective is nonconvex, nonsmooth, black-box, or otherwise lacks an effective bespoke solver.
	
	From a computational perspective, BLOC is naturally suited for parallelization. Each iteration evaluates $d(d-1)$ perturbations in the angular coordinate system, which can be distributed across independent threads with minimal communication. This parallel structure can substantially reduce runtime in moderate- to higher-dimensional problems, partly offsetting the additional computational cost of BLOC's broader search strategy. Nevertheless, specialized procedures can be considerably faster when tailored algorithms are available for a particular loss or penalty. The real-data analysis of TCGA reverse-phase protein array data further illustrates a different aspect of BLOC's flexibility: by specifying a biologically motivated penalty cover, we incorporated prior signaling pathway knowledge into correlation estimation without modifying the core algorithm. The resulting networks exhibited expected within-pathway coherence and tumor-specific cross-pathway differences, underscoring the interpretability and adaptability of the framework.
	
	Nevertheless, limitations remain. The generality of BLOC necessarily entails a computational trade-off: for standard smooth objectives with effective problem-specific solvers, specialized methods may achieve comparable or better solutions at substantially lower computational cost. BLOC is therefore most advantageous when such bespoke optimization machinery is unavailable or difficult to construct, particularly for nonconvex, nonsmooth, discontinuous, or black-box objectives and customized penalty structures. Future work may focus on accelerated implementations and adaptive strategies that reduce computational overhead while retaining BLOC's flexibility and exploratory capability. Overall, the proposed framework complements, rather than replaces, specialized covariance estimators by providing a general and parallelizable optimization strategy for problems that fall outside their usual scope.

	\if1\anon
	{
		\textbf{Data Availability Statement }
		The dataset used in this study is publicly available at \href{https://github.com/priyamdas2/BLOC}{https://github.com/priyamdas2/BLOC}.\vspace{-0.4cm}
	} \fi
	
	\if0\anon
	{
		\textbf{Data Availability Statement}
		The dataset used in this study will be shared publicly on Github.\vspace{-0.4cm}
	} \fi
	
	
	\if1\anon
	{
		\begin{description}
						\item[BLOC software:] 
			Reproducible code for the simulations and real data analysis is available on GitHub at \href{https://github.com/priyamdas2/BLOC}{https://github.com/priyamdas2/BLOC}. \vspace{-0.3cm}
			\item[Funding:]
			\begin{description}
				\item PD is partially supported by NIH-NCI Cancer Center Support Grant P30 CA016059.
			\end{description}
		\end{description}
	} \fi
	
	\if0\anon
	{
		\begin{description}
			\item[Supplementary text:]
			Supplementary material is provided as a separate pdf document.
			\item[BLOC software:]
			Reproducible code for the simulations and real data analysis will be made available on GitHub.
		\end{description}
	} \fi
\newpage
\appendix
\begin{center}
	\Large Supplement to ``BLOC: A Global Optimization Framework for Sparse Covariance Estimation with Non-Convex Penalties"
\end{center}	

\renewcommand{\thetable}{S\arabic{table}}
\renewcommand{\thefigure}{S\arabic{figure}}
\renewcommand{\thesection}{\Alph{section}}
\setcounter{algorithm}{0}
\renewcommand{\thealgorithm}{S\arabic{algorithm}}
\newtheorem{supplemma}{Lemma}
\newtheorem{supptheorem}{Theorem}
\renewcommand{\thesupplemma}{S\arabic{supplemma}}
\renewcommand{\thesupptheorem}{S\arabic{supptheorem}}
\renewcommand{\thetable}{S\arabic{table}}
\renewcommand{\thefigure}{S\arabic{figure}}
\renewcommand{\thesection}{\Alph{section}}
\setcounter{algorithm}{0}
\renewcommand{\thealgorithm}{S\arabic{algorithm}}
	\section{Global Optimization: Background}\label{sec:global_opt}
	Optimization has long been central to numerous disciplines including Mathematics, Statistics, and Computer Science. Classical approaches have been dominated by convex optimization methods, celebrated for their strong convergence properties and scalability across large parameter spaces \citep{Boyd2004}. Techniques such as Newton-Raphson, interior-point (IP), and sequential quadratic programming (SQP) continue to be widely applied \citep{Nocedal2006}. However, in settings where the objective function is non-convex, which is often the case in robust or penalized covariance estimation, these techniques can be limited by their inability to escape local optima.
	
	Advancements in computing power have led to increased use of global optimization methods in tackling non-convex problems. Metaheuristic algorithms such as Genetic Algorithms (GA; \citealp{Bethke1980}) and Simulated Annealing (SA; \citealp{Kirkpatrick1983}) are designed to explore complex landscapes and offer a probabilistic guaranty of finding global solutions. These heuristics employ mechanisms that help escape local optima, thereby increasing the chances of identifying global solutions. In the following, we highlight the fundamental differences between convex and global optimization through the lens of the \textit{exploration versus exploitation} paradigm.
	
	A key conceptual framework for understanding global optimization is the \textit{exploration vs. exploitation} trade-off. This idea, which is foundational across reinforcement learning and search theory, applies equally to numerical optimization \citep{BergerTal2014, Zhang2023}. In this context, gradient-based methods typically focus on \textit{exploitation}, refining solutions through descent directions, while full-grid search methods emphasize \textit{exploration}, evaluating a broad swath of the space. Modern global optimization algorithms strive to combine these two principles, which can be conceptually expressed as:
	\[
	\gamma \cdot \textit{exploitation} + (1 - \gamma) \cdot \textit{exploration}, \quad \gamma \in (0,1)
	\]
	Here, $\gamma$ quantifies the algorithm’s emphasis on local refinement versus global search. Despite significant advances, scalability remains a major challenge. For instance, in high-dimensional settings, search complexity grows exponentially for genetic algorithms (GAs), limiting their applicability \citep{Geris2012}. Furthermore, only a limited number of optimization methods are specifically designed for symmetric positive definite or correlation matrix spaces, and to the best of our knowledge, very few of them (if any) adhere to a global optimization framework that incorporates strategies to escape local optima. Although certain algorithms, such as Barzilai-Borwein, conjugate gradient, steepest descent, and trust-region methods, can be adapted to operate over such constrained matrix manifolds \citep{boumal2014manopt}, they typically lack mechanisms for escaping local solutions. Consequently, they are not well-suited for solving non-convex objective functions in these spaces, a limitation that is further substantiated in our benchmark study presented in Section \ref{sec:benchmark}.
	
	\paragraph{Recursive Modified Pattern Search:} One of the earliest attempts at derivative-free black-box optimization was pioneered by \cite{Fermi1952}, introducing a strategy commonly referred to as \textit{Fermi’s principle}. This method allows for optimizing an objective function over an unconstrained domain, even when the function is non-differentiable or discontinuous. At each iteration, the function is evaluated at \(2n\) neighboring points, corresponding to coordinate-wise movements in both positive and negative directions with a step-size denoted by \(s > 0\). The best-performing point is selected as the updated solution. By adjusting \(s\), candidate points can be sampled either from a local neighborhood (small \(s\)) or from more distant regions (large \(s\)), thus enabling adaptive exploration. 
	\textit{Fermi’s principle} eventually gave rise to methods such as Direct Search and Pattern Search (PS) algorithms \citep{Torczon1997, Kolda2003}. When employing Fermi’s principle, progressively reducing the step-size across iterations can lead to local convergence as \( s \to 0 \) \citep{Torczon1997, Das2023RMPSH}. Unlike gradient-based algorithms, PS methods incorporate a substantial but scalable degree of exploration over the parameter space. This exploration-driven approach makes PS a powerful and scalable tool for derivative-free optimization. However, due to the absence of any built-in mechanism to escape local minima, PS is typically not categorized under meta-heuristic global optimization strategies such as Genetic Algorithms (GA) and Simulated Annealing (SA).
	
	In order to extend the PS principle to global optimization problems, \cite{Das2023RMPSH} proposed the \textit{Recursive Modified Pattern Search (RMPS)} algorithm by modifying the classical PS method of \cite{Torczon1997} in two key ways: (i) incorporating an adaptive strategy to shrink the step-size $s \to 0$ based on improvements in the objective function, and (ii) introducing a restart mechanism that allows the algorithm to escape local optima. These enhancements enabled RMPS to outperform Genetic Algorithms (GA) and Simulated Annealing (SA) in benchmark studies, typically yielding superior solutions in less time. RMPS has since been extended to constrained domains such as the unit sphere \citep{Das2022} and the simplex \citep{Das2021}. It has also been applied to non-convex statistical problems in both frequentist and Bayesian frameworks, across domains such as finance \citep{tan2020}, environmental studies \citep{Das2017b}, and biomedical informatics \citep{Das2023, kim2025smart}. RMPS has also been shown to outperform conventional algorithms, such as the Expectation–Maximization (EM) algorithm, in mixture Markov clustering problems \citep{Das2023MsiCOR}.
	\section{BLOC algorithm}\label{sec:BLOC_algo}
	The complete BLOC procedure is summarized in Algorithm~\ref{alg:bloc}.
	\begin{algorithm}[!h]
		\caption{BLOC}\label{alg:bloc}
		{\footnotesize
			\noindent\textbf{Input:} Initial correlation matrix $\bm{\Gamma}^{(0)} \in \mathcal{C}_d$\\
			\textbf{Output:} Optimized solution $\widehat{\bm{\Gamma}} \in \mathcal{C}_d$
		}
		\begin{algorithmic}[1]
			\fontsize{8}{1}
			\STATE \textbf{Initialization:} $R \gets 1$ \hfill {\tiny ($R$ = run index)}
			\STATE \emph{top}:
			\IF {$R = 1$}
			\STATE $\widehat{\bm{\varphi}}^{(0)} \gets \textit{Corr2Phi}(\bm{\Gamma}^{(0)})$ {\tiny \bigg($\widehat{\bm{\varphi}}^{(r)}$ denotes the value of $\bm{\varphi}$ at the end of $r$-th \textit{run};
				\textit{Corr2Phi}$()$ converts correlation matrix in 
				\textcolor{white}{fill fill fill fill  fill fill fill fill fill fill fill fill fill}
				$\mathcal{C}_d$ to its polar coordinate representation on $\mathcal{A}_d$\bigg)}
			\STATE $\bm{\varphi}^{(0)} \gets \widehat{\bm{\varphi}}^{(0)},\;j \gets 1$ {\tiny \bigg($\bm{\varphi}^{(j)}$ denotes the value of $\bm{\varphi}$ at the end of $j$-th iteration inside current \textit{run}\bigg)}
			\ELSE
			\STATE $\bm{\varphi}^{(0)} \gets \widehat{\bm{\varphi}}^{(R-1)},\;j \gets 1$
			\ENDIF
			\STATE $s^{(0)} \gets s_{initial}$ {\tiny \bigg($s^{(j)}$ denotes the value of \textit{global step-size} at the end of $j$-th iteration inside current \textit{run}\bigg)}
			\WHILE {($j \leq max\_iter$ and $s^{(j-1)} > \kappa$)} 
			\STATE $F_1 \gets f(\Phi(\bm{\varphi}^{(j-1)}))$, $s^{(j)} \gets s^{(j-1)}$ {\tiny \bigg(note that, $\bm{\varphi}^{(j-1)} = (\varphi_1^{(j-1)}, \ldots, \varphi_N^{(j-1)})$\bigg)}
			\FOR {$h = 1:2N$} {\tiny \bigg(can be distributed among up to $2N$ parallel cores\bigg)}
			\STATE $i \gets [\frac{(h+1)}{2} ]$ {\tiny \bigg($[\cdot]$ denotes largest smaller integer function; $i$ denotes the vector location\bigg)}
			\STATE $\bm{\varphi}_h \gets \bm{\varphi}^{(j-1)}${\tiny\bigg(note that, 
				$\bm{\varphi}_h = \big\{\bm{\varphi}_h(i): 1 \leq i \leq N\big\}$\bigg)} 
			\STATE $ s_h \gets (-1)^hs^{(j)}$
			\STATE $\bm{\varphi}_h(i) \gets \mathscr{M}^i(\bm{\varphi}_h(i) + s_h)$ {\tiny \bigg(this ensures updated $\bm{\varphi}_h$ remains in $\widetilde{\mathcal A}_d$\bigg)}
			\STATE $f_h \gets f(\Phi(\bm{\varphi}_h))$
			\ENDFOR
			\STATE $h_{best} \gets \argmin_h f_h$ {\tiny over} $h = 1,\ldots,2N$
			\STATE $\bm{\varphi}_{temp} \gets \bm{\varphi}_{h_{best}}$
			\STATE $F_2 \gets f_{h_{best}}$
			\STATE $\bm{\varphi}^{(j)} \gets \bm{\varphi}^{(j-1)}$
			\IF {($F_2 < F_1$)} $\bm{\varphi}^{(j)} \gets \bm{\varphi}_{temp}$ 
			\ENDIF
			\IF {($j > 1$) and $|F_1-min(F_1,F_2)| < \tau_1$ and $s^{(j)}>\kappa$)} 
			\STATE $s^{(j)} \gets \frac{s^{(j)}}{\rho}$
			\ENDIF
			\STATE $j \gets j+1$
			\ENDWHILE
			\STATE $\widehat{\bm{\varphi}}^{(R)} \gets \bm{\varphi}^{(j-1)}$ 
			\IF {$R \ge max\_run$ or $|f(\Phi(\widehat{\bm{\varphi}}^{(R)})) - f(\Phi(\widehat{\bm{\varphi}}^{(R-1)}))| < \tau_2$ }
			\STATE $\widehat{\bm{\varphi}} =\widehat{\bm{\varphi}}^{(R)}$ {\tiny \big($\widehat{\bm{\varphi}}$ denotes BLOC optimized final solution in $\widetilde{\mathcal A}_d$\big)}
			\STATE \textbf{Return} $\widehat{\bm{\Gamma}} = \textit{Phi2Corr}(\widehat{\bm{\varphi}})$ {\tiny \bigg(\textit{Phi2Corr}$()$ maps back polar representation on $\widetilde{\mathcal A}_d$ to corresponding correlation matrix in $\mathcal{C}_d$\bigg)}
			\STATE \textbf{break} {\tiny \big(exiting BLOC\big)}
			\ELSE
			\STATE $R \gets R+1$
			\STATE \textbf{go to} \emph{top}
			\ENDIF 
		\end{algorithmic}
	\end{algorithm}

	\section{Benchmark study}
	\label{sec:benchmark}
	\noindent To assess the performance of BLOC on structured, high-dimensional constrained domains, we employ a set of four widely used benchmark functions, Ackley, Griewank, Rosenbrock, and Rastrigin, adapted to the space of correlation matrices. For a correlation matrix $\boldsymbol{C} \in \mathcal{C}_d$, we vectorize its off-diagonal elements into $\boldsymbol{x} \in \mathbb{R}^{d(d-1)}$. In their classical unconstrained forms, all four benchmarks achieve a global minimum of zero at the origin. In our adaptation, the identity correlation matrix plays the role of the origin, making it the global minimizer with objective value zero. This ensures that optimization quality can be judged relative to a common reference point while partially preserving the intrinsic landscape of each benchmark. The Ackley function, applied to inputs $x_i = 10 \cdot C_{pq}$ ($p \ne q$), is characterized by flat basins interspersed with oscillations and tests the optimizer’s ability to escape plateaus. The Griewank function, with inputs scaled as $x_i = 100 \cdot C_{pq}$, combines quadratic growth with multiplicative cosine modulations, introducing moderate multimodality. The Rastrigin function, evaluated on $x_i = 10 \cdot C_{pq}$, creates a highly repetitive landscape of local minima, stressing robustness in oscillatory terrains. Finally, the Rosenbrock function, applied to $x_i = 100 \cdot C_{pq}$, presents a curved, narrow valley that requires careful search along nonlinear dependencies. Embedding these functions into the space of positive definite correlation matrices preserves the structural challenges of the original objectives while respecting the unit-diagonal and positive definiteness constraints. For precise definitions and classical domains of these functions, see \cite{surjanovic2013virtual}. Table~\ref{tab:bloc_benchmark} summarizes the comparative performance of BLOC and its parallel implementation against a suite of optimization methods for positive definite matrix problems. Competing solvers include MATLAB’s \texttt{fmincon} variants (Interior-Point, SQP, Active-Set) as well as routines from the \texttt{Manopt} toolbox \citep{boumal2014manopt}. A maximum runtime of one hour is imposed for each solver. For fairness, all methods, including both BLOC variants, are initialized from the same 10 randomly chosen starting points per problem instance. The table reports the best objective value across replications, the associated standard error, and the mean runtime (seconds, with standard error in parentheses), for correlation matrices of dimension $d = 5, 10, 20, 50$. All experiments were conducted on a Windows 10 Enterprise workstation with 32 GB RAM and a 12th Gen Intel(R) Core(TM) i7-12700 processor (2.1 GHz, 12 cores, 20 threads).
	
	The comparative results in Table~\ref{tab:bloc_benchmark} show that BLOC is consistently competitive, often achieving the global minimum or values close to machine precision across all four benchmark functions. Among the other solvers, \texttt{fmincon:sqp} and \texttt{fmincon:active-set} perform reasonably well in lower dimensions, converging quickly on Ackley and Griewank with modest runtimes. However, their solution quality deteriorates as dimensionality increases, and they frequently fail to approach the global minimum on Rosenbrock and Rastrigin. Similarly, some \texttt{Manopt} routines, such as conjugate-gradient and steepest-descent, occasionally deliver acceptable solutions for $d=5$ or $d=10$, but their progress slows sharply with higher dimensions and they often plateau far above the optimum or exhaust the one-hour time budget. Overall, while a few alternatives show reasonable performance in small-scale problems, BLOC demonstrates more consistent accuracy and stability as the correlation matrix size increases.
	
	\begin{landscape}
		\begin{table}[h]
			\centering
			\renewcommand{\arraystretch}{1.15}
			\resizebox{0.99\columnwidth}{!}{\begin{tabular}{|c|l|ccc|ccc|ccc|ccc|}
					\hline
					\multirow{2}{*}{Functions} & \multirow{2}{*}{Methods} & \multicolumn{3}{c|}{$d = 5$} & \multicolumn{3}{c|}{$d = 10$} & \multicolumn{3}{c|}{$d = 20$} & \multicolumn{3}{c|}{$d = 50$} \\ \cline{3-14} 
					&  & min. value & s.e. of values & mean time (s.e.) & min. value & s.e. of values & mean time (s.e.) & min. value & s.e. of values & mean time (s.e.) & min. value & s.e. of values & mean time (s.e.) \\ \hline
					\multirow{9}{*}{Ackley} & BLOC & \textbf{4.44E-16} & 1.19E-15 & 0.21 (0.017) & \textbf{4.00E-15} & 3.38E-15 & 6.20 (0.333) & \textbf{2.18E-14} & 3.32E-15 & 690.14 (29.384) & \textbf{4.21E-06} & 6.64E-07 & 3600.79 (0.126) \\
					& BLOCparallel & \textbf{4.44E-16} & 1.19E-15 & 14.85 (1.045) & \textbf{4.00E-15} & 3.38E-15 & 79.84 (4.232) & \textbf{2.18E-14} & 3.32E-15 & 465.38 (22.530) & \textbf{8.57E-14} & 3.45E-15 & 2798.00 (127.924) \\
					& fmincon:active-set & 4.77E+00 & 6.92E-01 & 0.10 (0.038) & 2.71E+00 & 3.26E-01 & 0.25 (0.038) & 2.24E+00 & 1.97E-01 & 1.23 (0.088) & 2.58E+00 & 8.08E-02 & 14.84 (1.938) \\
					& fmincon:interior-point & 7.62E+00 & 5.15E-01 & 0.11 (0.031) & 2.76E+00 & 2.95E-01 & 0.22 (0.033) & 3.65E+00 & 1.25E-01 & 0.30 (0.031) & 4.12E+00 & 1.75E-01 & 0.74 (0.023) \\
					& fmincon:sqp & 2.32E+00 & 6.65E-01 & 0.03 (0.005) & 1.57E-01 & 3.47E-01 & 0.21 (0.008) & 2.11E+00 & 1.88E-01 & 3.41 (0.194) & 2.75E+00 & 6.34E-02 & 36.28 (7.507) \\
					& Manopt:barzilai-borwein & 2.59E-04 & 1.56E+00 & 1.20 (0.160) & 2.19E+00 & 2.09E-01 & 12.90 (0.106) & 2.74E+00 & 1.91E-01 & 1845.61 (16.090) & 2.91E+00 & 4.25E-02 & 3655.08 (9.139) \\
					& Manopt:conjugate-gradient & 3.40E+00 & 9.46E-01 & 0.34 (0.141) & 5.33E+00 & 4.60E-01 & 10.06 (1.096) & 2.59E+00 & 3.14E-01 & 1473.70 (35.334) & 2.75E+00 & 2.94E-02 & 3628.97 (8.379) \\
					& Manopt:steepest-descent & 2.01E+00 & 1.11E+00 & 0.90 (0.188) & 2.25E+00 & 5.86E-01 & 12.29 (0.647) & 1.37E+00 & 9.70E-02 & 1562.16 (84.209) & 3.92E-02 & 1.94E-01 & 3652.13 (12.547) \\
					& Manopt:trust-region & 8.58E+00 & 5.89E-01 & 3.41 (0.093) & 5.27E+00 & 4.55E-01 & 125.19 (57.512) & 2.13E+00 & 2.51E-01 & 3612.68 (5.701) & 2.12E+00 & 7.16E-02 & 3739.28 (22.810) \\ \hline
					\multirow{9}{*}{Griewank} & BLOC & \textbf{0.00E+00} & 3.56E-02 & 0.11 (0.014) & \textbf{0.00E+00} & 9.02E-02 & 2.61 (0.118) & \textbf{0.00E+00} & 6.57E-02 & 368.09 (22.019) & 1.79E+00 & 9.39E-02 & 3600.80 (0.145) \\
					& BLOCparallel & \textbf{0.00E+00} & 3.56E-02 & 7.12 (0.899) & \textbf{0.00E+00} & 9.02E-02 & 32.87 (1.511) & \textbf{0.00E+00} & 6.57E-02 & 168.62 (8.821) & \textbf{1.33E-15} & 3.27E-02 & 1453.56 (26.140) \\
					& fmincon:active-set & 3.54E-08 & 1.30E-01 & 0.06 (0.014) & 9.67E-01 & 1.50E-02 & 0.49 (0.048) & 1.12E+00 & 4.90E-02 & 8.87 (0.684) & 1.05E+00 & 8.26E-02 & 711.20 (49.887) \\
					& fmincon:interior-point & 3.62E-14 & 1.34E-01 & 0.09 (0.008) & 5.31E-09 & 1.14E-01 & 0.17 (0.003) & 8.66E-02 & 1.00E-01 & 0.26 (0.007) & 1.58E+01 & 6.60E-01 & 0.75 (0.012) \\
					& fmincon:sqp & 5.24E-12 & 1.35E-01 & 0.02 (0.001) & 2.06E-04 & 1.03E-01 & 0.16 (0.018) & 7.09E-03 & 9.69E-02 & 5.70 (2.323) & 8.09E-02 & 7.88E-02 & 523.22 (129.671) \\
					& Manopt:barzilai-borwein & 1.19E-01 & 1.22E-01 & 1.30 (0.155) & 1.22E-11 & 1.29E-01 & 12.23 (0.828) & 3.99E-10 & 9.02E-02 & 1645.40 (33.268) & 1.18E+00 & 1.63E-01 & 3651.76 (12.201) \\
					& Manopt:conjugate-gradient & 3.02E-13 & 1.09E-01 & 0.20 (0.079) & 1.51E-12 & 1.13E-01 & 7.40 (1.619) & 4.14E-12 & 7.48E-07 & 1006.67 (119.258) & \textbf{2.80E-02} & 9.04E-02 & 3648.72 (8.335) \\
					& Manopt:steepest-descent & 3.65E-13 & 1.48E-01 & 1.03 (0.205) & 3.20E-12 & 1.03E-01 & 11.63 (1.040) & 1.55E-04 & 7.65E-02 & 1412.82 (8.442) & 6.05E-01 & 5.84E-02 & 3666.81 (8.680) \\
					& Manopt:trust-region & 6.12E-14 & 1.40E-01 & 2.95 (0.459) & 8.57E-13 & 1.03E-01 & 88.85 (33.392) & 6.85E-11 & 5.54E-02 & 3424.90 (183.014) & 4.03E+00 & 3.86E-01 & 3977.01 (71.034) \\ \hline
					\multirow{9}{*}{Rosenbrock} & BLOC & \textbf{7.82E-22} & 2.51E+00 & 0.31 (0.021) & \textbf{4.48E-20} & 1.35E+01 & 18.48 (1.917) & \textbf{3.75E+02} & 6.49E-13 & 1031.01 (60.384) & 1.51E+05 & 3.21E+05 & 3600.54 (0.109) \\
					& BLOCparallel & \textbf{7.82E-22} & 2.51E+00 & 24.06 (1.642) & \textbf{4.48E-20} & 1.35E+01 & 251.39 (25.812) & \textbf{3.75E+02} & 4.43E-01 & 3600.29 (0.022) & \textbf{2.75E+03} & 3.79E+02 & 3600.13 (0.021) \\
					& fmincon:active-set & 6.73E+03 & 2.45E+05 & 0.09 (0.008) & 1.17E+04 & 9.86E+03 & 2.01 (0.041) & 9.96E+04 & 2.02E+04 & 377.28 (5.830) & --- & --- & --- \\
					& fmincon:interior-point & 2.68E-08 & 1.88E+00 & 0.11 (0.020) & 7.63E+02 & 7.18E+04 & 0.14 (0.002) & 4.13E+08 & 2.30E+08 & 0.20 (0.004) & 1.01E+09 & 2.36E+09 & 0.76 (0.030) \\
					& fmincon:sqp & 1.52E+02 & 2.43E+06 & 0.03 (0.002) & 1.61E+02 & 6.12E+05 & 0.87 (0.176) & 2.03E+03 & 4.59E+05 & 202.14 (18.788) & --- & --- & --- \\
					& Manopt:barzilai-borwein & 4.05E+09 & 4.25E+09 & 0.01 (0.001) & 1.63E+10 & 1.76E+09 & 0.02 (0.001) & 1.98E+10 & 1.56E+09 & 1.60 (0.044) & 2.48E+10 & 7.10E+08 & 119.75 (2.236) \\
					& Manopt:conjugate-gradient & 1.88E+01 & 9.28E-10 & 0.15 (0.017) & 8.81E+01 & 3.71E-02 & 8.03 (1.743) & \textbf{3.75E+02} & 4.29E-02 & 901.03 (174.999) & 8.06E+03 & 2.64E+03 & 3654.56 (8.325) \\
					& Manopt:steepest-descent & 1.88E+01 & 7.41E-10 & 0.38 (0.128) & 8.81E+01 & 5.23E+00 & 10.31 (1.460) & \textbf{3.75E+02} & 2.71E+00 & 958.70 (94.386) & \textbf{6.34E+03} & 1.69E+03 & 3643.65 (8.871) \\
					& Manopt:trust-region & 1.34E-09 & 1.88E+00 & 3.25 (0.019) & 8.81E+01 & 7.86E-04 & 66.38 (18.825) & \textbf{3.75E+02} & 1.38E-01 & 3255.01 (105.917) & 4.31E+07 & 2.78E+07 & 3926.85 (55.487) \\ \hline
					\multirow{9}{*}{Rastrigin} & BLOC & \textbf{1.99E+00} & 7.41E+00 & 0.13 (0.011) & \textbf{5.77E+01} & 1.81E+01 & 16.04 (9.370) & \textbf{1.87E+02} & 4.05E+01 & 2162.80 (332.831) & \textbf{2.09E+03} & 2.16E+02 & 3600.91 (0.199) \\
					& BLOCparallel & \textbf{1.99E+00} & 7.41E+00 & 9.65 (0.654) & \textbf{5.77E+01} & 1.81E+01 & 191.22 (111.360) & \textbf{1.87E+02} & 4.18E+01 & 2946.87 (347.088) & \textbf{1.02E+03} & 9.07E+01 & 3600.17 (0.028) \\
					& fmincon:active-set & 5.97E+01 & 5.47E+01 & 0.07 (0.025) & 5.69E+02 & 4.79E+01 & 1.55 (0.079) & --- & --- & --- & --- & --- & --- \\
					& fmincon:interior-point & 1.53E+02 & 4.04E+01 & 0.09 (0.021) & 3.78E+02 & 6.08E+01 & 0.21 (0.031) & 2.11E+03 & 1.25E+02 & 0.43 (0.047) & 1.95E+04 & 7.45E+02 & 1.00 (0.045) \\
					& fmincon:sqp & 4.78E+01 & 2.55E+01 & 0.02 (0.002) & 3.86E+02 & 2.92E+01 & 0.55 (0.017) & --- & --- & --- & --- & --- & --- \\
					& Manopt:barzilai-borwein & \textbf{1.99E+00} & 1.54E+01 & 1.66 (0.072) & \textbf{2.39E+01} & 1.09E+01 & 13.79 (0.240) & \textbf{2.24E+02} & 2.26E+02 & 1462.75 (128.662) & 1.94E+04 & 1.08E+03 & 3639.13 (9.707) \\
					& Manopt:conjugate-gradient & 2.37E+02 & 5.28E+01 & 0.32 (0.122) & 7.32E+02 & 2.91E+01 & 11.20 (1.175) & 1.67E+03 & 3.57E+01 & 1460.88 (15.544) & 6.44E+03 & 8.51E+01 & 3653.67 (9.825) \\
					& Manopt:steepest-descent & 1.93E+02 & 4.66E+01 & 1.08 (0.200) & 6.27E+02 & 4.69E+01 & 13.36 (0.056) & 1.59E+03 & 4.92E+01 & 1414.58 (19.622) & 8.82E+03 & 2.05E+02 & 3669.30 (11.029) \\
					& Manopt:trust-region & 2.37E+02 & 5.38E+01 & 3.50 (0.110) & 7.32E+02 & 2.86E+01 & 87.81 (40.669) & 1.61E+03 & 2.70E+01 & 3613.74 (5.821) & 8.70E+03 & 1.62E+02 & 3823.16 (35.164) \\ \hline
			\end{tabular}}
			\caption{Performance comparison of BLOC and its parallel implementation against baseline optimization algorithms across four (modified) benchmark objective functions (Ackley, Griewank, Rosenbrock, and Rastrigin), evaluated over increasing correlation matrix dimensions ($d \times d$ for $d = 5$, 10, 20, 50). For each setting, results are based on 10 replications initialized from the same set of randomly chosen starting points across all methods. The reported minimum objective value is the best among the 10 replications, while the mean runtime (in seconds) and standard error (in parentheses) are computed across those 10 replications.}
			\label{tab:bloc_benchmark}
		\end{table}
	\end{landscape}
	
	A particularly important distinction emerges when comparing BLOC with its parallelized implementation. For $d \leq 10$, the sequential BLOC is more efficient due to its minimal overhead, solving problems in fractions of a second to a few minutes while reaching near-optimal values. In these smaller settings, BLOC\textsc{parallel} incurs additional synchronization costs without clear gains in objective quality. By contrast, at $d=20$ and especially $d=50$, the benefits of parallelization become decisive. On Ackley and Griewank, BLOC\textsc{parallel} reduces runtimes by factors of two to three while still attaining machine-precision minima, whereas sequential BLOC slows down and often terminates with noticeably worse objective values. The improvement is also evident on Rosenbrock and Rastrigin, where BLOC\textsc{parallel} attains substantially lower minima within the allotted runtime, outperforming both sequential BLOC and all other solvers. 
	\begin{table}
		\centering
		\resizebox{0.6\columnwidth}{!}{%
			\begin{tabular}{|c|cccl|}
				\hline
				Functions & min. value & s.e. of values & \multicolumn{2}{c|}{mean time (s.e.)} \\ \hline
				Ackley & 4.46E-10 & 5.94E-02 & \multicolumn{2}{c|}{18000.47 (0.073)} \\
				Griewank & 4.98E-12 & 1.25E-01 & \multicolumn{2}{c|}{18000.47 (0.088)} \\
				Rosenbrock & 4.02E+04 & 1.90E+06 & \multicolumn{2}{c|}{18000.41 (0.091)} \\
				Rastrigin & 4.84E+03 & 4.14E+02 & \multicolumn{2}{c|}{18000.69 (0.171)} \\ \hline
		\end{tabular}}
		\caption{Performance of the parallel BLOC implementation at $d=100$ with a maximum runtime of five hours per benchmark function. Reported are the best objective value across 10 replications, the associated standard error, and the mean runtime in seconds (with standard error in parentheses).}
		\label{tab:comp_time_highdim} 
	\end{table}
	
	To further probe scalability, we conducted an additional experiment with $d=100$, restricting attention to BLOC\textsc{parallel} with a maximum runtime of five hours per function. Table~\ref{tab:comp_time_highdim} reports the results. On Ackley and Griewank, BLOC\textsc{parallel} successfully drives objective values to near machine-precision minima even at this scale, demonstrating its ability to handle extremely high-dimensional correlation matrices. For Rosenbrock and Rastrigin, the solutions are less tight than for Ackley and Griewank, which is expected given the increased complexity of their landscapes in high dimensions. Even so, BLOC\textsc{parallel} makes clear progress within the allotted budget and delivers results that are competitive relative to what is typically observed on these notoriously difficult benchmarks. Importantly, the method remained numerically stable and converged reliably across all 10 replications, confirming its robustness and reliability under demanding settings.
	
	Taken together, the benchmark studies demonstrate that BLOC is a competitive solver across a wide range of structured nonconvex landscapes, consistently outperforming existing alternatives as the problem size grows. The non-parallel variant is effective in low-dimensional settings, but BLOC\textsc{parallel} becomes the method of choice for moderate to large dimensions, balancing accuracy and runtime efficiency. These findings highlight BLOC’s versatility and scalability, and they establish BLOC\textsc{parallel} as a practical optimizer for high-dimensional correlation matrix problems.
	\section{Additional details on Simulation studies}
	\label{sec:simulation_supp}
	\subsection{Sparse covariance estimation with Gaussian likelihood ($n>d$)}
	\label{sec:sims_small_p}
	\paragraph{LLA--MM implementation.}
	For the SCAD- and MCP-based comparator, we implemented a local-linear-approximation majorization--minimization (LLA--MM) procedure, motivated by the LLA strategy underlying I-LAMM \citep{Fan2018ILAMM}, to optimize the same penalized Gaussian correlation objective used by BLOC. Let $\widehat{\bm\Gamma}_{S}=\bm X^\top\bm X/n$ and let $\bm\Gamma^{(t)}$ denote the current iterate. The objective is
	$\log|\bm\Gamma|+\operatorname{tr}(\widehat{\bm\Gamma}_{S}\bm\Gamma^{-1})
	+\sum_{i\ne j}p_{\lambda}(|\gamma_{ij}|)$,
	subject to $\operatorname{diag}(\bm\Gamma)=\bm 1$ and positive definiteness, with the off-diagonal penalty counted using exactly the same convention as in BLOC. At iteration $t$, concavity of $\log|\bm\Gamma|$ gives the tangent majorization
	$\log|\bm\Gamma|\leq \mathrm{const}+
	\operatorname{tr}\{(\bm\Gamma^{(t)})^{-1}\bm\Gamma\}$,
	while the nonconvex penalty is locally linearized using weights
	$w_{ij}^{(t)}=p'_{\lambda}(|\gamma_{ij}^{(t)}|)$.
	Thus, each update is obtained by solving the convex surrogate
	\[
	\min_{\bm\Gamma,\bm Y}\;
	\operatorname{tr}\{(\bm\Gamma^{(t)})^{-1}\bm\Gamma\}
	+\operatorname{tr}(\bm Y)
	+\sum_{i\ne j}w_{ij}^{(t)}|\gamma_{ij}|,
	\]
	subject to $\operatorname{diag}(\bm\Gamma)=\bm 1$,
	$\bm\Gamma\succeq 10^{-4}\bm I_d$, and
	$\left[\begin{smallmatrix}\bm\Gamma&\bm B\\
		\bm B^\top&\bm Y\end{smallmatrix}\right]\succeq0$,
	where $\widehat{\bm\Gamma}_{S}=\bm B\bm B^\top$; the Schur-complement constraint provides a convex representation of
	$\operatorname{tr}(\widehat{\bm\Gamma}_{S}\bm\Gamma^{-1})$.
	The resulting semidefinite programs were solved using \texttt{CVX}. We initialized the procedure at $\bm I_d$ and first performed an $\ell_1$-MM contraction stage with $w_{ij}=\lambda$ for all penalized off-diagonal entries, followed by SCAD or MCP LLA reweighting. Iterations were terminated when both the relative Frobenius change in $\bm\Gamma$ and the relative change in the original nonconvex objective were below $10^{-5}$, with at most 10 initialization iterations and 20 LLA--MM iterations.

	\paragraph{Simulation design.} To evaluate performance under different sparsity structures, we generated true correlation matrices according to two designs, with $d$ variables, $n$ samples, and $10$ data replications per setting. 
	\\[1ex]
	\emph{(i) Block-diagonal:} The dimension $d$ is partitioned into disjoint $5\times 5$ blocks along the diagonal, resulting in a block-diagonal structure with $d/5$ such blocks. Each block is an independently generated random correlation matrix. The blocks are embedded along the diagonal to form the full $d\times d$ correlation matrix, yielding a clustered sparsity pattern (within-block correlations, zeros elsewhere). Positive definiteness holds by construction.
	\\[1ex]
	\emph{(ii) Uniform-sparse:} We prescribe a fixed off-diagonal sparsity level depending on $(d,n)$, and then populate a random subset of nonzero entries with values drawn i.i.d.\ from $\mathrm{Uniform}[0.3,\,0.6]$, followed by symmetrization. To ensure positive definiteness, each generated matrix is checked via a Cholesky decomposition, and regeneration is repeated if necessary. The specific sparsity levels are: $0.95$ zeros for $(d,n)=(20,50)$, $0.98$ zeros for $(50,100)$, and $0.99$ zeros for $(100,150)$.
	
	For each setting, we generated ten independent datasets by drawing $n$ observations from a mean-zero multivariate normal distribution with the true correlation matrix specified by either the block-diagonal or the uniform-sparse design. This produced a consistent collection of replications across all $(d,n)$ configurations. The block-diagonal design reflects localized clusters of correlation, while the uniform-sparse design tests robustness to dispersed and irregular sparsity under positive-definiteness constraints. Table~1 in the main manuscript reports results for $(d,n) \in \{(20,50), (50,100), (100,150)\}$. 
	
	\begin{table}[!h]
		\centering
		\resizebox{0.99\columnwidth}{!}{%
			\begin{tabular}{|l|cc|cc|cc|}
				\hline
				\multirow{2}{*}{Methods} & \multicolumn{2}{c|}{$d=20, n=50$} & \multicolumn{2}{c|}{$d=50, n=100$} & \multicolumn{2}{c|}{$d=100, n=150$} \\ \cline{2-7} 
				& \begin{tabular}[c]{@{}c@{}}Block\\ diagonal\end{tabular} & \begin{tabular}[c]{@{}c@{}}Uniform\\ sparse\end{tabular} & \begin{tabular}[c]{@{}c@{}}Block\\ diagonal\end{tabular} & \begin{tabular}[c]{@{}c@{}}Uniform\\ sparse\end{tabular} & \begin{tabular}[c]{@{}c@{}}Block\\ diagonal\end{tabular} & \begin{tabular}[c]{@{}c@{}}Uniform\\ sparse\end{tabular} \\ \hline
				BLOC-MCP & 5282 (83.2) & 14506 (81.2) & 30901 (5586.4) & 37945 (3006.9) & 121953 (17194.5) & 59913 (840.1) \\
				BLOC-SCAD & 5057 (132.9) & 15281 (155.5) & 42896 (2205.4) & 32320 (5856.4) & 128687 (16725.3) & 53085 (955.6) \\
				LLA-MM-MCP & 696 (10.9) & 674 (5.4) & 9610 (116.7) & 8306 (51.4) & \textgreater  72 hrs & \textgreater  72 hrs \\
				LLA-MM-SCAD & 730 (6.3) & 699 (8.8) & 12021 (94.6) & 13657 (328.0) & \textgreater  72 hrs & \textgreater  72 hrs \\
				SPCOV & 242 (10.9) & 14 (1.2) & 845 (70.7) & 809 (49.6) & - & 4183 (144.1) \\ \hline
		\end{tabular}}
		\caption{Computation times (seconds) for sparse covariance estimation with Gaussian likelihood under the block-diagonal and uniform-sparse structures. Entries are means across 10 independent replications, with standard errors in parentheses. ``$>72$ hrs'' indicates that the procedure did not complete within the imposed 72-hour limit, and ``---'' indicates that no result is reported.}
		\label{tab:small_p_comptime}
	\end{table}

	Table~\ref{tab:small_p_comptime} summarizes the corresponding computation times. In the smaller settings, \texttt{Spcov} is the fastest procedure, while LLA-MM is substantially faster than BLOC and completes within roughly 700--13,700 seconds across the $(d,n)=(20,50)$ and $(50,100)$ configurations. BLOC is more computationally intensive in these settings, reflecting the additional cost of its global search strategy. At $(d,n)=(100,150)$, however, both LLA-MM variants exceeded the 72-hour computation limit under both covariance structures, whereas BLOC completed all configurations. \texttt{Spcov} failed under the block-diagonal design at this dimension; its reported runtime for the uniform-sparse case should be interpreted separately from the estimation-quality results discussed in the main manuscript.

	\subsection{Sparse covariance estimation with Frobenius norm ({$d\ge n$})}
	\label{sec:sims_large_p}
	
	We next examine the high-dimensional regime where the number of variables is comparable to or exceeds the number of observations ($d\ge n$). In this setting, the optimization problem in Equation~(1) of the main manuscript is instantiated using the Frobenius norm loss
	\[
	h_n(\bm\Gamma)=\|\bm\Gamma-\hat{\bm\Gamma}_S\|_F^2,
	\]
	where $\hat{\bm\Gamma}_S$ denotes the sample correlation matrix. We consider three sparse covariance structures following \cite{Wen2021}: the block-diagonal and banded designs are used as originally specified, while for the Toeplitz design we impose sparsity by truncating correlations below $0.001$, yielding a truncated-Toeplitz structure.
	\\[1ex]
	\emph{(i) Block-diagonal:} The index set $\{1,2,\dots,d\}$ is partitioned evenly into $K=d/10$ non-overlapping subsets $\{S_1,\dots,S_K\}$, each of size $10$. The entries of the covariance matrix $\bm\Sigma$ are defined by
	\[
	\Sigma_{ij}
	=
	0.2\,\mathbb{I}(i=j)
	+
	0.8\sum_{k=1}^K
	\mathbb{I}(i\in S_k,\;j\in S_k),
	\qquad 1\le i,j\le d,
	\]
	so that off-diagonal entries within the same block equal $0.8$, while entries across different blocks are zero.
	\\[1ex]
	\emph{(ii) Truncated Toeplitz:} Correlations decay geometrically with index distance and are truncated to zero when their magnitude falls below $0.001$:
	\[
	\Sigma_{ij}
	=
	0.75^{\,|i-j|}
	\mathbb{I}\!\left(0.75^{\,|i-j|}\ge 0.001\right),
	\qquad 1\le i,j\le d.
	\]
	\\[1ex]
	\emph{(iii) Banded:} Local dependence decreases linearly with index distance and is truncated beyond lag $10$:
	\[
	\Sigma_{ij}
	=
	\bigl(1-|i-j|/10\bigr)
	\mathbb{I}\{|i-j|\le 10\}.
	\]
	
	For each scenario, we generate $n$ i.i.d.\ observations from a zero-mean multivariate Normal distribution with covariance matrix $\bm\Sigma$ and estimate the corresponding correlation matrix under Frobenius loss. We consider $(d,n)\in\{(100,50),(100,100),(200,100)\}$.
	
	The proposed BLOC optimizer is coupled with the nonconvex SCAD and MCP penalties and evaluated against a broad collection of existing approaches for sparse covariance estimation. These include the $\ell_1$-ADMM estimator, which applies an $\ell_1$ penalty under an eigenvalue constraint and is solved by ADMM \citep{Liu2014}; iterative reweighting strategies such as IRW-ADMM (single-loop reweighted ADMM; \citealp{Wen2021}) and the traditional double-loop IRW with inner ADMM iterations \citep{Rothman2009}; block coordinate descent algorithms with nonconvex penalties, including $\ell_q$-BCD with $q=0.5$, $\ell_q$-BCD with $q=0$, and SCAD-BCD \citep{Wen2016}; and elementwise thresholding rules, namely soft, hard, and SCAD thresholding. We additionally include the MCP-based majorization--minimization estimator of \cite{WeiZhao2023} in its correlation- and covariance-matrix formulations, denoted WeiZhao-MCP-corr and WeiZhao-MCP-cov, respectively. MATLAB implementations of the remaining competing methods were obtained from the online repository associated with \cite{Wen2021}.
	
	For BLOC and both WeiZhao formulations, $\lambda$ is selected from 15 logarithmically spaced values over $[10^{-3},1]$ using five cross-validation splits. For BLOC, the SCAD shape parameter is set to $a=3.7$ and the MCP parameter to $\gamma=3$, while the WeiZhao-MCP implementations use their default MCP parameterization. Performance is evaluated using Frobenius and spectral norm errors, mean absolute deviation (MAD) of off-diagonal entries, true positive rate (TPR), false positive rate (FPR), Matthews correlation coefficient (MCC), and positive-definiteness (PD) rate.
	
	For $(d,n)=(100,50)$ (Table~\ref{tab:d100_n50}), BLOC performs particularly well under the truncated-Toeplitz structure, yielding the smallest Frobenius and spectral norm errors and the highest MCC ($0.53$--$0.55$), with essentially zero FPR. Under the banded structure, BLOC maintains near-zero FPR and MCC values of $0.79$--$0.83$, among the strongest values across methods. In the block-diagonal setting, several procedures recover the sparsity pattern nearly perfectly, with WeiZhao-MCP-corr providing the smallest estimation errors.
	
	For $(d,n)=(100,100)$ (Table~\ref{tab:d100_n100}), BLOC again performs strongly under the truncated-Toeplitz structure; BLOC-MCP attains the smallest norm-based errors and the highest MCC ($0.57$), while maintaining zero FPR. Under the banded structure, BLOC-SCAD attains an MCC of $0.87$ with essentially zero FPR, matching the highest MCC among the competing methods. Results for $(d,n)=(200,100)$, reported in the main manuscript, further show that BLOC substantially reduces Frobenius and spectral norm errors under both the truncated-Toeplitz and banded structures while retaining favorable TPR--FPR balance and high MCC. Computation times across all three dimensional settings are summarized in Table~\ref{tab:large_d_comptime}; as expected from its global-search strategy, BLOC is substantially more computationally intensive than the specialized competing procedures.

	\begin{table}[]
		\centering
		\resizebox{0.99\columnwidth}{!}{%
			\begin{tabular}{clccccccc}
				\hline
				Scenario & Method & Frob. norm & Spec. norm & MAD & TPR & FPR & MCC & PD rate \\ \hline
				\multicolumn{1}{c|}{\multirow{13}{*}{\begin{tabular}[c]{@{}c@{}}Block\\ diagonal\end{tabular}}} & BLOC-MCP & 5.39 (0.122) & 2.48 (0.108) & 0.014 (0.000) & 1.00 (0.000) & 0.00 (0.000) & 1.00 (0.002) & 1.00 (0.000) \\
				\multicolumn{1}{c|}{} & BLOC-SCAD & 3.48 (0.289) & 1.83 (0.206) & 0.009 (0.001) & 1.00 (0.000) & 0.04 (0.011) & 0.84 (0.041) & 1.00 (0.000) \\
				\multicolumn{1}{c|}{} & WeiZhao-MCP-corr & 3.25 (0.168) & 1.79 (0.119) & 0.007 (0.000) & 1.00 (0.000) & 0.00 (0.001) & 0.98 (0.004) & 1.00 (0.000) \\
				\multicolumn{1}{c|}{} & WeiZhao-MCP-cov & 7.39 (0.145) & 4.31 (0.222) & 0.023 (0.001) & 1.00 (0.000) & 0.09 (0.009) & 0.71 (0.026) & 1.00 (0.000) \\
				\multicolumn{1}{c|}{} & L1ADMM & 8.03 (0.184) & 4.16 (0.160) & 0.027 (0.001) & 1.00 (0.000) & 0.11 (0.021) & 0.67 (0.028) & 1.00 (0.000) \\
				\multicolumn{1}{c|}{} & LqBCD$(q=0.5)$ & 5.47 (0.174) & 3.25 (0.168) & 0.014 (0.001) & 1.00 (0.000) & 0.00 (0.001) & 0.99 (0.006) & 1.00 (0.000) \\
				\multicolumn{1}{c|}{} & LqBCD$(q=0)$ & 5.35 (0.134) & 3.04 (0.121) & 0.014 (0.000) & 1.00 (0.000) & 0.00 (0.000) & 1.00 (0.000) & 1.00 (0.000) \\
				\multicolumn{1}{c|}{} & SCAD-BCD & 5.46 (0.124) & 3.12 (0.146) & 0.014 (0.000) & 1.00 (0.000) & 0.02 (0.002) & 0.92 (0.009) & 1.00 (0.000) \\
				\multicolumn{1}{c|}{} & IRW-ADMM & 5.48 (0.124) & 3.29 (0.169) & 0.014 (0.000) & 1.00 (0.000) & 0.00 (0.000) & 1.00 (0.002) & 1.00 (0.000) \\
				\multicolumn{1}{c|}{} & TradIRW & 5.46 (0.157) & 3.28 (0.182) & 0.014 (0.001) & 1.00 (0.000) & 0.00 (0.000) & 1.00 (0.001) & 1.00 (0.000) \\
				\multicolumn{1}{c|}{} & Softthres & 8.29 (0.194) & 4.46 (0.191) & 0.029 (0.001) & 1.00 (0.000) & 0.14 (0.017) & 0.60 (0.021) & 0.90 (0.100) \\
				\multicolumn{1}{c|}{} & Hardthres & 6.94 (0.302) & 4.71 (0.419) & 0.016 (0.000) & 0.96 (0.010) & 0.00 (0.000) & 0.97 (0.005) & 0.10 (0.100) \\
				\multicolumn{1}{c|}{} & SCAD-thres & 7.36 (0.189) & 4.24 (0.276) & 0.023 (0.000) & 1.00 (0.000) & 0.10 (0.012) & 0.68 (0.027) & 0.00 (0.000) \\ \hline
				\multicolumn{1}{c|}{\multirow{13}{*}{\begin{tabular}[c]{@{}c@{}}Truncated\\ Toeplitz\end{tabular}}} & BLOC-MCP & 5.51 (0.143) & 2.33 (0.070) & 0.023 (0.001) & 0.42 (0.027) & 0.00 (0.002) & 0.53 (0.021) & 1.00 (0.000) \\
				\multicolumn{1}{c|}{} & BLOC-SCAD & 5.41 (0.118) & 2.23 (0.094) & 0.022 (0.000) & 0.44 (0.017) & 0.00 (0.001) & 0.55 (0.013) & 1.00 (0.000) \\
				\multicolumn{1}{c|}{} & WeiZhao-MCP-corr & 7.34 (0.086) & 3.21 (0.065) & 0.034 (0.000) & 0.30 (0.004) & 0.10 (0.005) & 0.25 (0.010) & 1.00 (0.000) \\
				\multicolumn{1}{c|}{} & WeiZhao-MCP-cov & 8.63 (0.164) & 3.64 (0.098) & 0.042 (0.002) & 0.37 (0.033) & 0.18 (0.033) & 0.23 (0.011) & 1.00 (0.000) \\
				\multicolumn{1}{c|}{} & L1ADMM & 8.53 (0.105) & 3.70 (0.103) & 0.040 (0.001) & 0.34 (0.027) & 0.15 (0.030) & 0.23 (0.014) & 1.00 (0.000) \\
				\multicolumn{1}{c|}{} & LqBCD$(q=0.5)$ & 8.32 (0.124) & 3.30 (0.089) & 0.034 (0.001) & 0.18 (0.003) & 0.02 (0.003) & 0.28 (0.007) & 1.00 (0.000) \\
				\multicolumn{1}{c|}{} & LqBCD$(q=0)$ & 8.48 (0.125) & 3.24 (0.088) & 0.035 (0.000) & 0.16 (0.002) & 0.01 (0.002) & 0.29 (0.005) & 1.00 (0.000) \\
				\multicolumn{1}{c|}{} & SCAD-BCD & 8.36 (0.116) & 3.39 (0.080) & 0.035 (0.000) & 0.21 (0.002) & 0.03 (0.003) & 0.29 (0.007) & 1.00 (0.000) \\
				\multicolumn{1}{c|}{} & IRW-ADMM & 8.13 (0.184) & 3.13 (0.141) & 0.034 (0.001) & 0.19 (0.009) & 0.02 (0.005) & 0.29 (0.008) & 1.00 (0.000) \\
				\multicolumn{1}{c|}{} & TradIRW & 8.14 (0.156) & 3.19 (0.123) & 0.033 (0.001) & 0.17 (0.005) & 0.01 (0.002) & 0.30 (0.006) & 1.00 (0.000) \\
				\multicolumn{1}{c|}{} & Softthres & 8.77 (0.145) & 3.81 (0.084) & 0.040 (0.001) & 0.34 (0.011) & 0.14 (0.010) & 0.23 (0.008) & 1.00 (0.000) \\
				\multicolumn{1}{c|}{} & Hardthres & 9.88 (0.233) & 3.52 (0.153) & 0.039 (0.001) & 0.14 (0.008) & 0.01 (0.001) & 0.27 (0.008) & 0.00 (0.000) \\
				\multicolumn{1}{c|}{} & SCAD-thres & 8.99 (0.225) & 3.62 (0.153) & 0.040 (0.001) & 0.32 (0.020) & 0.12 (0.016) & 0.24 (0.008) & 0.10 (0.100) \\ \hline
				\multicolumn{1}{c|}{\multirow{13}{*}{Banded}} & BLOC-MCP & 9.88 (0.094) & 4.04 (0.047) & 0.036 (0.000) & 0.67 (0.005) & 0.00 (0.000) & 0.79 (0.004) & 1.00 (0.000) \\
				\multicolumn{1}{c|}{} & BLOC-SCAD & 10.59 (0.930) & 4.32 (0.426) & 0.039 (0.004) & 0.73 (0.024) & 0.00 (0.001) & 0.83 (0.014) & 1.00 (0.000) \\
				\multicolumn{1}{c|}{} & WeiZhao-MCP-corr & 8.76 (0.199) & 3.95 (0.167) & 0.039 (0.003) & 0.82 (0.013) & 0.14 (0.045) & 0.61 (0.037) & 1.00 (0.000) \\
				\multicolumn{1}{c|}{} & WeiZhao-MCP-cov & 10.11 (0.317) & 4.74 (0.198) & 0.053 (0.004) & 0.88 (0.016) & 0.34 (0.066) & 0.44 (0.051) & 1.00 (0.000) \\
				\multicolumn{1}{c|}{} & L1ADMM & 9.36 (0.313) & 4.31 (0.206) & 0.043 (0.002) & 0.85 (0.015) & 0.21 (0.034) & 0.54 (0.032) & 1.00 (0.000) \\
				\multicolumn{1}{c|}{} & LqBCD$(q=0.5)$ & 8.86 (0.304) & 3.66 (0.195) & 0.032 (0.001) & 0.75 (0.005) & 0.02 (0.003) & 0.79 (0.011) & 1.00 (0.000) \\
				\multicolumn{1}{c|}{} & LqBCD$(q=0)$ & 9.35 (0.319) & 3.81 (0.210) & 0.034 (0.001) & 0.71 (0.011) & 0.02 (0.007) & 0.77 (0.014) & 1.00 (0.000) \\
				\multicolumn{1}{c|}{} & SCAD-BCD & 8.43 (0.289) & 3.51 (0.179) & 0.031 (0.001) & 0.83 (0.004) & 0.05 (0.005) & 0.76 (0.014) & 1.00 (0.000) \\
				\multicolumn{1}{c|}{} & IRW-ADMM & 8.53 (0.329) & 3.57 (0.204) & 0.031 (0.001) & 0.77 (0.012) & 0.07 (0.017) & 0.69 (0.028) & 1.00 (0.000) \\
				\multicolumn{1}{c|}{} & TradIRW & 8.68 (0.356) & 3.63 (0.213) & 0.032 (0.002) & 0.74 (0.015) & 0.02 (0.013) & 0.79 (0.021) & 1.00 (0.000) \\
				\multicolumn{1}{c|}{} & Softthres & 9.75 (0.326) & 4.47 (0.196) & 0.045 (0.002) & 0.85 (0.012) & 0.21 (0.026) & 0.53 (0.023) & 0.40 (0.163) \\
				\multicolumn{1}{c|}{} & Hardthres & 9.99 (0.418) & 4.21 (0.273) & 0.035 (0.002) & 0.65 (0.008) & 0.01 (0.002) & 0.74 (0.011) & 0.00 (0.000) \\
				\multicolumn{1}{c|}{} & SCAD-thres & 9.58 (0.371) & 4.21 (0.231) & 0.041 (0.002) & 0.83 (0.007) & 0.15 (0.007) & 0.57 (0.010) & 0.00 (0.000) \\ \hline
		\end{tabular}}
		\caption{Simulation results for sparse covariance estimation under Frobenius loss with $(d,n)=(100,50)$. Performance is evaluated across block-diagonal, truncated-Toeplitz, and banded covariance structures. Reported metrics are Frobenius norm error (Frob. norm), spectral norm error (Spec. norm), mean absolute deviation (MAD), true positive rate (TPR), false positive rate (FPR), Matthews correlation coefficient (MCC), and positive-definiteness (PD) rate. Entries are means across 10 independent replications, with standard errors in parentheses.}
		\label{tab:d100_n50} 
	\end{table}

	\begin{table}[]
		\centering
		\resizebox{0.99\columnwidth}{!}{%
			\begin{tabular}{clccccccc}
				\hline
				Scenario & Method & Frob. norm & Spec. norm & MAD & TPR & FPR & MCC & PD rate \\ \hline
				\multicolumn{1}{c|}{\multirow{13}{*}{\begin{tabular}[c]{@{}c@{}}Block\\ diagonal\end{tabular}}} & BLOC-MCP & 5.06 (0.097) & 2.00 (0.097) & 0.014 (0.000) & 1.00 (0.000) & 0.00 (0.000) & 1.00 (0.001) & 1.00 (0.000) \\
				\multicolumn{1}{c|}{} & BLOC-SCAD & 2.70 (0.187) & 1.40 (0.138) & 0.007 (0.001) & 1.00 (0.000) & 0.04 (0.008) & 0.83 (0.032) & 1.00 (0.000) \\
				\multicolumn{1}{c|}{} & WeiZhao-MCP-corr & 1.23 (0.079) & 0.58 (0.088) & 0.003 (0.000) & 1.00 (0.000) & 0.00 (0.000) & 1.00 (0.000) & 1.00 (0.000) \\
				\multicolumn{1}{c|}{} & WeiZhao-MCP-cov & 4.28 (0.329) & 2.44 (0.296) & 0.012 (0.001) & 1.00 (0.000) & 0.02 (0.004) & 0.91 (0.017) & 1.00 (0.000) \\
				\multicolumn{1}{c|}{} & L1ADMM & 5.62 (0.213) & 2.87 (0.175) & 0.019 (0.001) & 1.00 (0.000) & 0.15 (0.009) & 0.59 (0.014) & 1.00 (0.000) \\
				\multicolumn{1}{c|}{} & LqBCD$(q=0.5)$ & 4.18 (0.292) & 2.32 (0.240) & 0.011 (0.001) & 1.00 (0.000) & 0.00 (0.000) & 1.00 (0.000) & 1.00 (0.000) \\
				\multicolumn{1}{c|}{} & LqBCD$(q=0)$ & 4.13 (0.333) & 2.38 (0.278) & 0.011 (0.001) & 1.00 (0.000) & 0.00 (0.000) & 1.00 (0.000) & 1.00 (0.000) \\
				\multicolumn{1}{c|}{} & SCAD-BCD & 4.19 (0.339) & 2.37 (0.273) & 0.011 (0.001) & 1.00 (0.000) & 0.01 (0.006) & 0.97 (0.024) & 1.00 (0.000) \\
				\multicolumn{1}{c|}{} & IRW-ADMM & 4.16 (0.308) & 2.37 (0.268) & 0.011 (0.001) & 1.00 (0.000) & 0.00 (0.000) & 1.00 (0.000) & 1.00 (0.000) \\
				\multicolumn{1}{c|}{} & TradIRW & 4.20 (0.291) & 2.32 (0.211) & 0.011 (0.001) & 1.00 (0.000) & 0.00 (0.001) & 1.00 (0.004) & 1.00 (0.000) \\
				\multicolumn{1}{c|}{} & Softthres & 6.04 (0.249) & 3.11 (0.203) & 0.021 (0.001) & 1.00 (0.000) & 0.15 (0.022) & 0.59 (0.030) & 1.00 (0.000) \\
				\multicolumn{1}{c|}{} & Hardthres & 4.16 (0.328) & 2.38 (0.277) & 0.011 (0.001) & 1.00 (0.000) & 0.00 (0.000) & 1.00 (0.000) & 0.80 (0.133) \\
				\multicolumn{1}{c|}{} & SCAD-thres & 4.77 (0.328) & 2.78 (0.308) & 0.015 (0.001) & 1.00 (0.000) & 0.10 (0.020) & 0.69 (0.050) & 0.40 (0.163) \\ \hline
				\multicolumn{1}{c|}{\multirow{13}{*}{\begin{tabular}[c]{@{}c@{}}Truncated\\ Toeplitz\end{tabular}}} & BLOC-MCP & 4.10 (0.088) & 1.69 (0.078) & 0.017 (0.000) & 0.45 (0.020) & 0.00 (0.000) & 0.57 (0.015) & 1.00 (0.000) \\
				\multicolumn{1}{c|}{} & BLOC-SCAD & 5.42 (0.490) & 2.28 (0.239) & 0.023 (0.002) & 0.36 (0.043) & 0.00 (0.001) & 0.49 (0.034) & 1.00 (0.000) \\
				\multicolumn{1}{c|}{} & WeiZhao-MCP-corr & 5.72 (0.129) & 2.46 (0.093) & 0.027 (0.000) & 0.35 (0.026) & 0.12 (0.022) & 0.29 (0.013) & 1.00 (0.000) \\
				\multicolumn{1}{c|}{} & WeiZhao-MCP-cov & 6.37 (0.152) & 2.69 (0.134) & 0.030 (0.000) & 0.37 (0.023) & 0.14 (0.020) & 0.28 (0.010) & 1.00 (0.000) \\
				\multicolumn{1}{c|}{} & L1ADMM & 6.46 (0.112) & 2.90 (0.113) & 0.032 (0.001) & 0.41 (0.020) & 0.18 (0.024) & 0.25 (0.009) & 1.00 (0.000) \\
				\multicolumn{1}{c|}{} & LqBCD$(q=0.5)$ & 6.21 (0.084) & 2.44 (0.105) & 0.026 (0.000) & 0.22 (0.004) & 0.02 (0.002) & 0.33 (0.005) & 1.00 (0.000) \\
				\multicolumn{1}{c|}{} & LqBCD$(q=0)$ & 6.44 (0.088) & 2.43 (0.108) & 0.027 (0.000) & 0.21 (0.009) & 0.01 (0.004) & 0.34 (0.004) & 1.00 (0.000) \\
				\multicolumn{1}{c|}{} & SCAD-BCD & 6.02 (0.084) & 2.37 (0.105) & 0.027 (0.000) & 0.30 (0.007) & 0.06 (0.005) & 0.32 (0.005) & 1.00 (0.000) \\
				\multicolumn{1}{c|}{} & IRW-ADMM & 6.15 (0.084) & 2.38 (0.118) & 0.026 (0.000) & 0.22 (0.009) & 0.01 (0.004) & 0.34 (0.004) & 1.00 (0.000) \\
				\multicolumn{1}{c|}{} & TradIRW & 6.20 (0.108) & 2.41 (0.119) & 0.026 (0.000) & 0.22 (0.010) & 0.01 (0.004) & 0.34 (0.003) & 1.00 (0.000) \\
				\multicolumn{1}{c|}{} & Softthres & 6.61 (0.150) & 2.91 (0.156) & 0.033 (0.000) & 0.43 (0.025) & 0.21 (0.028) & 0.24 (0.011) & 1.00 (0.000) \\
				\multicolumn{1}{c|}{} & Hardthres & 6.97 (0.164) & 2.52 (0.112) & 0.029 (0.001) & 0.19 (0.008) & 0.00 (0.001) & 0.33 (0.006) & 0.00 (0.000) \\
				\multicolumn{1}{c|}{} & SCAD-thres & 6.33 (0.093) & 2.57 (0.124) & 0.030 (0.000) & 0.38 (0.015) & 0.15 (0.016) & 0.27 (0.007) & 0.00 (0.000) \\ \hline
				\multicolumn{1}{c|}{\multirow{13}{*}{Banded}} & BLOC-MCP & 10.04 (0.150) & 4.10 (0.068) & 0.037 (0.001) & 0.64 (0.008) & 0.00 (0.000) & 0.77 (0.005) & 1.00 (0.000) \\
				\multicolumn{1}{c|}{} & BLOC-SCAD & 8.20 (0.793) & 3.26 (0.338) & 0.029 (0.004) & 0.79 (0.020) & 0.00 (0.000) & 0.87 (0.013) & 1.00 (0.000) \\
				\multicolumn{1}{c|}{} & WeiZhao-MCP-corr & 6.29 (0.124) & 2.95 (0.085) & 0.028 (0.001) & 0.89 (0.008) & 0.16 (0.006) & 0.61 (0.007) & 1.00 (0.000) \\
				\multicolumn{1}{c|}{} & WeiZhao-MCP-cov & 7.05 (0.200) & 3.44 (0.169) & 0.030 (0.001) & 0.89 (0.010) & 0.17 (0.007) & 0.60 (0.006) & 1.00 (0.000) \\
				\multicolumn{1}{c|}{} & L1ADMM & 6.92 (0.226) & 3.34 (0.181) & 0.032 (0.001) & 0.90 (0.012) & 0.23 (0.039) & 0.55 (0.032) & 1.00 (0.000) \\
				\multicolumn{1}{c|}{} & LqBCD$(q=0.5)$ & 6.00 (0.184) & 2.67 (0.153) & 0.021 (0.001) & 0.86 (0.009) & 0.02 (0.002) & 0.87 (0.004) & 1.00 (0.000) \\
				\multicolumn{1}{c|}{} & LqBCD$(q=0)$ & 6.44 (0.208) & 2.89 (0.133) & 0.023 (0.001) & 0.83 (0.017) & 0.01 (0.005) & 0.85 (0.007) & 1.00 (0.000) \\
				\multicolumn{1}{c|}{} & SCAD-BCD & 5.66 (0.143) & 2.56 (0.128) & 0.022 (0.001) & 0.92 (0.011) & 0.09 (0.011) & 0.74 (0.016) & 1.00 (0.000) \\
				\multicolumn{1}{c|}{} & IRW-ADMM & 6.00 (0.182) & 2.68 (0.153) & 0.022 (0.001) & 0.83 (0.015) & 0.05 (0.013) & 0.77 (0.020) & 1.00 (0.000) \\
				\multicolumn{1}{c|}{} & TradIRW & 5.84 (0.167) & 2.59 (0.153) & 0.021 (0.001) & 0.83 (0.011) & 0.01 (0.004) & 0.86 (0.007) & 1.00 (0.000) \\
				\multicolumn{1}{c|}{} & Softthres & 7.02 (0.198) & 3.38 (0.179) & 0.034 (0.001) & 0.91 (0.011) & 0.26 (0.028) & 0.51 (0.024) & 0.20 (0.133) \\
				\multicolumn{1}{c|}{} & Hardthres & 6.54 (0.112) & 2.73 (0.132) & 0.022 (0.000) & 0.74 (0.012) & 0.00 (0.001) & 0.82 (0.006) & 0.00 (0.000) \\
				\multicolumn{1}{c|}{} & SCAD-thres & 6.35 (0.148) & 2.93 (0.142) & 0.027 (0.000) & 0.88 (0.011) & 0.15 (0.014) & 0.62 (0.013) & 0.00 (0.000) \\ \hline
		\end{tabular}}
		\caption{Simulation results for sparse covariance estimation under Frobenius loss with $(d,n)=(100,100)$. Performance is evaluated across block-diagonal, truncated-Toeplitz, and banded covariance structures. Reported metrics are Frobenius norm error (Frob. norm), spectral norm error (Spec. norm), mean absolute deviation (MAD), true positive rate (TPR), false positive rate (FPR), Matthews correlation coefficient (MCC), and positive-definiteness (PD) rate. Entries are means across 10 independent replications, with standard errors in parentheses.}
		\label{tab:d100_n100} 
	\end{table}

	\begin{sidewaystable}[htbp]
		\renewcommand{\arraystretch}{1.25}
		\resizebox{0.99\columnwidth}{!}{%
			\begin{tabular}{|l|ccc|ccc|ccc|}
				\hline
				\multirow{2}{*}{Method} & \multicolumn{3}{c|}{$d= 100, n = 50$} & \multicolumn{3}{c|}{$d= 100, n = 100$} & \multicolumn{3}{c|}{$d= 200, n = 100$} \\ \cline{2-10} 
				& \begin{tabular}[c]{@{}c@{}}Block\\ diagonal\end{tabular} & \begin{tabular}[c]{@{}c@{}}Truncated\\ Toeplitz\end{tabular} & Banded & \begin{tabular}[c]{@{}c@{}}Block\\ diagonal\end{tabular} & \begin{tabular}[c]{@{}c@{}}Truncated\\ Toeplitz\end{tabular} & Banded & \begin{tabular}[c]{@{}c@{}}Block\\ diagonal\end{tabular} & \begin{tabular}[c]{@{}c@{}}Truncated\\ Toeplitz\end{tabular} & Banded \\ \hline
				BLOC-MCP & 7502.52 (8.878) & 8034.35 (87.548) & 7404.20 (4.106) & 7508.57 (11.522) & 7944.62 (54.638) & 7361.07 (5.375) & 53283.30 (2555.576) & 48452.50 (3574.939) & 45121.10 (5272.815) \\
				BLOC-SCAD & 7799.78 (102.901) & 8111.07 (61.301) & 7638.56 (20.714) & 8235.77 (119.255) & 8982.13 (248.992) & 7536.56 (15.657) & 52966.10 (2362.879) & 34290.30 (3066.436) & 44685.60 (3164.055) \\
				WeiZhao-MCP-corr & 1.62 (0.020) & 1.69 (0.021) & 1.78 (0.019) & 1.57 (0.014) & 1.79 (0.019) & 1.81 (0.015) & 5.22 (0.041) & 5.09 (0.027) & 5.06 (0.035) \\
				WeiZhao-MCP-cov & 1.62 (0.024) & 1.75 (0.030) & 1.94 (0.056) & 1.85 (0.068) & 1.97 (0.024) & 2.11 (0.111) & 5.70 (0.146) & 5.66 (0.051) & 5.82 (0.046) \\
				L1ADMM & 1.10 (0.035) & 1.02 (0.027) & 1.09 (0.025) & 0.97 (0.023) & 1.76 (0.267) & 1.01 (0.031) & 7.10 (2.271) & 4.69 (0.217) & 4.87 (0.085) \\
				LqBCD$(q=0.5)$ & 8.26 (0.173) & 8.05 (0.256) & 8.81 (0.191) & 7.36 (0.033) & 16.74 (3.128) & 8.93 (0.227) & 36.47 (1.342) & 34.23 (0.920) & 37.62 (0.506) \\
				LqBCD$(q=0)$ & 7.32 (0.157) & 7.20 (0.140) & 8.12 (0.180) & 7.10 (0.097) & 14.33 (2.500) & 8.05 (0.057) & 34.75 (0.988) & 36.55 (1.096) & 36.68 (0.451) \\
				SCAD-BCD & 4.97 (0.076) & 4.62 (0.101) & 5.53 (0.084) & 5.02 (0.075) & 8.61 (1.614) & 5.45 (0.091) & 26.44 (2.604) & 20.27 (0.473) & 24.19 (0.494) \\
				IRW-ADMM & 1.98 (0.037) & 1.99 (0.023) & 2.12 (0.029) & 1.91 (0.021) & 3.82 (0.597) & 2.22 (0.031) & 10.42 (1.394) & 9.53 (0.256) & 10.00 (0.143) \\
				TradIRW & 10.13 (0.243) & 10.61 (0.177) & 12.03 (0.213) & 8.52 (0.146) & 19.73 (3.189) & 12.23 (0.113) & 53.54 (4.510) & 53.12 (1.007) & 59.31 (0.809) \\
				Softthres & 0.06 (0.003) & 0.05 (0.002) & 0.06 (0.003) & 0.05 (0.001) & 0.08 (0.011) & 0.05 (0.001) & 0.15 (0.011) & 0.13 (0.005) & 0.13 (0.003) \\
				Hardthres & 0.06 (0.003) & 0.06 (0.003) & 0.06 (0.003) & 0.05 (0.001) & 0.09 (0.012) & 0.05 (0.001) & 0.15 (0.013) & 0.15 (0.011) & 0.14 (0.006) \\
				SCAD-thres & 0.10 (0.005) & 0.10 (0.003) & 0.10 (0.002) & 0.10 (0.002) & 0.17 (0.025) & 0.10 (0.002) & 0.27 (0.011) & 0.25 (0.009) & 0.26 (0.009) \\ \hline
		\end{tabular}}
		\caption{Computation times (seconds) for sparse covariance estimation under Frobenius loss across $(d,n)=(100,50)$, $(100,100)$, and $(200,100)$ and the block-diagonal, truncated-Toeplitz, and banded structures. Entries are means across 10 independent replications, with standard errors in parentheses.}
		\label{tab:large_d_comptime} 
	\end{sidewaystable}
	\newpage
	\subsection{Sparse covariance estimation with matrix-depth loss}
	\label{sec:supp_mat_depth}
	
	We consider an additional experiment designed to evaluate optimization under a genuinely discontinuous statistical criterion. The construction is motivated by the matrix-depth formulation of \cite{chen2018robust}, which defines covariance depth through one-dimensional projections. For the present computational experiment, we use a finite-direction average version of the empirical depth.
	
	Let $\bm u_1,\ldots,\bm u_K$ denote fixed unit vectors in $\mathbb{R}^d$ and let $\bm X_1,\ldots,\bm X_n$ denote the observed sample. For a candidate correlation matrix $\bm\Gamma\in\mathcal C_d$, define the empirical directional depth
	\[
	D_k(\bm\Gamma)
	=
	\min\left\{
	\frac{1}{n}\sum_{i=1}^n
	I\!\left[
	(\bm u_k^\top\bm X_i)^2
	\le
	\beta\,\bm u_k^\top\bm\Gamma\bm u_k
	\right],
	\;
	\frac{1}{n}\sum_{i=1}^n
	I\!\left[
	(\bm u_k^\top\bm X_i)^2
	\ge
	\beta\,\bm u_k^\top\bm\Gamma\bm u_k
	\right]
	\right\},
	\]
	where
	\[
	\beta=\{\Phi^{-1}(0.75)\}^2.
	\]
	The finite-direction average depth is
	\[
	\overline D_K(\bm\Gamma)
	=
	\frac{1}{K}\sum_{k=1}^K D_k(\bm\Gamma),
	\]
	with maximum possible value $0.5$. The sparse estimator is obtained by minimizing
	\[
	Q_\lambda(\bm\Gamma)
	=
	-\overline D_K(\bm\Gamma)
	+
	\sum_{i<j}p_\lambda(|\gamma_{ij}|),
	\qquad
	\bm\Gamma\in\mathcal C_d,
	\]
	where $p_\lambda$ is either the SCAD or MCP penalty. We use the conventional SCAD shape parameter $a=3.7$ and MCP parameter $\gamma=3$. Because
	$\overline D_K(\bm\Gamma)$ depends on empirical indicator functions, it changes only when one or more projected observations cross a directional threshold. Consequently, the depth component is piecewise constant with discontinuities, creating extended flat regions separated by jumps and providing a substantially different optimization landscape from the smooth Gaussian-likelihood and Frobenius-loss examples considered in Sections~5.2--5.3 of the main manuscript.
	
	\paragraph{Simulation design and contamination.}
	We set $(d,n)=(20,200)$ and consider two structured true correlation matrices: a block-diagonal structure, defined as in the preceding Frobenius-loss experiments but with block size $5$, and the same banded structure considered there. For each configuration, 10 independent data sets are generated. An uncontaminated sample is first generated from $N_d(\bm 0,\bm\Gamma_0)$, after which exactly $10\%$ of the observations are selected at random and replaced by draws from
	\[
	N_d(\bm 0,9\bm I_d).
	\]
	The same contaminated samples are used for all optimization methods. We use $K=500$ fixed random directions, generated once and held identical across all
	replications, penalties, and competing optimizers.
	
	\paragraph{Tuning.}
	The regularization parameter is selected from
	\[
	\lambda\in\{10^{x_\ell}:x_\ell
	\text{ equally spaced from }-3\text{ to }0,\ \ell=1,\ldots,15\}.
	\]
	We use five repeated train--test splits. Within each split, $n_2=\operatorname{round}\{n/\log(n)\}$ observations are held out and the remaining
	$n_1=n-n_2$ observations are used for estimation. Because the experiment is designed for contaminated data, the validation target is a robust correlation estimate based on Kendall's $\tau$. Specifically, if $\widehat\tau_{jk}$ denotes the held-out Kendall correlation, we form
	\[
	\widetilde\Gamma_{jk}
	=
	\sin\left(\frac{\pi}{2}\widehat\tau_{jk}\right),
	\qquad
	\widetilde\Gamma_{jj}=1.
	\]
	For each candidate $\lambda$, the fitted correlation matrix is compared with this held-out robust target using Frobenius distance. The $\lambda$ minimizing the mean validation error across the five splits is selected, after which each method is refitted to the complete contaminated sample.

	\paragraph{Competing optimization methods.}
	For each of SCAD and MCP, BLOC is compared with three constrained optimization routines from MATLAB's \texttt{fmincon}: active-set (AS), interior-point (IP), and sequential quadratic programming (SQP), together with three local methods from the \texttt{Manopt} toolbox: Barzilai--Borwein (BB), conjugate gradient (CG), and steepest descent (SD). Each optimizer is supplied with the same penalized matrix-depth objective and evaluated using the same data realizations, fixed directions, tuning grid, and validation splits. Performance is assessed using Frobenius and spectral norm errors, MAD, TPR, FPR, MCC, and achieved unpenalized matrix depth. MCC is reported as undefined when its denominator vanishes because of a degenerate selected support.
	
	As shown in Table~3 of the main manuscript, BLOC provides substantially better matrix estimation and reaches markedly deeper solutions under both covariance structures. Under the block-diagonal design, BLOC also achieves high sensitivity while maintaining a low false-positive rate, whereas the local methods either recover few true edges or incur substantial numbers of false positives. Under the banded design, some local methods attain relatively high TPR, but only at the expense of markedly inflated FPR. Thus, support recovery cannot be assessed through sensitivity alone: BLOC provides a substantially better balance between true- and false-positive edge recovery, while also yielding smaller Frobenius and spectral errors and higher attained depth. For the local methods, MCC is undefined in several settings because of degenerate support-recovery contingency tables.
	\begin{table}[]
		\centering
		\renewcommand{\arraystretch}{1.1}
		\resizebox{0.45\columnwidth}{!}{%
			\begin{tabular}{lcc}
				\hline
				Method & \begin{tabular}[c]{@{}c@{}}Block\\ diagonal\end{tabular} & Banded \\ \hline
				BLOC-SCAD & 21.0 (1.75) & 35.3 (4.11) \\
				BLOC-MCP & 20.4 (1.49) & 34.3 (3.88) \\
				fmincon:AS-SCAD & 0.1 (0.01) & 0.1 (0.04) \\
				fmincon:AS-MCP & 0.1 (0.01) & 0.1 (0.03) \\
				fmincon:IP-SCAD & 0.1 (0.00) & 0.1 (0.00) \\
				fmincon:IP-MCP & 0.1 (0.00) & 0.1 (0.00) \\
				fmincon:SQP-SCAD & 0.1 (0.02) & 0.1 (0.04) \\
				fmincon:SQP-MCP & 0.1 (0.02) & 0.1 (0.04) \\
				Manopt:BB-SCAD & 7.1 (5.64) & 30.0 (25.35) \\
				Manopt:BB-MCP & 7.6 (5.70) & 4.0 (2.51) \\
				Manopt:CG-SCAD & 5.1 (1.84) & 5.3 (1.74) \\
				Manopt:CG-MCP & 3.5 (0.66) & 4.2 (1.65) \\
				Manopt:SD-SCAD & 2.1 (1.04) & 2.5 (0.48) \\
				Manopt:SD-MCP & 1.5 (0.50) & 3.1 (0.57) \\ \hline
		\end{tabular}}
		\caption{Final-fit computation times (seconds) for sparse correlation estimation under the matrix-depth loss with $(d,n)=(20,200)$ and $10\%$ gross contamination. Entries are means across 10 replications, with standard errors in parentheses. Times correspond to the final full-data optimization after selection of $\lambda$ and therefore exclude the tuning stage. AS, IP, and SQP denote the active-set, interior-point, and sequential quadratic programming implementations of \texttt{fmincon}; BB, CG, and SD denote the Barzilai--Borwein, conjugate-gradient, and steepest-descent implementations from \texttt{Manopt}.}
		\label{tab:mat_depth_comptime} 
	\end{table}
	
	Table~\ref{tab:mat_depth_comptime} reports the corresponding final-fit computation times. BLOC requires approximately 20--35 seconds on average across the two structures. Several local solvers terminate much more rapidly; however, these runtimes should be interpreted together with Table~3 in the main manuscript, since their corresponding solutions generally attain substantially lower depth and larger estimation error. Thus, the additional computational cost of BLOC in this experiment reflects its broader exploration of the discontinuous objective surface rather than rapid termination within locally flat regions.
	
	\section{Proofs of Theorems}\label{app:theorem}
	\subsection{Notations}
	Let $\mathcal{C}_d = \{ \bm{C} \in \mathbb{R}^{d \times d} : \bm{C} \succ 0,\; \bm C=\bm C^T,\; \text{diag}(\bm{C}) = 1 \}$ denote the set of $d\times d$ matrices that are symmetric, positive definite and have all their diagonal elements equal to $1$. Throughout this article, the entries of $\bm \Gamma \in \mathcal C_d$ will be denoted by $\gamma_{ij},~1\le i,j\le d$. For a locally Lipschitz function \footnote{A function $f : \mathbb R^d\to\mathbb R$ is locally Lipschitz if for all $\bm x\in\mathbb R^d$, there exists a neighborhood $S$ of $\bm x$, such that $f$ is Lipschitz in $S$. Hence, for all $\bm x'\in S$ there exists $L$ such that $|f(\bm x) - f(\bm x')| \le L\|x-x'\|$.} $f:\mathcal C_d\to\mathbb R$ its Clarke differential \citep{clarke1990optimization} at $\bm\Gamma$ will be denoted by $\partial f(\bm\Gamma)$. The cardinality of a set $S$ will be denoted by $|S|$ and its complement by $S^c$. For any matrix $\bm A=(a_{ij})$, $\|\bm A\|_F$ represents its Frobenius norm and its operator norm will be denoted by $\|\bm A\|$. We will use
	$\|\bm A\|_{\ell_1}
	=
	\sum_{i,j}|a_{ij}|$
	to denote the entry-wise \(\ell_1\)-norm. 
	\subsection{Proofs of Propositions 1, 2 and Theorems 3, 4}
	\label{app:bloc_proofs}
	\subsubsection{Proof of Proposition 1}
	\begin{proof}
		\textbf{($\Rightarrow$)} Suppose $\bm{C}$ is a correlation matrix. Then $\bm{C}$ is symmetric, positive definite, and has ones on its diagonal. Since $\bm{C}$ is positive definite, it admits a unique Cholesky decomposition:
		$$\bm{C} = \bm{L}^*{\bm{L}^*}^\top,$$
		where $\bm{L}^*$ is a lower triangular full-rank matrix with strictly positive diagonal entries. Let $\bm{l}^*_m$ denote the $m$-th row of $\bm{L}^*$. Then,
		
		$$C_{mm} = (\bm{L}^*{\bm{L}^*}^\top)_{mm} = \bm{l}^*_m {\bm{l}^*_m}^\top = \|\bm{l}_m\|_2^2.$$
		But $C_{mm} = 1$ (since $\bm{C}$ is a correlation matrix), so $\|\bm{l}^*_m\|_2 = 1$ for all $m$. Hence, $\bm{L}^*$ satisfies all the properties of $\bm{L}$ and can be taken as such. Since the Cholesky factorization is unique under the condition of positive diagonal entries, the representation $\bm{C} = \bm{L}\bm{L}^\top$ with the specified properties is unique.
		
		\textbf{($\Leftarrow$)} Conversely, suppose there exists a full-rank lower triangular matrix $\bm{L}$ with strictly positive diagonal entries and each row having L2 norm 1 such that $\bm{C} = \bm{LL}^\top$. Then:
		\begin{itemize}
			\item $\bm{C}$ is symmetric because $\bm{C} = \bm{LL}^\top$,
			\item $\bm{C}$ is positive definite since for any $x \in \mathbb{R}^d$, we have
			$$\bm{x}^\top \bm{C} \bm{x} = \bm{x}^\top \bm{LL}^\top \bm{x} = \|\bm{L}^\top \bm{x}\|_2^2 > 0,$$
			and since $\bm{L}$ is full rank, $C$ is positive definite,
			\item the diagonal entries of $C$ are
			$$C_{mm} = (\bm{LL}^\top)_{mm} = \bm{l}_m \bm{l}_m^\top = \|\bm{l}_m\|_2^2 = 1,$$
			by the assumption on the row norms of $\bm{L}$.
		\end{itemize}
		\noindent Therefore, $\bm{C}$ is a correlation matrix. Uniqueness follows from the uniqueness of the Cholesky factorization with positive diagonals.
	\end{proof}
	\subsubsection{Proof of Proposition 2}
	\begin{proof}
		We demonstrate both directions of the claimed bijection.
		
		\noindent \emph{(i) Angular coordinates yield a valid correlation matrix:}  
		Suppose we are given a collection of angles \( \{\omega_{21}, \Omega_3, \ldots, \Omega_d\} \) from their respective domains. Construct the lower triangular matrix \( \bm{L} \) as follows:
		\begin{itemize}
			\item Set \( l_{11} = 1 \),
			\item For \( m = 2 \), define \( l_{21} = \sin\omega_{21} \) and \( l_{22} = \cos\omega_{21} \),
			\item For each \( m \geq 3 \), define the row vector \( \bm{l}_m \) according to the recursive angular formulation given earlier.
		\end{itemize}
		By construction, each \( \bm{l}_m \) lies on the unit sphere in \( \mathbb{R}^m \) with its last coordinate strictly positive, i.e., \( \|\bm{l}_m\| = 1 \) and \( l_{mm} > 0 \). Hence,
		\begin{itemize}
			\item All rows of \( \bm{L} \) are unit-norm,
			\item The matrix \( \bm{C} = \bm{L} \bm{L}^\top \) satisfies \( C_{ii} = \bm{l}_i^\top \bm{l}_i = 1 \), ensuring unit diagonal entries,
			\item Since \( \bm{L} \) is full-rank with positive diagonals, \( \bm{C} \) is symmetric positive definite.
		\end{itemize}
		It follows that the resulting matrix \( \bm{C} \) is a valid correlation matrix.
		
		\medskip
		
		\noindent \emph{(ii) Any correlation matrix corresponds to a unique angular representation:}  
		Let \( \bm{C} \in \mathcal{C}_d \) be a given full-rank correlation matrix. Then its Cholesky factorization \( \bm{C} = \bm{L} \bm{L}^\top \) exists and is unique under the constraint that \( \bm{L} \) has positive diagonal entries.
		
		\begin{itemize}
			\item The condition \( C_{11} = 1 \) implies \( l_{11} = 1 \),
			\item For \( m = 2 \), the angle can be recovered as \( \omega_{21} = \arctan(l_{21} / l_{22}) \in (-\pi/2, \pi/2) \),
			\item For \( m \geq 3 \), since \( \bm{l}_m \in \mathbb{S}_+^{m-1} \), and the spherical coordinate map \( \xi_m : \Theta_m \to \mathbb{S}_+^{m-1} \) is known to be bijective \citep{blumenson1960}, there exists a unique angular vector \( \Omega_m \) such that \( \bm{l}_m = \xi_m(\Omega_m) \).
		\end{itemize}
		Therefore, one can uniquely reconstruct the angular coordinates from any valid correlation matrix.
		
		\medskip
		
		\noindent \emph{Conclusion:}  
		Together, parts (i) and (ii) confirm that the mapping between angular coordinate tuples \( \{\omega_{21}, \Omega_3, \ldots, \Omega_d\} \) and the space of correlation matrices \( \mathcal{C}_d \) is bijective via the Cholesky decomposition.
	\end{proof}
	
	\subsubsection{Proof of Theorem 3}
	\begin{proof}
		Take an open neighborhood \( \bm{U} \subset \mathcal{C} \) with respect to the \( \ell_\infty \)-norm containing \( \bm{\nu} \), on which \( f \) is convex. So, there exists \( r > 0 \) such that \( \bm{U} = \prod_{i=1}^N U_i \), where \( U_i = (\nu_i - r, \nu_i + r) \) for \( i = 1,\ldots,N \). Fix any \( i \in \{1,\ldots,N\} \), and define the univariate slice:
		\[
		g_i(z) := f(\nu_1, \ldots, \nu_{i-1}, z, \nu_{i+1}, \ldots, \nu_N), \quad z \in U_i.
		\]
		Since \( f \) is convex on \( \bm{U} \), each \( g_i \) is convex on \( U_i \), and differentiable. We claim \( g_i(\nu_i) \leq g_i(z) \) for all \( z \in U_i \).
		
		Suppose instead that there exists \( \nu_i^* \in U_i \) such that \( g_i(\nu_i^*) < g_i(\nu_i) \). Let \( d = |\nu_i^* - \nu_i| > 0 \). Without loss of generality, assume \( \nu_i^* > \nu_i \), so \( \nu_i^* = \nu_i + d \).
		Since \( \delta_k \to 0 \), there exists some \( K \in \mathbb{N} \) such that \( \delta_K < d \). Then:
		\[
		\nu_i < \nu_i + \delta_K < \nu_i + d = \nu_i^*.
		\]
		Therefore, there exists \( \lambda \in (0,1) \) such that:
		\[
		\nu_i + \delta_K = \lambda \nu_i + (1 - \lambda)(\nu_i + d).
		\]
		Using convexity of \( g_i \), we have:
		\[
		g_i(\nu_i + \delta_K) \leq \lambda g_i(\nu_i) + (1 - \lambda) g_i(\nu_i^*) < g_i(\nu_i),
		\]
		since \( g_i(\nu_i^*) < g_i(\nu_i) \). But this contradicts our assumption that \( f(\bm{\nu}) \leq f(\bm{\nu}_k^{(i+)}) \). So no such \( \nu_i^* \) can exist, and \( g_i(\nu_i) \leq g_i(z) \) for all \( z \in U_i \).
		
		Because \( g_i \) is differentiable and attains a minimum at \( z = \nu_i \), we conclude \( g_i'(\nu_i) = 0 \), i.e.,
		\[
		\frac{\partial}{\partial x_i} f(\bm{\nu}) = 0.
		\]
		This holds for all \( i = 1,\ldots,N \), so \( \nabla f(\bm{\nu}) = \bm{0} \).
		
		This result is aligned with the pattern search literature, notably \citet[Theorem~3.6]{Torczon1997}, which shows that under coordinate-wise polling with decreasing step sizes, accumulation points of iterates satisfy first-order stationarity conditions.
	\end{proof}
	\subsubsection{Proof of Theorem 4}
	\begin{proof}
		Fix \(r\) and let \(\bm{\nu}=\bm{\nu}^{(r)}\), with
		\(s_r=s_0/\rho^r\) denoting the step size at the corresponding
		unsuccessful poll. By assumption, no polling direction
		\(\bm{d}_i\in\{\pm\bm{e}_1,\ldots,\pm\bm{e}_N\}\) results in a
		decrease in the objective function. That is, for all
		\(j=1,\ldots,N\), we have:
		\[
		f(\bm{\nu} + s_r \bm{e}_j) \ge f(\bm{\nu}) \quad \text{and} \quad f(\bm{\nu} - s_r \bm{e}_j) \ge f(\bm{\nu}).
		\]
		
		Now consider a Taylor expansion using the Lipschitz continuity of \( \nabla f \), which gives for any direction \( \bm{d}_i \):
		\[
		f(\bm{\nu} + s_r \bm{d}_i) \le f(\bm{\nu}) + s_r \nabla f(\bm{\nu})^\top \bm{d}_i + \frac{L}{2} s_r^2.
		\]
		
		\paragraph{Case 1: \( \bm{d}_i = +\bm{e}_j \).}
		From the assumption that this move fails to improve:
		\[
		f(\bm{\nu} + s_r \bm{e}_j) \ge f(\bm{\nu}) \Rightarrow s_r \frac{\partial f}{\partial \nu_j} + \frac{L}{2} s_r^2 \ge 0.
		\]
		Dividing by \( s_r \) ($>$ 0), we get:
		\[
		\frac{\partial f}{\partial \nu_j} \ge -\frac{L}{2} s_r.
		\]
		
		\paragraph{Case 2: \( \bm{d}_i = -\bm{e}_j \).}
		Similarly, the expansion gives:
		\[
		f(\bm{\nu} - s_r \bm{e}_j) \le f(\bm{\nu}) - s_r \frac{\partial f}{\partial \nu_j} + \frac{L}{2} s_r^2,
		\]
		and the assumption that this move is also rejected gives:
		\[
		-\frac{\partial f}{\partial \nu_j} + \frac{L}{2} s_r \ge 0 \Rightarrow \frac{\partial f}{\partial \nu_j} \le \frac{L}{2} s_r.
		\]
		Combining both:
		\[
		\left| \frac{\partial f}{\partial \nu_j} \right| \le \frac{L}{2} s_r \quad \text{for all } j = 1, \ldots, N.
		\]
		Hence:
		\[
		\| \nabla f(\bm{\nu}) \|_2 \le \sqrt{N} \cdot \frac{L}{2} s_r.
		\]
		By convexity of $f$,
		\[
		f(\bm{\nu})-f(\bm{\nu}^*)
		\le
		\nabla f(\bm{\nu})^\top(\bm{\nu}-\bm{\nu}^*)
		\le
		\|\nabla f(\bm{\nu})\|_2\,\|\bm{\nu}-\bm{\nu}^*\|_2.
		\]
		Since $\mathcal C$ is compact, its diameter
		\[
		D=\sup_{\bm{\nu}_1,\bm{\nu}_2\in\mathcal C}
		\|\bm{\nu}_1-\bm{\nu}_2\|_2
		\]
		is finite. Therefore,
		\[
		f(\bm{\nu})-f(\bm{\nu}^*)
		\le
		D\|\nabla f(\bm{\nu})\|_2
		\le
		\frac{DL\sqrt{N}}{2}s_r.
		\]
		Since $s_r=s_0/\rho^r$ and $\bm{\nu}=\bm{\nu}^{(r)}$,
		\[
		f(\bm{\nu}^{(r)})-f(\bm{\nu}^*)
		\le
		\frac{DL\sqrt{N}s_0}{2}\rho^{-r}.
		\]
		For every $\rho>1$, there exists a constant $C_\rho>0$ such that
		$\rho^{-r}\le C_\rho/(r+1)$. Hence,
		\[
		f(\bm{\nu}^{(r)})-f(\bm{\nu}^*)
		\le
		\frac{C}{r+1}
		\]
		for some constant $C>0$, which establishes the stated sublinear bound.
	\end{proof}
	
	\paragraph{Comments on Theorem 4.} Theorem~4 establishes an $O(1/r)$ convergence-rate bound for BLOC within a single \emph{run} when applied to smooth convex objectives. Importantly, however, $r$ denotes the number of unsuccessful polls, equivalently the number of step-size reductions, rather than the total number of BLOC iterations. The number of intermediate polling iterations between successive reductions depends on the objective landscape. Consequently, the stated rate should not be interpreted as directly equivalent to the usual iteration-wise $O(1/r)$ rates reported for gradient descent or coordinate descent in smooth convex optimization \citep{Nesterov2004, Wright2015, Nesterov2018}.
	
	The result nevertheless provides a useful characterization of the refinement behavior of BLOC using only function evaluations: as unsuccessful polls trigger progressively smaller step sizes, the objective gap is controlled at the corresponding sequence of iterates. This is particularly relevant to the derivative-free nature of BLOC, although the smoothness and convexity assumptions of Theorem~4 are stronger than those required for the algorithm to be applied in practice. Complementary global-search properties under an explicit randomized multi-start restart scheme are given in Lemma~\ref{thm:open_ball_reach} and Theorem~\ref{thm:global_conv_prob}.
	\subsection{Global-search guarantees under randomized restarts}
	\label{sec:global_restart}
	
	The results in this subsection complement the within-\emph{run} stationarity and convergence-rate results presented in the main manuscript. We consider a theoretical randomized multi-start version of BLOC in which successive \emph{runs} are initialized independently from points on increasingly refined grids over the search domain. This restart mechanism is used to formalize the global exploratory behavior of the algorithm and differs from the default BLOC implementation, where each successive \emph{run} is initialized from the best solution obtained in the preceding \emph{run} while resetting the step size. Accordingly, the results below should be interpreted as guarantees under this explicit randomized restart scheme rather than as direct convergence guarantees for the default warm-restart implementation.
	
	We first establish an open-ball reachability result. Supplementary Lemma~\ref{thm:open_ball_reach} shows that, under the refining-grid restart mechanism, BLOC reaches any prescribed neighborhood of a global minimizer with probability one after finitely many \emph{runs}. This result concerns global reachability and serves as the basis for the subsequent convergence result.
	\begin{supplemma}[Open-ball reachability under grid-supported refining restarts]
		\label{thm:open_ball_reach}
		Let $f : \mathcal{C} \to \mathbb{R}$ be continuous on a compact, convex $\mathcal{C}\subset\mathbb{R}^N$.
		Let $\bm{\nu}^* \in \arg\min_{\bm{\nu}\in\mathcal{C}} f(\bm{\nu})$, and fix any $\delta>0$. Denote
		$B_\delta(\bm{\nu}^*)=\{\bm{\nu}\in\mathcal{C}:\|\bm{\nu}-\bm{\nu}^*\|_2<\delta\}$.
		Assume BLOC is \emph{run} in multiple \emph{runs} $r=1,2,\dots$ as in the algorithmic description.
		\begin{enumerate}
			\item[(B1)] In \emph{run} $r$, the algorithm evaluates the $2N$
			candidates $\bm{\nu}\pm s^{(r)}_j \bm{e}_i$ ($i=1,\dots,N$) at iteration $j$ and accepts an update only if $f$ decreases.
			
			\item[(B2)] The in-\emph{run} step size is initialized at $s^{(r)}_0>0$,
			and is reduced geometrically by a factor $\rho>1$ \emph{only after an unsuccessful iteration}. Each \emph{run} enforces
			a positive floor $\kappa>0$ and terminates when $s^{(r)}_j\le \kappa$ or a prescribed budget is reached.
			
			\item[(B3)] After each \emph{run} $r$, the next \emph{run}
			starts from $\bm{\nu}^{(r+1)}_0$ drawn \emph{uniformly at random} from the finite grid
			\[
			G_r:=\mathcal{C}\cap\big(\alpha+s_r\mathbb{Z}^N\big),\qquad s_r\downarrow 0,
			\]
			where $\alpha\in\mathbb{R}^N$ is a fixed offset (lattice phase). Restart draws are independent across \emph{runs}.
		\end{enumerate}
		Then, with probability one, there exists a finite iteration index $T<\infty$ such that $\bm{\nu}^{(T)}\in B_\delta(\bm{\nu}^*)$.
	\end{supplemma}
	
	\begin{proof}
		To begin, we show that for sufficiently fine mesh size $s_r$ there must exist a grid point lying inside $B_\delta(\bm{\nu}^*)$. Fix $\delta>0$. Since $s_r\downarrow 0$, choose $r$ with $s_r<2\delta/\sqrt{N}$.
		For the cubic lattice $\alpha+s_r\mathbb{Z}^N$, the standard lattice covering property
		(see \citealp{Conway1988}, Ch.~1) ensures that for any $x\in\mathbb{R}^N$ there exists
		$g\in\alpha+s_r\mathbb{Z}^N$ such that
		\[
		\|x-g\|_\infty \le s_r/2.
		\]
		Applying this with $x=\bm{\nu}^*$ yields some $\bm{\nu}_{\mathrm{grid}}$ with
		\[
		\|\bm{\nu}_{\mathrm{grid}}-\bm{\nu}^*\|_2
		\;\le\; \sqrt{N}\,\|\bm{\nu}_{\mathrm{grid}}-\bm{\nu}^*\|_\infty
		\;\le\; \sqrt{N}\,s_r/2
		\;<\;\delta.
		\]
		Therefore $G_r\cap B_\delta(\bm{\nu}^*)\neq\emptyset$ for all sufficiently large $r$.
		
		Next, we show that the probability of a restart landing inside $B_\delta(\bm{\nu}^*)$ does not vanish as the grid refines. Under (B3), the probability that the restart of \emph{run} $r$ lands in the ball is
		\[
		p_r:=\mathbb{P}\!\big(\bm{\nu}^{(r+1)}_0\in B_\delta(\bm{\nu}^*)\big)
		=\frac{|\,G_r\cap B_\delta(\bm{\nu}^*)\,|}{|\,G_r\,|}.
		\]
		As $s_r\to 0$, these ratios converge to the volume fraction
		\[
		\frac{\operatorname{vol}(B_\delta(\bm{\nu}^*)\cap\mathcal{C})}{\operatorname{vol}(\mathcal{C})}
		\;=\; v \;>\; 0,
		\]
		by the standard Riemann-sum convergence for indicator functions of sets with measure-zero boundary
		(e.g., any real-analysis text; the boundary of a compact convex set has Lebesgue measure $0$).
		Hence, for any $\varepsilon\in(0,v)$ there exists $r_0$ such that $p_r\ge v-\varepsilon=:p_0>0$ for all $r\ge r_0$.
		
		Finally, to conclude the argument we invoke the Borel--Cantelli lemma. Restarts are independent across \emph{runs} by (B3), and occur infinitely often. Therefore,
		\[
		\sum_{r=1}^{\infty}\mathbb{P}\big(\bm{\nu}^{(r+1)}_0\in B_\delta(\bm{\nu}^*)\big)
		\;\ge\;\sum_{r\ge r_0} p_0 \;=\; \infty.
		\]
		By the second Borel--Cantelli lemma (see, e.g., \citealp{billingsley1995probability}),
		with probability one there exists a \emph{run} index $r$ such that $\bm{\nu}^{(r+1)}_0\in B_\delta(\bm{\nu}^*)$.
		Let $T$ be the iteration index at that restart; then $\bm{\nu}^{(T)}\in B_\delta(\bm{\nu}^*)$.
		\medskip
		
		Note that, that conclusion depends only on (B3). Assumptions (B1)–(B2)
		make the statement faithful to BLOC's in-\emph{run} mechanics (coordinate polling, geometric reduction only after failure, and a step-size floor $\kappa$), but they are not invoked in Steps 1–3.
	\end{proof}
	
	
	Supplementary Theorem~\ref{thm:global_conv_prob} strengthens the preceding reachability result under additional regularity conditions. In particular, when the objective has a unique strict global minimizer and is continuously differentiable and locally convex in a neighborhood of that minimizer, the randomized multi-start version of BLOC converges in probability to the global minimizer.
	
	\begin{supptheorem}[Global Convergence in Probability of BLOC]
		\label{thm:global_conv_prob}
		Let \( f : \mathcal{C} \to \mathbb{R} \) be a continuous function defined on a compact, convex set \( \mathcal{C} \subset \mathbb{R}^N \), and suppose \( f \) has a unique strict global minimizer \( \bm{\nu}^* \in \operatorname{int}(\mathcal{C}) \). 
		Assume that BLOC satisfies the same conditions (B1)–(B3) as in Lemma~\ref{thm:open_ball_reach}, and in addition:
		\begin{enumerate}
			\item[(B4)] The function \( f \) is continuously differentiable and locally convex in some neighborhood of \( \bm{\nu}^* \).
		\end{enumerate}
		Then the BLOC algorithm converges in probability to \( \bm{\nu}^* \). That is, for any \( \varepsilon > 0 \), \vspace{-0.2cm}
		\[
		\mathbb{P}\!\big(\| \bm{\nu}^{(r)} - \bm{\nu}^* \| > \varepsilon\big) \longrightarrow 0 \quad \text{as } r \to \infty.
		\]
	\end{supptheorem}
	
	\begin{proof}
		Let \( \bm{\nu}^{(r)} \) denote the output of \emph{run} \( r \), that is, the best point identified at the end of that \emph{run}. We first establish that the uniqueness and strictness of the minimizer imply a separation property. Since \( f \) is continuous on the compact domain \( \mathcal{C} \), it attains its minimum at \( \bm{\nu}^* \). Because the minimizer is strict, for every $\delta>0$ there exists $\eta>0$ such that
		\[
		f(\bm{\nu}) \;\ge\; f(\bm{\nu}^*) + \eta 
		\qquad \text{for all } \bm{\nu}\in \mathcal{C}\setminus B_\delta(\bm{\nu}^*).
		\]
		Thus, outside the ball \( B_\delta(\bm{\nu}^*) \), the function value is uniformly larger than at the minimizer. This guarantees that any point within $B_\delta(\bm{\nu}^*)$ is strictly preferable to all points outside.
		
		Next, we show that the algorithm will eventually enter such a neighborhood. By Lemma~\ref{thm:open_ball_reach}, which holds under assumptions (B1)–(B3), the restart mechanism guarantees that with probability one there exists some finite \emph{run} $r$ such that $\bm{\nu}^{(r)}_0\in B_\delta(\bm{\nu}^*)$. In other words, the algorithm almost surely restarts inside any prescribed ball around the minimizer, no matter how small.
		
		Once the algorithm enters this region, we argue that it cannot leave without converging. Assumption (B4) ensures that \( f \) is differentiable and convex in some neighborhood of \( \bm{\nu}^* \). Within this neighborhood, the coordinate polling and step-size schedule in (B1)–(B2) guarantee that the algorithm can only accept descent steps. If no further coordinate direction yields improvement, then by Theorem~3 the current iterate must satisfy \( \nabla f = 0 \). By local convexity, the only stationary point in this neighborhood is the minimizer itself, and thus the algorithm must converge to \( \bm{\nu}^* \).
		
		Finally, we combine these observations to establish convergence in probability. Fix $\varepsilon > 0$ and choose $\delta < \varepsilon$ such that $B_\delta(\bm{\nu}^*)$ lies entirely within the convex neighborhood where (B4) holds. By the restart argument, the probability that BLOC enters $B_\delta(\bm{\nu}^*)$ by \emph{run} $r$ approaches one as $r\to\infty$. Once inside, the descent mechanism ensures convergence to $\bm{\nu}^*$. Therefore
		\[
		\lim_{r\to\infty}\mathbb{P}\big(\|\bm{\nu}^{(r)} - \bm{\nu}^*\| < \varepsilon\big)=1,
		\]
		which proves convergence in probability.
	\end{proof}
	
	\paragraph{Remarks on scope and interpretation.}
	The randomized restart assumptions in Supplementary Lemma~\ref{thm:open_ball_reach}
	and Theorem~\ref{thm:global_conv_prob} provide an explicit mechanism for establishing
	global reachability and convergence in probability. This theoretical restart scheme differs
	from the default BLOC implementation, in which each successive \emph{run} is initialized
	from the best solution obtained in the preceding \emph{run} while resetting the step size.
	Accordingly, these results characterize a randomized exploratory variant of BLOC and should
	be viewed as complementary to the stationarity and convergence-rate results in the main
	manuscript, which concern the coordinate-polling behavior of BLOC within a \emph{run}.
	
	\begin{remark}
		\label{rem:openball}
		Although Supplementary Lemma~\ref{thm:open_ball_reach} establishes global reachability
		through repeated randomized restarts, BLOC can also explore the domain substantially within
		a single \emph{run}. As the step size decreases, the associated lattice of candidate points
		becomes increasingly fine; once $s_t<\delta/(2\sqrt{N})$, the grid
		$\bm{\nu}^{(0)}+s_t\mathbb{Z}^N$ can contain points within
		$B_\delta(\bm{\nu}^*)$ by the lattice covering property \citep{Conway1988}.
		This observation helps explain why, in practice, a modest number of \emph{runs}
		(typically 5--20) can be sufficient to obtain high-quality solutions, although it does not
		itself provide the probabilistic reachability guarantee of Supplementary
		Lemma~\ref{thm:open_ball_reach}.
	\end{remark}
	
	\begin{remark}
		\label{rem:patsearch}
		The reachability result in Supplementary Lemma~\ref{thm:open_ball_reach} is distinct from
		classical pattern-search results such as \cite{Torczon1997}, which establish local
		stationarity rather than global open-ball reachability under a randomized restart scheme.
	\end{remark}
	\subsection{Discussion of Assumptions (A1) -- (A5)}
	Assumption \textnormal{(A1)} is the analogue of the bounded-eigenvalue
	conditions used throughout \cite{lam2009sparsistency}. The lower eigenvalue bound
	ensures that \(\bm\Gamma_0\) remains uniformly separated from the boundary
	of the positive-semidefinite cone. The upper
	eigenvalue bound is not essential for the correlation-matrix rate in Theorem 1, but is useful for the subsequent covariance-matrix bound.
	
	Assumptions \textnormal{(A2)}--\textnormal{(A3)} replace the explicit Gaussian likelihood expansion used by \citet{lam2009sparsistency}. In their setting, the likelihood admits a second-order Taylor expansion whose quadratic term is uniformly positive in a shrinking neighborhood of the true matrix, while the first-order score term is controlled probabilistically. Because our framework allows a general locally Lipschitz loss that may be nondifferentiable, we formulate these two ingredients directly. Specifically, \textnormal{(A2)} imposes local quadratic growth around \(\bm\Gamma_0\), while \textnormal{(A3)} controls a corresponding Clarke subgradient at the high-dimensional rate \(q_n\). These conditions are satisfied by smooth locally strongly convex losses and by nonsmooth losses that retain such quadratic growth, for example through a strongly convex component. Their main limitation is that they do not cover arbitrary nonsmooth convex losses, such as purely piecewise-linear objectives, for which local quadratic curvature may fail. A broader theory could replace the quadratic-growth requirement in \textnormal{(A2)} by a more general local margin condition, with the resulting convergence rate determined by the corresponding margin behavior. Developing such an extension is left for future work.
	
	Assumption \textnormal{(A4)} is the generalized version of the
	singularity-at-zero condition used by \cite{lam2009sparsistency} for sparsity-inducing
	nonconvex penalties. Their condition requires the penalty to have a
	strictly positive effective slope near the origin. Here this is written
	in the uniform form $
	p_{\lambda_n}(t)
	\ge
	k_1\lambda_n t$ for sufficiently small \(t/\lambda_n\). The uniformity
	in \(n\) is made explicit because the tuning parameter
	\(\lambda_n\) varies with the sample size. Standard folded-concave
	penalties such as SCAD and MCP satisfy this condition.
	
	Assumption \textnormal{(A5)} corresponds to the conditions imposed by \cite{lam2009sparsistency}
	on the behavior of the penalty at the true nonzero entries. 
	%
	Specifically, the condition
	\[
	b_n(C)
	=
	\max_{(i,j)\in S}
	\sup_{\substack{
			u>0:\\
			|u-|\gamma_{0,ij}||\le Cr_n
	}}
	\left|
	p_{\lambda_n}''(u)
	\right|
	\to0
	\]
	makes the Taylor control of the active-coordinate penalty explicit.
	\cite{lam2009sparsistency} impose analogous smoothness conditions on the second
	derivative of the penalty away from zero. In our context, this
	condition is written directly over the shrinking neighborhoods that are
	actually traversed in the proof. For SCAD and MCP, once the true
	signals are sufficiently large relative to \(\lambda_n\), these
	derivatives vanish or become negligible on the relevant neighborhoods.
	\subsection{The additional condition in Theorem~2}
	The additional condition in Theorem~2 is needed because sparsistency requires control of the loss not only at the true parameter \(\bm\Gamma_0\), but also along local perturbations of the estimated matrix. Assumption~(A3) controls a Clarke subgradient of \(h_n\) at \(\bm\Gamma_0\), which is sufficient for the rate argument in Theorem~1. In contrast, the sparsistency proof compares \(\hat{\bm\Gamma}_n\) with a nearby matrix obtained by setting a candidate false-positive entry to zero. By Lebourg's mean value theorem, the corresponding change in \(h_n\) is governed by a Clarke subgradient evaluated at an intermediate point on the line segment joining these two matrices. Thus, control only at \(\bm\Gamma_0\) is not sufficient.
	
	The additional condition ensures that all Clarke subgradients remain uniformly small throughout the \(O(r_n)\) neighborhood containing the estimator and the comparison matrices. Specifically, it yields
	
	$$
	\sup_{\substack{
			\bm\Gamma\in\mathcal C_d\\
			\|\bm\Gamma-\bm\Gamma_0\|_F\le C_0r_n
	}}
	\sup_{g\in\partial h_n(\bm\Gamma)}
	\|g\|_\infty
	=
	O(q_n+r_n)
	=
	o(\lambda_n),
	$$
	
	where the last relation follows from \(q_n\le r_n\) and
	\(r_n/\lambda_n\to0\). Consequently, the possible increase in the loss from setting a false-positive entry to zero is of smaller order than the corresponding decrease in the penalty, which is at least of order
	$
	\lambda_n|\hat{\gamma}_{n,ij}|
	$ under Assumption~(A4). The resulting decrease in the penalized objective contradicts the minimizing property of \(\hat{\bm\Gamma}_n\) over the shrinking neighborhood used in Theorem~1, thereby establishing exact recovery of the zero off-diagonal entries.

	\subsection{Proofs of Theorems 1 and 2}
	\label{app:stat_proofs}
	We first introduce some notations that will be used throughout the proofs of theorems 1 and 2. 
	
	Recall $\mathcal C_d
	=
	\left\{
	\bm\Gamma\in\mathcal D_d:
	\bm\Gamma\succ0
	\right\}$, and consider the penalized criterion
	\[
	Q_n(\bm\Gamma)
	=
	h_n(\bm\Gamma)
	+
	P_{\lambda_n}(\bm\Gamma),
	~
	P_{\lambda_n}(\bm\Gamma)
	=
	\sum_{i\neq j}
	p_{\lambda_n}(|\gamma_{ij}|),
	~\bm\Gamma\in\mathcal C_d.
	\]
	For \(\Delta\in\mathcal T_d\), define
	$
	\Delta_S
	=
	\left(
	\Delta_{ij}\mathbf 1\{(i,j)\in S\}
	\right)_{1\le i,j\le d}$, and analogously define \(\Delta_{S^c}\). 
	In particular,
	\begin{equation}
		\label{eq:A.1}
		\|\Delta_S\|_1
		\le
		\sqrt{s}\,\|\Delta_S\|_F
		\le
		\sqrt{s}\,\|\Delta\|_F.
	\end{equation}
	\subsubsection{Proof of Theorem 1}
	Fix an arbitrary constant \(k_1\in(0,k)\). For a fixed constant
	\(C>0\), define
	\[
	\mathcal S_n(C)
	=
	\left\{
	\Delta\in\mathcal T_d:
	\|\Delta\|_F=Cr_n
	\right\}.
	\]
	We first show that \(C\) can be chosen sufficiently large, but fixed
	independently of \(n\), such that
	\begin{equation}
		\label{eq:A.16}
		\inf_{\Delta\in\mathcal S_n(C)}
		\left\{
		Q_n(\bm\Gamma_0+\Delta)-Q_n(\bm\Gamma_0)
		\right\}
		>0,~\text{with probability tending to 1.}
	\end{equation}
	By Assumption \textnormal{(A1)}, $\lambda_{\min}(\bm\Gamma_0)\ge c_\Gamma>0$ and from the statement of Theorem 1, $
	r_n\to0$. Hence, for every fixed \(C>0\),
	\begin{equation}
		\label{eq:A.17}
		Cr_n
		<
		\min\left\{
		\frac{c_\Gamma}{2},
		\delta_h
		\right\}
	\end{equation}
	for all sufficiently large \(n\), where \(\delta_h>0\) is the neighborhood radius in Assumption~(A2). Also, for any
	\(\Delta\in\mathcal S_n(C)\), $
	\|\Delta\|
	\le
	\|\Delta\|_F
	=
	Cr_n$. Weyl's inequality therefore gives
	\[
	\lambda_{\min}(\bm\Gamma_0+\Delta)
	\ge
	\lambda_{\min}(\bm\Gamma_0)-\|\Delta\|
	\ge
	c_\Gamma-Cr_n
	\ge
	\frac{c_\Gamma}{2}>0.
	\]
	Since \(\Delta\in\mathcal T_d\), $\operatorname{diag}(\bm\Gamma_0+\Delta)=\mathbf 1$. Thus
	$\bm\Gamma_0+\Delta\in\mathcal C_d$ for all sufficiently large \(n\), uniformly over
	\(\Delta\in\mathcal S_n(C)\). For \(\Delta\in\mathcal S_n(C)\), decompose 
	\[
	Q_n(\bm\Gamma_0+\Delta)-Q_n(\bm\Gamma_0)
	=
	I_1(\Delta)+I_2(\Delta)+I_3(\Delta),
	\]
	where
	\[
	I_1(\Delta)
	=
	h_n(\bm\Gamma_0+\Delta)-h_n(\bm\Gamma_0),
	\]
	\[
	I_2(\Delta)
	=
	\sum_{(i,j)\in S^c}
	p_{\lambda_n}(|\Delta_{ij}|),
	\]
	and
	\[
	I_3(\Delta)
	=
	\sum_{(i,j)\in S}
	\left[
	p_{\lambda_n}
	\left(
	|\gamma_{0,ij}+\Delta_{ij}|
	\right)
	-
	p_{\lambda_n}(|\gamma_{0,ij}|)
	\right].
	\]
	By Assumption \textnormal{(A2)}, on an event whose probability tends to one,
	uniformly over \(\Delta\in\mathcal S_n(C)\),
	\begin{equation}
		\label{eq:A.20}
		I_1(\Delta)
		\ge
		\langle g_{0,n},\Delta\rangle
		+
		\left\{
		\frac{\kappa}{2}-\eta_n(C)
		\right\}
		\|\Delta\|_F^2.
	\end{equation}
	Decompose the linear term as $
	\langle g_{0,n},\Delta\rangle
	=
	\langle(g_{0,n})_S,\Delta_S\rangle
	+
	\langle(g_{0,n})_{S^c},\Delta_{S^c}\rangle$ and define $E_{g,n}
	=
	\left\{
	\|g_{0,n}\|_\infty\le C_gq_n
	\right\}$. By Assumption \textnormal{(A3)}, the event $E_{g,n}$ has probability tending to one. On \(E_{g,n}\), using Equation \eqref{eq:A.1},
	\begin{equation}
		\label{eq:A.22}
		\begin{aligned}
			\langle(g_{0,n})_S,\Delta_S\rangle
			&\ge
			-
			\|g_{0,n}\|_\infty
			\|\Delta_S\|_1
			\\
			&\ge
			-
			C_gq_n\sqrt{s}\,\|\Delta\|_F
			\\
			&=
			-C_gCr_n^2,
		\end{aligned}
	\end{equation}
	because $r_n=\sqrt{s}\,q_n$ and $\|\Delta\|_F=Cr_n$.
	
	We next treat the coordinates in \(S^c\). Recall that ${Cr_n}/{\lambda_n}\longrightarrow0$ as $n\to\infty$. Let \(\varepsilon_0>0\) be the constant supplied by
	Assumption \textnormal{(A4)} for the fixed \(k_1\). Then, for all sufficiently large
	\(n\), $
	Cr_n\le\varepsilon_0\lambda_n$. Furthermore, since $
	|\Delta_{ij}|
	\le
	\|\Delta\|_F
	=
	Cr_n$, Assumption \textnormal{(A4)} gives, uniformly over
	\(\Delta\in\mathcal S_n(C)\),
	\[
	I_2(\Delta)
	\ge
	k_1\lambda_n
	\|\Delta_{S^c}\|_1.
	\]
	Moreover, $\langle(g_{0,n})_{S^c},\Delta_{S^c}\rangle
	\ge
	-
	\|g_{0,n}\|_\infty
	\|\Delta_{S^c}\|_1$. Hence
	\[
	I_2(\Delta)
	+
	\langle(g_{0,n})_{S^c},\Delta_{S^c}\rangle
	\ge
	\left\{
	k_1\lambda_n-\|g_{0,n}\|_\infty
	\right\}
	\|\Delta_{S^c}\|_1.
	\]
	Since \(s\ge1\), $q_n\le r_n$. Thus, ${q_n}/{\lambda_n}
	\longrightarrow0$ as $n\to\infty$. Now, on \(E_{g,n}\),
	\[
	\frac{\|g_{0,n}\|_\infty}{\lambda_n}
	\le
	C_g\frac{q_n}{\lambda_n}
	\longrightarrow0.
	\]
	Consequently, for all sufficiently large \(n\), $k_1\lambda_n-\|g_{0,n}\|_\infty>0$
	on \(E_{g,n}\), and therefore
	\begin{equation}
		\label{eq:A.25}
		I_2(\Delta)
		+
		\langle(g_{0,n})_{S^c},\Delta_{S^c}\rangle
		\ge0
	\end{equation}
	uniformly over \(\Delta\in\mathcal S_n(C)\), on an event whose
	probability tends to one.
	
	It remains to control \(I_3(\Delta)\). Let
	$\theta_{ij}
	=
	|\gamma_{0,ij}|,
	~
	(i,j)\in S$, and define
	$
	\theta_{\min,n}
	=
	\min_{(i,j)\in S}\theta_{ij}$. By Assumption \textnormal{(A5)},
	${\theta_{\min,n}}/{\lambda_n}
	\longrightarrow\infty$ whereas ${Cr_n}/{\lambda_n}
	\longrightarrow0$ as $n\to\infty$. Therefore
	\[
	\frac{Cr_n}{\theta_{\min,n}}
	=
	\frac{Cr_n/\lambda_n}
	{\theta_{\min,n}/\lambda_n}
	\longrightarrow0.
	\]
	In particular, for all sufficiently large \(n\),
	$Cr_n<\theta_{\min,n}$. Thus the intervals defining \(\mathcal J_n(C)\) are contained in
	\((0,\infty)\). Moreover, for every
	\((i,j)\in S\) and every
	\(\Delta\in\mathcal S_n(C)\),
	\[
	|\Delta_{ij}|
	\le Cr_n
	<
	|\gamma_{0,ij}|,
	\]
	and hence $
	\operatorname{sgn}
	(\gamma_{0,ij}+\Delta_{ij})
	=
	\operatorname{sgn}(\gamma_{0,ij})$. By the twice continuous differentiability of \(p_{\lambda_n}\) on $\mathcal J_n(C)$, it follows that the function
	$x\longmapsto p_{\lambda_n}(|x|)
	$ is twice continuously differentiable on the entire line segment
	joining \(\gamma_{0,ij}\) and
	\(\gamma_{0,ij}+\Delta_{ij}\). Taylor's theorem therefore yields
	\begin{equation}
		\label{eq:A.29}
		\begin{aligned}
			&
			p_{\lambda_n}
			\left(
			|\gamma_{0,ij}+\Delta_{ij}|
			\right)
			-
			p_{\lambda_n}(\theta_{ij})
			\\
			&\qquad =
			p_{\lambda_n}'(\theta_{ij})
			\operatorname{sgn}(\gamma_{0,ij})
			\Delta_{ij}
			+
			\frac12
			p_{\lambda_n}''(\widetilde\theta_{ij})
			\Delta_{ij}^2
		\end{aligned}
	\end{equation}
	for some \(\widetilde\theta_{ij}>0\) satisfying
	$|\widetilde\theta_{ij}-\theta_{ij}|
	\le
	|\Delta_{ij}|
	\le
	Cr_n$. By Assumption \textnormal{(A5)}, $
	\left|
	p_{\lambda_n}'(\theta_{ij})
	\right|
	\le
	a_n
	\le
	C_aq_n$, and $
	\left|
	p_{\lambda_n}''(\widetilde\theta_{ij})
	\right|
	\le
	b_n(C)$. Summing Equation \eqref{eq:A.29} over \(S\), we obtain
	\begin{equation}
		\label{eq:A.31}
		\begin{aligned}
			I_3(\Delta)
			&\ge
			-a_n\|\Delta_S\|_1
			-
			\frac12b_n(C)\|\Delta_S\|_F^2
			\\
			&\ge
			-C_aq_n\sqrt{s}\,\|\Delta\|_F
			-
			\frac12b_n(C)\|\Delta\|_F^2
			\\
			&=
			-C_aCr_n^2
			-
			\frac12b_n(C)C^2r_n^2.
		\end{aligned}
	\end{equation}
	Combining equations \eqref{eq:A.20}, \eqref{eq:A.22}, \eqref{eq:A.25}, and \eqref{eq:A.31}, we
	obtain, uniformly over \(\Delta\in\mathcal S_n(C)\),
	\begin{equation}
		\label{eq:A.32}
		\begin{aligned}
			Q_n(\bm\Gamma_0+\Delta)-Q_n(\bm\Gamma_0)
			\ge{}&
			\left\{
			\frac{\kappa}{2}
			-
			\eta_n(C)
			-
			\frac12b_n(C)
			\right\}
			C^2r_n^2
			\\
			&-
			(C_g+C_a)Cr_n^2
		\end{aligned}
	\end{equation}
	on an event whose probability tends to one.
	
	Choose a fixed deterministic constant \(C\) satisfying
	\begin{equation}
		\label{eq:A.33}
		C>
		\frac{4(C_g+C_a)}{\kappa}.
	\end{equation}
	For this fixed \(C\), $\eta_n(C)=o_P(1)$ and, by Assumption \textnormal{(A5)}, $b_n(C)\longrightarrow0$. Hence
	\[\eta_n(C)+\frac12b_n(C)
	\le
	\frac{\kappa}{4}, \text{with probability tending to 1.}
	\]
	On the event $\{\eta_n(C)+\frac12b_n(C)
	\le
	\frac{\kappa}{4}\}$, Equation \eqref{eq:A.32} gives
	\[
	\begin{aligned}
		Q_n(\bm\Gamma_0+\Delta)-Q_n(\bm\Gamma_0)
		&\ge
		r_n^2
		\left\{
		\frac{\kappa}{4}C^2
		-
		(C_g+C_a)C
		\right\}
		\\
		&>0
	\end{aligned}
	\]
	uniformly over \(\Delta\in\mathcal S_n(C)\), where the final strict
	inequality follows from Equation \eqref{eq:A.33}. This proves Equation \eqref{eq:A.16}.
	
	Now define the closed local ball
	$\mathcal B_n(C)
	=
	\left\{
	\Gamma_0+\Delta:
	\Delta\in\mathcal T_d,\ 
	\|\Delta\|_F\le Cr_n
	\right\}$. The preceding Weyl argument shows that, for all sufficiently large
	\(n\),
	$\mathcal B_n(C)\subset\mathcal C_d$. Moreover, by Equation \eqref{eq:A.17}, $\mathcal B_n(C)\subset\mathcal U_h=\{\bm\Gamma\in\mathcal D_d:\|\bm\Gamma-\bm\Gamma_0\|_F<\delta_h\}$ for all sufficiently large \(n\). Hence, on the local-Lipschitz event
	in Assumption\textnormal{(A2)}, \(h_n\) is continuous on
	\(\mathcal B_n(C)\). Since \(p_{\lambda_n}\) is continuous,
	\(Q_n\) is continuous on \(\mathcal B_n(C)\).
	
	The set \(\mathcal B_n(C)\) is compact, being a closed and bounded
	subset of the finite-dimensional affine space \(\mathcal D_d\).
	Consequently, \(Q_n\) attains its minimum on
	\(\mathcal B_n(C)\). Let
	\begin{equation}
		\label{eq:A.37}
		\tilde{\bm\Gamma}_n
		\in
		\arg\min_{\bm\Gamma\in\mathcal B_n(C)}
		Q_n(\bm\Gamma).
	\end{equation}
	Since
	$\bm\Gamma_0\in\mathcal B_n(C)$, we have
	\begin{equation}
		\label{eq:A.38}
		Q_n(\tilde{\bm\Gamma}_n)
		\le
		Q_n(\bm\Gamma_0).
	\end{equation}
	On the event in Equation \eqref{eq:A.16}, every point on the boundary of
	\(\mathcal B_n(C)\) has objective value strictly larger than
	\(Q_n(\bm\Gamma_0)\). Thus, Equation \eqref{eq:A.38} implies
	\begin{equation}
		\label{eq:A.39}
		\|\tilde{\bm\Gamma}_n-\bm\Gamma_0\|_F<Cr_n.
	\end{equation}
	Define $\rho_n
	=
	Cr_n-
	\|\tilde{\bm\Gamma}_n-\bm\Gamma_0\|_F
	>0$. If $\bm\Gamma\in\mathcal C_d$
	and  $\|\bm\Gamma-\tilde{\bm\Gamma}_n\|_F<\rho_n$, then, by the triangle inequality,
	$
	\|\bm\Gamma-\bm\Gamma_0\|_F
	<
	Cr_n$, and hence
	$\bm\Gamma\in\mathcal B_n(C)$. By the minimizing property in Equation \eqref{eq:A.37}, $Q_n(\tilde{\bm\Gamma}_n)
	\le
	Q_n(\bm\Gamma)$
	for every such \(\bm\Gamma\). Therefore
	\(\tilde{\bm\Gamma}_n\) is a local minimizer of \(Q_n\) over
	\(\mathcal C_d\).
	
	On the intersection of the event in Equation~\eqref{eq:A.16} and the local-Lipschitz event in Assumption~\textnormal{(A2)}, which has probability tending to one, the minimizer \(\tilde{\bm\Gamma}_n\) over \(\mathcal B_n(C)\) exists by compactness and continuity, lies in the interior of \(\mathcal B_n(C)\) by Equation~\eqref{eq:A.39}, and is therefore a local minimizer of \(Q_n\) over \(\mathcal C_d\). Hence, with probability tending to one, there exists a local minimizer \(\hat{\bm\Gamma}_n\) of \(Q_n\) such that
	$$
	\|\hat{\bm\Gamma}_n-\bm\Gamma_0\|_F
	\le Cr_n.
	$$
	Since $r_n=\sqrt{s\log d/n}$, this proves the first statement of Theorem 1.
	
	For the covariance-matrix conclusion, write
	\begin{equation}
		\label{eq:A.40}
		\begin{aligned}
			\hat{\bm\Sigma}_n-\bm\Sigma_0
			={}&
			\hat{\bm W}\hat{\bm\Gamma}_n\hat{\bm W}
			-
			\bm W_0\bm\Gamma_0\bm W_0
			\\
			={}&
			\hat{\bm W}
			(\hat{\bm\Gamma}_n-\bm\Gamma_0)
			\hat{\bm W}
			+
			(\hat{\bm W}-\bm W_0)\bm\Gamma_0\hat{\bm W}
			+
			\bm W_0\bm \Gamma_0(\hat{\bm W}-\bm W_0).
		\end{aligned}
	\end{equation}
	By Assumption \textnormal{(A1)},
	$\|\bm\Gamma_0\|
	\le
	C_\Gamma$, and since $\|\bm W_0\|=O(1),~
	\|\hat{\bm W}-\bm W_0\|
	=
	O_P(q_n)$, we have
	\[
	\|\hat{\bm W}\|
	\le
	\|\bm W_0\|
	+
	\|\hat{\bm W}-\bm W_0\|
	=
	O_P(1).
	\]
	Also,
	$\|\hat{\bm\Gamma}_n-\bm\Gamma_0\|
	\le
	\|\hat{\bm\Gamma}_n-\bm\Gamma_0\|_F
	=
	O_P(r_n)$. Using submultiplicativity of the operator norm in Equation \eqref{eq:A.40},
	\[
	\begin{aligned}
		\|\hat{\bm\Sigma}_n-\bm\Sigma_0\|
		\le{}&
		\|\hat{\bm W}\|^2
		\|\hat{\bm\Gamma}_n-\bm\Gamma_0\|
		\\
		&+
		\|\hat{\bm W}-\bm W_0\|
		\|\bm\Gamma_0\|
		\|\hat{\bm W}\|
		\\
		&+
		\|\bm W_0\|
		\|\bm\Gamma_0\|
		\|\hat{\bm W}-\bm W_0\|.
	\end{aligned}
	\]
	Consequently,
	$\|\hat{\bm\Sigma}_n-\bm\Sigma_0\|
	=
	O_P(r_n)+O_P(q_n)$. Since
	$r_n=\sqrt{s}\,q_n$, 
	we have
	$r_n+q_n
	=
	(\sqrt{s}+1)q_n
	\le
	\sqrt{2(s+1)}\,q_n$. It follows that
	\[
	\|\hat{\bm\Sigma}_n-\bm\Sigma_0\|
	=
	O_P\left(
	\sqrt{\frac{(s+1)\log d}{n}}
	\right).
	\]
	Squaring yields the desired result.
	\hspace{\fill} $\blacksquare$
	\subsubsection{Proof of Theorem 2}
	Let \(C_0<\infty\) be the fixed constant used in the
	shrinking-ball construction in the proof of Theorem~1, and define
	$\mathcal B_n(C_0)
	=
	\left\{
	\bm\Gamma_0+\Delta:
	\Delta=\Delta^\top,\ 
	\operatorname{diag}(\Delta)=0,\ 
	\|\Delta\|_F\le C_0r_n
	\right\}$. The proof of Theorem~1 constructs
	\[
	\hat{\bm\Gamma}_n
	\in
	\arg\min_{\bm\Gamma\in\mathcal B_n(C_0)}
	Q_n(\bm\Gamma),
	\]
	and shows that for sufficiently large $n$,
	\(\mathcal B_n(C_0)\subset\mathcal C_d\), and with probability tending to one, \(\hat{\bm\Gamma}_n\) lies in the interior of
	\(\mathcal B_n(C_0)\), and hence that \(\hat{\bm\Gamma}_n\) is a
	local minimizer of \(Q_n\). In particular,
	$\|\hat{\bm\Gamma}_n-\bm\Gamma_0\|_F
	=
	O_P(r_n)$.
	
	Fix \(k_1\in(0,k)\), where \(k>0\) is the constant in the
	near-origin penalty condition in Assumption \textnormal{(A4)}.
	By that assumption, there exists \(\varepsilon_0>0\), independent of
	\(n\), such that, for all sufficiently large \(n\),
	\begin{equation}
		\label{eq:B.4}
		p_{\lambda_n}(t)
		\ge
		k_1\lambda_n t,
		\qquad
		0\le t\le\varepsilon_0\lambda_n.
	\end{equation}
	Since ${r_n}/{\lambda_n}\to0$ and \(s\ge1\), we have $q_n\le r_n$ and therefore ${q_n}/{\lambda_n}\to0$. Consequently,
	\[
	\frac{C_gq_n+C_hr_n}{\lambda_n}\to0.
	\]
	Hence, for all sufficiently large \(n\),
	\begin{equation}
		\label{eq:B.8}
		C_gq_n+C_hr_n
		<
		\frac{k_1}{2}\lambda_n.
	\end{equation}
	Let \(E_n\) be the intersection of the following events:
	\[
	\mathcal B_n(C_0)\subset\mathcal C_d,
	\]
	\begin{equation}
		\label{eq:B.10}
		\hat{\bm\Gamma}_n
		\in
		\arg\min_{\bm\Gamma\in\mathcal B_n(C_0)}
		Q_n(\bm\Gamma),
	\end{equation}
	the local Lipschitz property of \(h_n\) in Assumption
	\textnormal{(A2)}, and
	\begin{equation}
		\label{eq:B.11}
		\sup_{\substack{
				\bm\Gamma\in\mathcal C_d\\
				\|\bm\Gamma-\bm\Gamma_0\|_F\le C_0r_n
		}}
		\sup_{g\in\partial h_n(\bm\Gamma)}
		\|g\|_\infty
		\le
		C_gq_n+C_hr_n.
	\end{equation}
	By Theorem~1, Assumption \textnormal{(A2)}, and the statement of Theorem 2, the event $E_n$ has probability tending to 1. Furthermore, because \(r_n\to0\), for all sufficiently large \(n\),
	\[
	C_0r_n<\frac{\delta_h}{2},
	\]
	where \(\delta_h\) is the radius of the locally Lipschitz
	neighborhood in Assumption \textnormal{(A2)}.
	
	Work on \(E_n\). Suppose that there exists $(i,j)\in S^c$ such that
	\begin{equation}
		\label{eq:B.14}
		\widehat\gamma_{n,ij}\neq0.
	\end{equation}
	Since $\gamma_{0,ij}=0$, and
	$\hat{\bm\Gamma}_n\in\mathcal B_n(C_0)$, we have
	\begin{equation}
		\label{eq:B.15}
		|\widehat\gamma_{n,ij}|
		=
		|\widehat\gamma_{n,ij}-\gamma_{0,ij}|
		\le
		\|\hat{\bm\Gamma}_n-\bm\Gamma_0\|_F
		\le
		C_0r_n.
	\end{equation}
	Since $r_n/\lambda_n\to0$ as $n\to\infty$, for all sufficiently large \(n\),
	\begin{equation}
		\label{eq:B.16}
		C_0r_n
		\le
		\varepsilon_0\lambda_n.
	\end{equation}
	Combining equations \eqref{eq:B.4}, \eqref{eq:B.15}, and \eqref{eq:B.16},
	\begin{equation}
		\label{eq:B.17}
		p_{\lambda_n}(|\widehat\gamma_{n,ij}|)
		\ge
		k_1\lambda_n|\widehat\gamma_{n,ij}|.
	\end{equation}
	Define \(\tilde{\bm\Gamma}_n^{(ij)}\) by replacing the symmetric pair
	\((\widehat\gamma_{n,ij},\widehat\gamma_{n,ji})\) with zero and
	leaving all other entries unchanged. Thus, $\widetilde\gamma_{n,ij}^{(ij)}
	=
	\widetilde\gamma_{n,ji}^{(ij)}
	=
	0$, while
	$
	\widetilde\gamma_{n,kl}^{(ij)}
	=
	\widehat\gamma_{n,kl}$ for all other entries.
	
	Because $\gamma_{0,ij}=\gamma_{0,ji}=0$, we have
	\[
	\begin{aligned}
		\|
		\tilde{\bm\Gamma}_n^{(ij)}-\bm\Gamma_0
		\|_F^2
		&=
		\|
		\hat{\bm\Gamma}_n-\bm\Gamma_0
		\|_F^2
		-
		2|\widehat\gamma_{n,ij}|^2
		\\
		&\le
		\|
		\hat{\bm\Gamma}_n-\bm\Gamma_0
		\|_F^2
		\\
		&\le
		C_0^2r_n^2.
	\end{aligned}
	\]
	Therefore
	\begin{equation}
		\label{eq:B.19}
		\tilde{\bm\Gamma}_n^{(ij)}
		\in
		\mathcal B_n(C_0).
	\end{equation}
	On \(E_n\), \(\mathcal B_n(C_0)\subset\mathcal C_d\), so the
	comparator is feasible.
	
	Consider the line segment
	\[
	\bm\Gamma_n(t)
	=
	(1-t)\hat{\bm\Gamma}_n
	+
	t\tilde{\bm\Gamma}_n^{(ij)},
	\qquad
	0\le t\le1.
	\]
	Since \(\mathcal B_n(C_0)\) is convex,
	$\bm\Gamma_n(t)\in\mathcal B_n(C_0)
	$ for every \(t\in[0,1]\). Hence, on \(E_n\),
	$\bm\Gamma_n(t)\in\mathcal C_d$ and
	\[
	\|\bm\Gamma_n(t)-\bm\Gamma_0\|_F
	\le
	C_0r_n
	<
	\frac{\delta_h}{2}.
	\]
	Thus, the entire segment lies strictly inside the neighborhood on which
	\(h_n\) is locally Lipschitz.
	
	By Lebourg's mean value theorem applied relative to the affine space \(\mathcal D_d\), there exist
	$t_n^\ast\in(0,1)$ and $g_n^\ast
	\in
	\partial h_n(\bm\Gamma_n(t_n^\ast))
	$
	such that
	\[
	h_n(\tilde{\bm\Gamma}_n^{(ij)})
	-
	h_n(\hat{\bm\Gamma}_n)
	=
	\left\langle
	g_n^\ast,
	\tilde{\bm\Gamma}_n^{(ij)}
	-
	\hat{\bm\Gamma}_n
	\right\rangle.
	\]
	The difference matrix has only the two symmetric off-diagonal
	entries $-\widehat\gamma_{n,ij}~\text{and}~
	-\widehat\gamma_{n,ji}
	=
	-\widehat\gamma_{n,ij}$ different from zero. Therefore
	\begin{equation}
		\label{eq:B.23}
		\left\|
		\widetilde\Gamma_n^{(ij)}
		-
		\widehat\Gamma_n
		\right\|_{\ell_1}
		=
		2|\widehat\gamma_{n,ij}|.
	\end{equation}
	Using Hölder's inequality, equations \eqref{eq:B.11}, and \eqref{eq:B.23},
	\begin{equation}
		\label{eq:B.24}
		\begin{aligned}
			h_n(\tilde{\bm\Gamma}_n^{(ij)})
			-
			h_n(\hat{\bm\Gamma}_n)
			&\le
			\|g_n^\ast\|_\infty
			\left\|
			\tilde{\bm\Gamma}_n^{(ij)}
			-
			\hat{\bm\Gamma}_n
			\right\|_{\ell_1}
			\\
			&\le
			2(C_gq_n+C_hr_n)
			|\widehat\gamma_{n,ij}|.
		\end{aligned}
	\end{equation}
	The penalty difference is
	\[
	P_{\lambda_n}(\tilde{\bm\Gamma}_n^{(ij)})
	-
	P_{\lambda_n}(\hat{\bm\Gamma}_n)
	=
	-2p_{\lambda_n}(|\widehat\gamma_{n,ij}|),
	\]
	because the penalty is summed over the full off-diagonal index set.
	Using Equation \eqref{eq:B.17},
	\begin{equation}
		\label{eq:B.26}
		P_{\lambda_n}(\tilde{\bm\Gamma}_n^{(ij)})
		-
		P_{\lambda_n}(\hat{\bm\Gamma}_n)
		\le
		-2k_1\lambda_n
		|\widehat\gamma_{n,ij}|.
	\end{equation}
	Combining equations \eqref{eq:B.24} and \eqref{eq:B.26},
	\[
	\begin{aligned}
		&
		Q_n(\tilde{\bm\Gamma}_n^{(ij)})
		-
		Q_n(\hat{\bm\Gamma}_n)
		\\
		&\qquad\le
		2
		\left\{
		C_gq_n+C_hr_n-k_1\lambda_n
		\right\}
		|\widehat\gamma_{n,ij}|.
	\end{aligned}
	\]
	By Equation \eqref{eq:B.8},
	\[
	C_gq_n+C_hr_n-k_1\lambda_n
	<
	-\frac{k_1}{2}\lambda_n,
	\]
	and therefore, under Equation \eqref{eq:B.14},
	\begin{equation}
		\label{eq:B.28}
		Q_n(\tilde{\bm\Gamma}_n^{(ij)})
		-
		Q_n(\hat{\bm\Gamma}_n)
		<
		-
		k_1\lambda_n
		|\widehat\gamma_{n,ij}|
		<
		0.
	\end{equation}
	On the other hand, equations \eqref{eq:B.10} and \eqref{eq:B.19} imply
	\[
	Q_n(\hat{\bm\Gamma}_n)
	\le
	Q_n(\tilde{\bm\Gamma}_n^{(ij)}),
	\]
	contradicting equation \eqref{eq:B.28}. Hence
	$\widehat\gamma_{n,ij}=0$ for every \((i,j)\in S^c\) on \(E_n\).
	
	The argument is uniform over all \((i,j)\in S^c\), since both the
	Clarke subdifferential bound and the penalty bound are uniform over
	the relevant coordinates. Since $E_n$ holds with probability tending to 1, $
	\widehat\gamma_{n,ij}=0$ for every $(i,j)\in S^c$ with probability tending to 1.
	
	This proves the result.\hspace{\fill} $\blacksquare$
	\bibliography{Bibliography-MM-MC}
\end{document}